\documentclass[a4paper,UKenglish,cleveref, autoref, thm-restate,final]{lipics-v2021}

\usepackage{hyperref}

\usepackage[utf8]{inputenc}
\usepackage{amssymb}
\usepackage{amsmath}
\usepackage[shortlabels]{enumitem}
\usepackage{xspace}
\usepackage{scalerel}
\usepackage{stmaryrd}
\usepackage{booktabs}
\usepackage{multirow}
\usepackage{pifont}
\usepackage{ifdraft}
\usepackage{wrapfig}

\usepackage{tikz}
\usetikzlibrary{calc}
\usetikzlibrary{positioning}
\usetikzlibrary{arrows,automata}

\RequirePackage{centernot}
\RequirePackage[Symbol]{upgreek}

\RequirePackage{etoolbox}
\RequirePackage{mathtools}

\let\originalleft\left
\let\originalright\right
\renewcommand{\left}{\mathopen{}\mathclose\bgroup\originalleft}
\renewcommand{\right}{\aftergroup\egroup\originalright}

\newcommand{\Nat}{\mathbb{N}}

\newcommand{\Procs}{\ensuremath{\mathcal{P}}}
\definecolor{roleColor}{rgb}{0.1, 0.3, 0.1}\newcommand{\roleCol}[1]{{\color{roleColor}#1}}\newcommand{\roleFmt}[1]{\boldsymbol{\roleCol{\mathtt{#1}}}}

\newcommand{\procA}{{\color{roleColor}\roleFmt{p}}}
\newcommand{\procB}{{\color{roleColor}\roleFmt{q}}}
\newcommand{\procC}{{\color{roleColor}\roleFmt{r}}}
\newcommand{\procD}{{\color{roleColor}\roleFmt{s}}}
\newcommand{\procE}{{\color{roleColor}\roleFmt{t}}}

\newcommand{\val}{\ensuremath{m}}

\newcommand{\Alphabet}{\Sigma}

\newcommand{\MsgVals}{\ensuremath{\mathcal{V}}}

\newcommand{\CSM}[1]{\ensuremath{\{\!\!\{#1_\procA\}\!\!\}_{\procA \in \Procs}}}
\newcommand{\CSMl}[1]{\ensuremath{\{\!\!\{{#1}\}\!\!\}_{\procA \in \Procs}}}

\newcommand{\emptystring}{\varepsilon}

\newcommand{\set}[1]{\{#1\}}
\newcommand{\lang}{\mathcal{L}}

\newcommand{\interswaplang}{\mathcal{C}^{\interswap}}

\newcommand{\channels}{\ensuremath{\mathsf{Chan}}}
\newcommand{\channel}[2]{\ensuremath{\langle#1,#2\rangle}}

\newcommand{\trace}{\operatorname{trace}}

\newcommand{\GG}{\mathbf{G}}
\newcommand{\getMu}{\mathit{get\mu}}
\newcommand{\getMuG}{\getMu_\GG}
\newcommand{\apair}[2]{\langle #1, #2 \rangle}
\newcommand{\epair}[2]{[ #1, #2 ]}

\newcommand{\final}{\ensuremath{\ast}}
\newcommand{\nonfinal}{\ensuremath{\circ}}
\newcommand{\optionfn}{\ensuremath{\circledast}}
\newcommand{\semglobal}{\ensuremath{\mathsf{GAut}}}
\newcommand{\semlocal}{\ensuremath{\mathsf{LAut}}}
\newcommand{\semextlocal}{\ensuremath{\mathsf{EAut}}}

\newcommand{\interswap}{\ensuremath{\sim}}

\DeclareMathOperator*{\ExtCh}{\&}
\DeclareMathOperator*{\IntCh}{⊕}
\DeclareMathOperator*{\Merge}{\merge}
\DeclareMathOperator*{\Mmerge}{\mmerge}

\def \ifempty#1{\def\temp{#1} \ifx\temp\empty }

\newcommand{\snd}[3]{\ifempty{#1} #2!#3 \else #1\triangleright#2!#3 \fi}
\newcommand{\rcv}[3]{\ifempty{#2} #1?#3 \else #2\triangleleft#1?#3 \fi}
\newcommand{\msgFromTo}[3]{#1\!\to\!#2\!:\!#3}

\newcommand{\pref}{\operatorname{pref}}
\newcommand{\preforder}{\ensuremath{\leq}}
\newcommand{\avail}{\operatorname{avail}}

\newcommand{\channelcompliant}{channel-compliant\xspace}
\newcommand{\channelcompliancy}{channel-compliancy\xspace}
\newcommand{\Channelcompliant}{Channel-compliant\xspace}
\newcommand{\projectable}{projectable\xspace}

\newcommand{\Projectable}{Projectable\xspace}

\newcommand{\tproj}{{\ensuremath{\upharpoonright}}}
\newcommand{\wproj}{{\ensuremath{\Downarrow}}}

\renewcommand{\merge}{\sqcap}
\def\mmerge{\mathrel{\ThisStyle{\stretchrel*{\ooalign{\raise0.2\LMex\hbox{$\SavedStyle\sqcap$}\cr \raise-0.2\LMex\hbox{$\SavedStyle\sqcap$}}}{\sqcap}}}}
\def\mmmerge{\mathrel{\ThisStyle{\stretchrel*{\ooalign{\raise0.6\LMex\hbox{$\SavedStyle\sqcap$}\cr \raise0.2\LMex\hbox{$\SavedStyle\sqcap$}\cr \raise-0.2\LMex\hbox{$\SavedStyle\sqcap$}}}{\sqcap}}}}

\tikzstyle{hmscarrow}=[->, thick]
\tikzstyle{bmscbox}=[rounded corners, opacity=0.5]

\newcommand{\powerset}{\ensuremath{\mathscr{P}}}

\newcommand{\union}{\cup}
\newcommand{\inters}{\cap}
\newcommand{\Union}{\bigcup}

\DeclarePairedDelimiter\card{\lvert}{\rvert}

 \providecommand{\Coloneqq}{\mathrel{\mathop{::}}=} \newcommand{\is}{\coloneq}

\newcommand{\from}{\colon}

\def\grammOr{\hspace{3pt}\mid\hspace{3pt}}
\def\grammIs{\Coloneqq}

\begingroup
\catcode`\|=\active \gdef\@grammar@bar{\catcode`\|=\active \def|{\grammOr}}
\endgroup

\newcommand{\gramm}[1]{\begingroup
  \def\is{\grammIs}\@grammar@bar #1\endgroup }

\newenvironment{grammar}{\begin{equation*}\def\is{& \grammIs }\@grammar@bar \aligned }
{\endaligned \end{equation*}\aftergroup\ignorespaces }

\newcommand{\hole}{\hbox{-}}

\DeclarePairedDelimiterXPP\aenc[2]{\constr{a}}{(}{)}{_{#2}}{\strip@parens#1}
\DeclarePairedDelimiterXPP\pub[1]{\constr{p}}{(}{)}{}{\strip@parens#1}

\newcommand{\Parallel}{\@ifstar{\prod}{{\textstyle\prod}}}
\newcommand{\Alt}{\@ifstar{\sum}{{\textstyle\sum}}}
\newcommand{\Sum}{\Sigma}

\newcommand{\redtoover}[1]{\xrightarrow{#1}}

 \usepackage{subcaption}
\usepackage{graphicx}

\usepackage[noend]{algpseudocode}
\usepackage{algorithm}

\usepackage[capitalise]{cleveref}
\usepackage{newunicodechar}
\newunicodechar{∃}{\ensuremath{\exists}}
\newunicodechar{∀}{\ensuremath{\forall}}
\newunicodechar{θ}{\ensuremath{\theta}}
\newunicodechar{τ}{\ensuremath{\tau}}
\newunicodechar{φ}{\ensuremath{\varphi}}
\newunicodechar{ξ}{\ensuremath{\xi}}
\newunicodechar{ζ}{\ensuremath{\zeta}}
\newunicodechar{ψ}{\ensuremath{\psi}}
\newunicodechar{π}{\ensuremath{\pi}}
\newunicodechar{α}{\ensuremath{\alpha}}
\newunicodechar{β}{\ensuremath{\beta}}
\newunicodechar{γ}{\ensuremath{\gamma}}
\newunicodechar{δ}{\ensuremath{\delta}}
\newunicodechar{ε}{\ensuremath{\varepsilon}}
\newunicodechar{κ}{\ensuremath{\kappa}}
\newunicodechar{λ}{\ensuremath{\lambda}}
\newunicodechar{μ}{\ensuremath{\mu}}
\newunicodechar{ρ}{\ensuremath{\rho}}
\newunicodechar{σ}{\ensuremath{\sigma}}
\newunicodechar{ω}{\ensuremath{\omega}}
\newunicodechar{Γ}{\ensuremath{\Gamma}}
\newunicodechar{Φ}{\ensuremath{\Phi}}
\newunicodechar{Δ}{\ensuremath{\Delta}}
\newunicodechar{Σ}{\ensuremath{\Sigma}}
\newunicodechar{Π}{\ensuremath{\Pi}}
\newunicodechar{∑}{\ensuremath{\Sigma}}
\newunicodechar{∏}{\ensuremath{\Pi}}
\newunicodechar{Θ}{\ensuremath{\Theta}}
\newunicodechar{Ω}{\ensuremath{\Omega}}
\newunicodechar{⇒}{\ensuremath{\Rightarrow}}
\newunicodechar{⇐}{\ensuremath{\Leftarrow}}
\newunicodechar{⇔}{\ensuremath{\Leftrightarrow}}
\newunicodechar{→}{\ensuremath{\rightarrow}}
\newunicodechar{←}{\ensuremath{\leftarrow}}
\newunicodechar{↔}{\ensuremath{\leftrightarrow}}
\newunicodechar{¬}{\ensuremath{\neg}}
\newunicodechar{∧}{\ensuremath{\land}}
\newunicodechar{∨}{\ensuremath{\lor}}
\newunicodechar{≠}{\ensuremath{\neq}}
\newunicodechar{≡}{\ensuremath{\equiv}}
\newunicodechar{∼}{\ensuremath{\sim}}
\newunicodechar{≈}{\ensuremath{\approx}}
\newunicodechar{≥}{\ensuremath{\geq}}
\newunicodechar{≤}{\ensuremath{\leq}}
\newunicodechar{≫}{\ensuremath{\gg}}
\newunicodechar{≪}{\ensuremath{\ll}}
\newunicodechar{∅}{\ensuremath{\emptyset}}
\newunicodechar{⊆}{\ensuremath{\subseteq}}
\newunicodechar{⊂}{\ensuremath{\subset}}
\newunicodechar{∩}{\ensuremath{\cap}}
\newunicodechar{⋂}{\ensuremath{\cap}}
\newunicodechar{∪}{\ensuremath{\cup}}
\newunicodechar{⋃}{\ensuremath{\cup}}
\newunicodechar{⊎}{\ensuremath{\uplus}}
\newunicodechar{∈}{\ensuremath{\in}}
\newunicodechar{∉}{\ensuremath{\not\in}}
\newunicodechar{⊤}{\ensuremath{\top}}
\newunicodechar{⊥}{\ensuremath{\bot}}
\newunicodechar{₀}{\ensuremath{_0}}
\newunicodechar{₁}{\ensuremath{_1}}
\newunicodechar{₂}{\ensuremath{_2}}
\newunicodechar{₃}{\ensuremath{_3}}
\newunicodechar{₄}{\ensuremath{_4}}
\newunicodechar{₅}{\ensuremath{_5}}
\newunicodechar{₆}{\ensuremath{_6}}
\newunicodechar{₇}{\ensuremath{_7}}
\newunicodechar{₈}{\ensuremath{_8}}
\newunicodechar{₉}{\ensuremath{_9}}
\newunicodechar{⁰}{\ensuremath{^0}}
\newunicodechar{¹}{\ensuremath{^1}}
\newunicodechar{²}{\ensuremath{^2}}
\newunicodechar{³}{\ensuremath{^3}}
\newunicodechar{⁴}{\ensuremath{^4}}
\newunicodechar{⁵}{\ensuremath{^5}}
\newunicodechar{⁶}{\ensuremath{^6}}
\newunicodechar{⁷}{\ensuremath{^7}}
\newunicodechar{⁸}{\ensuremath{^8}}
\newunicodechar{⁹}{\ensuremath{^9}}
\newunicodechar{𝔹}{\ensuremath{\mathbb{B}}}
\newunicodechar{ℝ}{\ensuremath{\mathbb{R}}}
\newunicodechar{ℕ}{\ensuremath{\mathbb{N}}}
\newunicodechar{ℂ}{\ensuremath{\mathbb{C}}}
\newunicodechar{ℚ}{\ensuremath{\mathbb{Q}}}
\newunicodechar{𝕋}{\ensuremath{\mathbb{T}}}
\newunicodechar{𝕏}{\ensuremath{\mathbb{X}}}
\newunicodechar{ℤ}{\ensuremath{\mathbb{Z}}}
\newunicodechar{✓}{\ding{51}}
\newunicodechar{✗}{\ding{55}}
\newunicodechar{◊}{\ensuremath{\lozenge}}
\newunicodechar{□}{\ensuremath{\square}}
\newunicodechar{𝓐}{\ensuremath{\mathcal{A}}}
\newunicodechar{𝓑}{\ensuremath{\mathcal{B}}}
\newunicodechar{𝓒}{\ensuremath{\mathcal{C}}}
\newunicodechar{𝓓}{\ensuremath{\mathcal{D}}}
\newunicodechar{𝓔}{\ensuremath{\mathcal{E}}}
\newunicodechar{𝓕}{\ensuremath{\mathcal{F}}}
\newunicodechar{𝓖}{\ensuremath{\mathcal{G}}}
\newunicodechar{𝓗}{\ensuremath{\mathcal{H}}}
\newunicodechar{𝓘}{\ensuremath{\mathcal{I}}}
\newunicodechar{𝓙}{\ensuremath{\mathcal{J}}}
\newunicodechar{𝓚}{\ensuremath{\mathcal{K}}}
\newunicodechar{𝓛}{\ensuremath{\mathcal{L}}}
\newunicodechar{𝓜}{\ensuremath{\mathcal{M}}}
\newunicodechar{𝓝}{\ensuremath{\mathcal{N}}}
\newunicodechar{𝓞}{\ensuremath{\mathcal{O}}}
\newunicodechar{𝓟}{\ensuremath{\mathcal{P}}}
\newunicodechar{𝓠}{\ensuremath{\mathcal{Q}}}
\newunicodechar{𝓡}{\ensuremath{\mathcal{R}}}
\newunicodechar{𝓢}{\ensuremath{\mathcal{S}}}
\newunicodechar{𝓣}{\ensuremath{\mathcal{T}}}
\newunicodechar{𝓤}{\ensuremath{\mathcal{U}}}
\newunicodechar{𝓥}{\ensuremath{\mathcal{V}}}
\newunicodechar{𝓦}{\ensuremath{\mathcal{W}}}
\newunicodechar{𝓧}{\ensuremath{\mathcal{X}}}
\newunicodechar{𝓨}{\ensuremath{\mathcal{Y}}}
\newunicodechar{𝓩}{\ensuremath{\mathcal{Z}}}
\newunicodechar{…}{\ensuremath{\ldots}}
\newunicodechar{∗}{\ensuremath{\ast}}
\newunicodechar{⊢}{\ensuremath{\vdash}}
\newunicodechar{⊧}{\ensuremath{\models}}
\newunicodechar{′}{\ensuremath{'}}
\newunicodechar{″}{\ensuremath{''}}
\newunicodechar{‴}{\ensuremath{'''}}
\newunicodechar{∥}{\ensuremath{\|}}
\newunicodechar{⊕}{\ensuremath{\oplus}}
\newunicodechar{⁺}{\ensuremath{^+}}
\newunicodechar{⊇}{\ensuremath{\supseteq}}
\newunicodechar{∘}{\ensuremath{\circ}}
\newunicodechar{∙}{\ensuremath{\cdot}}
\newunicodechar{⋅}{\ensuremath{\cdot}}
\newunicodechar{≈}{\ensuremath{\approx}}
\newunicodechar{×}{\ensuremath{\times}}
\newunicodechar{∞}{\ensuremath{\infty}}
\newunicodechar{⊑}{\ensuremath{\sqsubseteq}}
 \usepackage{mathpartir}

\title{Generalising Projection in \\ Asynchronous Multiparty Session Types}
\titlerunning{
Generalising Projection in Asynchronous Multiparty Session Types}

\author{Rupak Majumdar}
{
Max Planck Institute for Software Systems, Kaiserslautern, Germany
}
{rupak@mpi-sws.org}
{}
{}
\author{Madhavan Mukund}
{
Chennai Mathematical Institute, India
\and CNRS IRL 2000, ReLaX, Chennai, India
}
{madhavan@cmi.ac.in}
{}
{}
\author{Felix Stutz}
{
Max Planck Institute for Software Systems, Kaiserslautern, Germany
}
{fstutz@mpi-sws.org}
{https://orcid.org/0000-0003-3638-4096}
{}

\author{Damien Zufferey}
{
Max Planck Institute for Software Systems, Kaiserslautern, Germany
}
{zufferey@mpi-sws.org}
{https://orcid.org/0000-0002-3197-8736}
{}

\authorrunning{R.\, Majumdar, M.\, Mukund, F.\, Stutz, and D.\, Zufferey} 
\Copyright{Rupak Majumdar, Madhavan Mukund, Felix Stutz, and Damien Zufferey}

\funding{This research was funded in part by the Deutsche Forschungsgemeinschaft project 389792660-TRR 248 and by the European Research Council under ERC Synergy Grant ImPACT (610150).}

\acknowledgements{The authors would like to thank Nobuko Yoshida, Thomas Wies, Elaine Li, and the anonymous reviewers for their feedback and suggestions.
They also thank Franco Barbanera, Ivan Lanese, and Emilio Tuosto for pointing out that an earlier version of the counterexample in \cref{fig:conditions-unsound-choreography-automata} was over-simplified and hence correctly rejected by their conditions \cite{DBLP:conf/coordination/BarbaneraLT20}.
}

\supplement{Tool available at: \\ Software (Source Code): 
\url{https://doi.org/10.5281/zenodo.5144684}
}

\nolinenumbers

\EventEditors{Serge Haddad and Daniele Varacca}
\EventNoEds{2}
\EventLongTitle{32nd International Conference on Concurrency Theory (CONCUR 2021)}
\EventShortTitle{CONCUR 2021}
\EventAcronym{CONCUR}
\EventYear{2021}
\EventDate{August 23--27, 2021}
\EventLocation{Virtual Conference}
\EventLogo{}
\SeriesVolume{203}
\ArticleNo{35}

\ccsdesc[500]{Theory of computation~Concurrency}

\begin{document}
\maketitle              \begin{abstract}

Multiparty session types (MSTs) provide an efficient methodology for specifying and verifying message passing software systems.
In the theory of MSTs,  a global type specifies the interaction among the roles at the global level.
A local specification for each role is generated by projecting from the global type on to the message exchanges it participates in.
Whenever a global type can be projected on to each role,
the composition of the projections is deadlock free and has exactly the behaviours specified by the global type.
The key to the usability of MSTs is the projection operation: a more expressive projection allows more systems to be type-checked
but requires a more difficult soundness argument.

In this paper, we generalise the standard projection operation in MSTs.
This allows us to model and type-check many design patterns in distributed systems, such as load balancing,
that are rejected by the standard projection.
The key to the new projection is an analysis that tracks causality between messages. 
Our soundness proof uses novel graph-theoretic techniques from the theory of message-sequence charts.
We demonstrate the efficacy of the new projection operation by showing many global types for common patterns
that can be projected under our projection but not under the standard projection operation.
 \keywords{Multiparty session types,
           Verification,
           Communicating state machines}
\end{abstract}
\section{Introduction}
\label{sec:intro}

Distributed message-passing systems are both widespread and challenging to design and implement.
A process tries to implement its role in a protocol with only the partial information received through messages.
The unpredictable communication delays mean that messages from different sources can be arbitrarily reordered.
Combining concurrency, asynchrony, and message buffering makes the verification problem
algorithmically undecidable \cite{DBLP:journals/jacm/BrandZ83} and principled design and verification
of such systems is an important challenge.

Multiparty Session Types (MSTs) 
\cite{DBLP:journals/jacm/HondaYC16,DBLP:journals/pacmpl/ScalasY19}
provide an appealing type-based approach for formalising and compositionally verifying structured concurrent distributed systems.
They have been successfully applied to
web services~\cite{DBLP:conf/tgc/YoshidaHNN13},
distributed algorithms~\cite{DBLP:journals/corr/abs-1902-01353},
smart contracts~\cite{DBLP:journals/corr/abs-1902-06056},
operating systems~\cite{DBLP:conf/eurosys/FahndrichAHHHLL06},
high performance computing~\cite{DBLP:conf/pvm/HondaMMNVY12},
timed systems~\cite{DBLP:conf/esop/BocchiMVY19},
cyber-physical systems~\cite{DBLP:conf/ecoop/MajumdarPYZ19},
etc.
By decomposing the problem of asynchronous verification on to local roles,
MSTs provide a clean and modular approach to the verification of distributed systems (see
the surveys \cite{DBLP:journals/ftpl/AnconaBB0CDGGGH16,DBLP:journals/csur/HuttelLVCCDMPRT16}).

The key step in MSTs is the \emph{projection} from a \emph{global type}, specifying all possible global message exchanges, to \emph{local types} for each role.
The soundness theorem of MSTs states that every projectable global type is \emph{implementable}:
 there is a distributed implementation that is free from communication safety errors such as deadlocks and unexpected messages.

The projection keeps only the operations observable by a given role
and yet maintains the invariant that every choice
can be distinguished in an unambiguous way.
Most current projection operations ensure this invariant by syntactically merging different paths locally for each role.
While these projections are syntactic and efficient, they are also very conservative and disallow many common design patterns in distributed systems.

In this paper, we describe a more general projection for MSTs to address the conservatism of
existing projections.
To motivate our extension, consider a simple load balancing protocol:
a client sends a request to a server and the server forwards the request to one of two workers.
The workers serve the request and directly reply to the client.
(We provide the formal syntax later.)
This common protocol is disallowed by existing MST systems, either because they syntactically disallow
such messages (the \emph{directed choice} restriction that states the sender and recipient must be the same
along every branch of a choice), or because the projection operates only on the global type and disallows
inferred choice.

The key difficulty in projection is to manage the interaction between choice and concurrency in a distributed setting.
Without choice, all roles would just follow one predetermined sequence of send and receive operations.
Introducing choice means a role either decides whom to send which message next,
or reacts to the choices of other roles---even if such choices are not locally visible.
This is only possible when the outcome of every choice propagates unambiguously.
At each point, every role either is agnostic to a prior choice or knows exactly the outcome of the choice,
even though it may only receive information about the choice indirectly through subsequent communication with other roles.
Unfortunately, computing how choice propagates in a system is undecidable in general \cite{DBLP:journals/tcs/AlurEY05};
this is the reason why conservative restrictions are used in practice.

The key insight in our projection operation is to manage the interaction of choice and concurrency via a 
\emph{message causality analysis}, inspired by the theory of communicating state machines (CSMs) and
message sequence charts (MSCs), that provides a more global view.
We resolve choice based on \emph{available messages} along different branches.
The causality analysis provides more information when merging two paths based on expected messages.

We show that our generalised projection subsumes previous approaches that lift the directed choice restriction \cite{DBLP:journals/corr/abs-1203-0780,DBLP:conf/fase/HuY17,DBLP:journals/acta/CastellaniDG19,DBLP:conf/esop/JongmansY20}.
Empirically, it allows us to model and verify common distributed programming patterns such as load balancing,
replicated data, and caching---where a server needs to choose between different workers---that are
not in scope of current MSTs, while preserving the efficiency of projection.

We show type soundness for generalised projection. 
This generalisation is non-trivial, since soundness depends on subtle arguments about asynchronous messages in the system.
We prove the result using an automata-theoretic approach, also inspired by the theory of MSCs, that argues about traces in 
communicating state machines.
Our language-theoretic proof is different from the usual proof-theoretic approaches in soundness proofs of MST systems,
and builds upon technical machinery from the theory of MSCs.

We show empirically that generalised choice is key to modelling several interesting instances in distributed
systems while maintaining the efficiency of more conservative systems.
Our global type specifications go beyond existing examples in the literature of MSTs.

 \section{Multiparty Session Types with Generalised Choice}
\label{sec:local}

In this section, we define global and local types. 
We explain how multiparty session types (MSTs) work and present a shortcoming of current MSTs.
Our MSTs overcome this shortcoming by allowing a role to wait for messages coming from different senders.
We define a new projection operation from global to local types: 
the projection represents global message exchanges from the perspective of a single role.
The key to the new projection is a generalised \emph{merge} operator that prevents confusion between messages from different senders.

\subsection{Global and Local Types}

We describe the syntax of global types following work by Honda et al.~\cite{DBLP:conf/popl/HondaYC08}, Hu and Yoshida~\cite{DBLP:conf/fase/HuY17}, and Scalas and Yoshida~\cite{DBLP:journals/pacmpl/ScalasY19}. 
We focus on the core message-passing aspects of asynchronous MSTs and do not include
features such as delegation. 

\begin{definition}[Syntax]
\emph{Global types for MSTs} are defined by the grammar:
    \begin{grammar}
     G \is
       0
     | \sum_{i ∈ I} \msgFromTo{\procA}{\procB_{i}}{\val_i.G_i}
     | μ t. \; G
     | t
    \end{grammar}\\[-3ex]
where $\procA, \procB_i$ range over a set of roles $\Procs$, $\val_i$ over a set of messages $\MsgVals$, and $t$ over type variables.
\end{definition}

Note that our definition of global types extends the standard syntax (see, e.g., \cite{DBLP:conf/popl/HondaYC08}), which has a 
\emph{directed choice} restriction, requiring that a sender must send messages to the same role along different branches of a choice.
Our syntax $\sum_{i\in I}\msgFromTo{\procA}{\procB_i}{\val_i.G_i}$
allows a sender to send messages to different roles along different branches (as in, e.g., \cite{DBLP:conf/fase/HuY17}).
For readability, we sometimes use the infix operator $+$ for choice, instead of~$\sum$. 
When $|I| = 1$, we omit~$\sum$.

In a global type, the send and the receive operations of a message exchange are specified atomically.
An expression $\msgFromTo{\procA}{\procB}{\val}$ represents two events: 
a \emph{send} $\snd{\procA}{\procB}{\val}$ and a \emph{receive} $\rcv{\procA}{\procB}{\val}$.
We require the sender and receiver processes to be different: $\procA ≠ \procB$.
A \emph{choice}~($\sum$) occurs at the sender role.
Each branch of a choice needs to be uniquely distinguishable: $∀ i,j ∈ I.\, i≠j ⇒ (\procB_{i},\val_i) ≠ (\procB_{j},\val_j)$.
The least fixed point operator encodes loops and we require recursion to be guarded, i.e., in $μ t. \, G$, there is at least one message between $μt$ and each $t$ in $G$. Without loss of generality, we assume that all occurrences of~$t$ are bound and each bound variable~$t$ is distinct.
As the recursion is limited to tail recursion, it is memoryless and
generates regular sequences, so a global type can be interpreted as a regular language of message exchanges.

\begin{figure}[t]
\begin{center}
\begin{subfigure}[b]{0.33\textwidth}
    \includegraphics[width=1.00\textwidth]{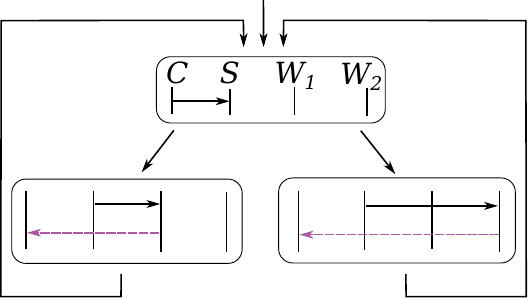}
    \caption{Load balancing}
    \label{fig:load-balance}
\end{subfigure}
    \hfill
\begin{subfigure}[b]{0.34\textwidth}
    \includegraphics[width=0.95\textwidth]{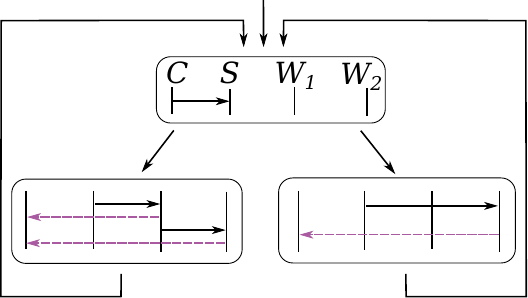}
    \caption{Variant of load balancing}\label{fig:load-balance-confusion}
\end{subfigure}
    \hfill
\begin{subfigure}[b]{0.15\textwidth}
    \includegraphics[width=0.95\textwidth]{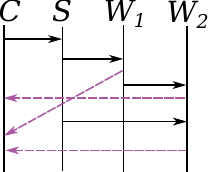}
    \caption{Execution} \label{fig:load-balance-confusion-execution}
\end{subfigure}
    \caption{Load balancing and some variant with potential for confusion exemplified by an execution}
    \label{fig:load-balance-confusion-combined}
\end{center}
\end{figure}

\begin{example}[Load balancing]
\label{ex:load-balancing}

A simple load balancing scenario can be modelled with the global type:
~~ {\small
$
\mu t. \; \msgFromTo{\mathsf{Client}}{\mathsf{Server}}{\mathit{req}}.
  + \begin{cases}
    \msgFromTo{\mathsf{Server}}{\mathsf{Worker_1}}{\mathit{req}}. \;\msgFromTo{\mathsf{Worker_1}}{\mathsf{Client}}{\mathit{reply}}.\; t \\
    \msgFromTo{\mathsf{Server}}{\mathsf{Worker_2}}{\mathit{req}}. \;\msgFromTo{\mathsf{Worker_2}}{\mathsf{Client}}{\mathit{reply}}.\; t
    \end{cases}
$}

The least fixed point operator $μ$ encodes a loop in which a client sends a request to a server.
The server then non-deterministically forwards the request to one of two workers.
The chosen worker handles the request and replies to the client.
In this protocol, the server communicates with a different worker in each branch. \cref{fig:load-balance} shows this example as a high-level message sequence chart (HMSC).
The timeline of roles is shown with vertical lines and the messages with horizontal arrows.
Different message contents are represented by different styles of arrows.
\qed
\end{example}

Next, we define local types, which specify a role's view of a protocol.

\begin{definition}[Local types]
\label{def:local-type}
The \emph{local types} for a role $\procA$ are defined as: 
    \begin{grammar}
     L \is 0
         | \IntCh_{i ∈ I} \snd{}{\procB_i}{\val_i}.L_i
         | \ExtCh_{i ∈ I} \rcv{\procB_{i}}{}{\val_i}.L_i
         | μ t. L
         | t
    \end{grammar}
where the internal choice $(\IntCh)$ and  external choice $(\ExtCh)$ both respect
$∀ i,j ∈ I.\; i≠j ⇒ (\procB_i, \val_i) ≠ (\procB_j, \val_j)$.
As for global types,
we assume every recursion variable is bound, each recursion
operator ($μ$) uses a different identifier $t$,
and we may omit $\IntCh$ and $\ExtCh$ if $\card{I} = 1$. \end{definition}

Note that a role can send to, resp.\ receive from, multiple roles in a choice:
we generalise $\IntCh_{i ∈ I} \snd{}{\procB}{\val_i}.L_i$ of standard MSTs to $\IntCh_{i ∈ I} \snd{}{\procB_{\textcolor{red}{i}}}{\val_i}.L_i$ and
              $\ExtCh_{i ∈ I} \rcv{\procB}{}{\val_i}.L_i$ to $\ExtCh_{i ∈ I} \rcv{\procB_{\textcolor{red}{i}}}{}{\val_i}.L_i$.

\begin{example}\label{ex:load-balance-local-types}
We can give the following local types for  \cref{fig:load-balance}:
 
\vspace{1ex}
{\small
\hspace{5mm}
$
\begin{array}{rl}
\mathsf{Server} : &
μ t.\, \rcv{\mathsf{Client}}{}{\mathit{req}}.\,
    \left(
    \snd{}{\mathsf{Worker_1}}{\mathit{req}}.\, t
  \; \IntCh \;
    \snd{}{\mathsf{Worker_2}}{\mathit{req}}.\, t
    \right)
    \\
\mathsf{Client} : &
μ t.\, \snd{}{\mathsf{Server}}{\mathit{req}}.\,
    \left(
    \rcv{\mathsf{Worker_1}}{}{\mathit{reply}}.\, t
    \; \ExtCh \;
    \rcv{\mathsf{Worker_2}}{}{\mathit{reply}}.\, t
    \right)
\\
\mathsf{Worker_{i}} : &
μ t.\, \rcv{\mathsf{Server}}{}{\mathit{req}}.\,
    \snd{}{\mathsf{Client}}{\mathit{reply}}.\, t
\text{ for }
  i \in \set{1, 2}
\end{array}
$
}
\vspace{1ex}

\noindent
Note that their structure, i.e., having a loop with at most two options, resembles the one of the global type in \cref{ex:load-balancing}.
\qed
\end{example}

Our goal is to define a partial \emph{projection} operation from a given global type to a local type for each role.
If the projection is defined, we expect that the type is \emph{implementable}.
We shall show that the global type of \cref{ex:load-balancing} projects to the local types in \cref{ex:load-balance-local-types}.
As a consequence, the global type is implementable.
Intuitively, when each role in the example executes based on its local type, they agree on a unique global path in an unrolling of the global type.
We formalise projection and soundness in Section~\ref{sec:soundness}.
We note that existing projection operations, including the ones by Hu and Yoshida~\cite{DBLP:conf/fase/HuY17} as well as Scalas and Yoshida~\cite{DBLP:journals/pacmpl/ScalasY19},
reject the above global type as not implementable.

\subparagraph*{Notations and Assumptions.} 
We write $\GG$ for the global type we try to project.
When traversing the global type $\GG$, we use $G$ for the current term (which is a subterm of $\GG$).
To simplify the notation, we assume that the index $i$ of a choice uniquely determines the sender and the message $\rcv{\procB_i}{}{\val_i}$.
Using this notation, we write $I ∩ J$ to select the set of choices with identical sender and message value and $I \setminus J$ to select the alternatives present in $I$ but not in $J$.
When looking at send and receive events in a global setting we write $\snd{\procA}{\procB}{\val}$ for $\procA$ sending to~$\procB$ and $\rcv{\procA}{\procB}{\val}$ for $\procB$ receiving from $\procA$.

In later definitions, we unfold the recursion in types.
We could get the unfolding through a congruence relation.
However, this requires dealing with infinite structures, which makes some definitions not effective. Instead, we precompute the map from each recursion variable~$t$ to its unfolding.
For a given global type, let $\getMu$ be a function that returns a map from $t$ to~$G$ for each subterm $μt.\,G$.
Recall, each $t$ in a type is different.
$\getMu$ is defined as follows:

{} \hfill
$
  \getMu(0) \is [\,]
$ \hfill $
  \getMu(t) \is [\,]
$ \hfill $
  \getMu(μt.G) \is [t \mapsto G] ∪ \getMu(G)
$  \hfill {}

{} \hfill $
  \getMu(\sum_{i ∈ I} \msgFromTo{\procA}{\procB_i}{\val_i.G_i}) \is \bigcup_{i∈I} \getMu(G_i)
$  \hfill {}

\noindent
We write $\getMuG$ as shorthand for the map returned by $\getMu(\GG)$.

\subsection{Generalised Projection and Merge}
\label{sec:projection-merge}

We now define a partial \emph{projection} operation
that projects a global type on to each role.
The projection on to a role $\procC$ is a local type and keeps only $\procC$'s actions.
Intuitively, it gives the ``local view'' of message exchanges performed by $\procC$.
While projecting, non-determinism may arise due to choices that $\procC$ does not observe directly.
In this case, the different branches are merged using a partial \emph{merge} operator ($\merge$). 
The merge operator checks that a role, which has not yet learned the outcome of a choice,
only performs actions that are allowed in all possible branches.
The role can perform branch-specific actions after it has received a message that resolves the choice. 
For a role that is agnostic to the choice, i.e., behaves the same on all the branches, 
the merge allows the role to proceed as if the choice does not exist.

So far, the idea follows standard asynchronous MSTs.
What distinguishes our new projection operator from prior ones (e.g., \cite{DBLP:conf/fase/HuY17,DBLP:journals/pacmpl/ScalasY19}), 
is that we allow a role to learn which branch has been taken through messages received from different senders. 
This generalisation is non-trivial.
When limiting the reception to messages from a single role, one can rely on the FIFO order provided by the corresponding channel.
However, messages coming from different sources are only partially ordered.
Thus, unlike previous approaches, our merge operator looks at the result of a \emph{causality analysis} on the 
global type to make sure that this partial ordering cannot introduce any confusion.

\begin{example}[Intricacies of generalising projection]
\label{ex:extended-proj}
\label{ex:need-for-extended-projection}
We demonstrate that a straightforward generalisation of existing projection operators can lead to unsoundness.
Consider a \emph{naive} projection that merges branches with internal choice if they are equal,
and for receives, simply always merges external choices -- also from different senders.
In addition, it removes empty loops. 
For \cref{fig:load-balance}, this naive projection yields the expected local types presented in \cref{ex:load-balance-local-types}.
We show that naive projection can be unsound.
\cref{fig:load-balance-confusion} shows a variant of load balancing, for which naive projection yields the following local types:

\vspace{1ex}
\noindent
\hspace{5mm}
$
\begin{array}{rl}
\mathsf{Server} : &
μ t.\, \rcv{\mathsf{Client}}{}{\mathit{req}}.\,
    \left(
    \snd{}{\mathsf{Worker_1}}{\mathit{req}}.\, t
  \; \IntCh \;
    \snd{}{\mathsf{Worker_2}}{\mathit{req}}.\, t
    \right)
    \\
\mathsf{Client} : &
μ t.\,
  \snd{}{\mathsf{Server}}{\mathit{req}}.\,\left(
      \rcv{\mathsf{Worker_1}}{}{\mathit{reply}}.\,\rcv{\mathsf{Worker_2}}{}{\mathit{reply}}.\,t
    \; \& \;
      \rcv{\mathsf{Worker_2}}{}{\mathit{reply}}.\,t
  \right)
\\
\mathsf{Worker_{1}} : &
μ t.\,
    \rcv{\mathsf{Server}}{}{\mathit{req}}.\,
    \snd{}{\mathsf{Client}}{\mathit{reply}}.\,
    \snd{}{\mathsf{Worker_2}}{\mathit{req}}.\, t
\\
\mathsf{Worker_{2}} : &
μ t.\,
\left(
        \rcv{\mathsf{Worker_1}}{}{\mathit{req}}.\,
        \snd{}{\mathsf{Client}}{\mathit{reply}}.\,t
        \; \& \;
        \rcv{\mathsf{Server}}{}{\mathit{req}}.\,
        \snd{}{\mathsf{Client}}{\mathit{reply}}.\,t
    \right)
\end{array}
$
\vspace{1ex}

\noindent
Unfortunately, the global type is not implementable. The problem is that, for the $\mathit{Client}$, the two messages on its left branch are not causally related.
Consider the execution prefix in \cref{fig:load-balance-confusion-execution} which is not specified in the global type.
The $\mathsf{Server}$ decided to first take the left ($L$) and then the right ($R$) branch.
For $\mathsf{Server}$, the order $LR$ is obvious from its events and the same applies for $\mathsf{Worker_2}$.
For $\mathsf{Worker_1}$, every possible order $R^*LR^*$ is plausible as it does not have events in the 
right branch. Since $LR$ belongs to the set of plausible orders, there is no confusion. Now, the messages from the two workers to the client are independent and, therefore, 
can be received in any order.
If the client receives $\rcv{\mathit{Worker_2}}{}{\mathit{reply}}$ first, then 
its local view is consistent with the choice $RL$ as the order of branches.
This can lead to confusion and, thus, execution prefixes which are not specified in the global type. 
\qed
\end{example}

We shall now define our generalised projection operation.
To identify confusion as above, we keep track of causality between messages.
We determine what messages a role \emph{could} receive at a given point in the global type through an \emph{available messages} analysis.
Tracking causality needs to be done at the level of the global type.
We look for chains of dependent messages and we also need to unfold loops.
Fortunately, since we only check for the presence or absence of some messages, it is sufficient to unfold each recursion at most once.

\subparagraph*{Projection and Interferences from Independent Messages.}
The challenge of projecting a global type lies in resolving the non-determinism introduced by having only the endpoint~view.
\cref{ex:extended-proj} shows that in order to decide if a choice is safe, we need to know which messages can arrive at the same time.
To enable this, we annotate local types with  the messages that could be received at that point in the protocol.
We call these \emph{availability annotated local types} and write them as $\mathit{AL} = \apair{L}{\mathit{Msg}}$ where
    $L$ is a local type and
    $\mathit{Msg}$ is a set of messages.
This signifies that when a role has reached $\mathit{AL}$, the messages in $\mathit{Msg}$ can be present in the communication channels.
We annotate types using the  grammar for local types (\cref{def:local-type}), where each subterm is annotated.
To recover a local type, we erase the annotation, i.e., recursively replace each $\mathit{AL} = \apair{L}{\mathit{Msg}}$ by $L$.
The projection internally uses annotated types.

The projection of $\GG$ on to $\procC$, written $\GG \tproj_\procC$, 
traverses $\GG$ to erase the operations that do not involve $\procC$. 
During this phase, we also compute the messages that $\procC$ may receive.
The function $\avail(B, T, G)$ computes the set of messages that other roles can send while $\procC$ has not yet learned the outcome of the choice.
This set depends on
$B$, the set of \emph{blocked} roles, i.e., the roles which are waiting to receive a message and hence cannot move;
$T$, the set of recursion variables we have already visited; and
$G$, the subterm in $\GG$ at which we compute the available messages.
We defer the definition of $\avail(B, T, G)$ to later in this section.
\subparagraph*{Empty Paths Elimination.}
When projecting, there may be paths and loops where a role neither sends nor receives a message, e.g., the right loop in \cref{ex:load-balancing} for $\mathsf{Worker_1}$.
Such paths can be removed during projection.
Even if conceptually simple, the notational overhead impedes understandability of how our message availability analysis is used.
Therefore, we first focus on the message availability analysis and define a projection operation that does not account for empty paths elimination.
After defining the merge operator $\merge$, we give the full definition of our generalised projection operation.

\begin{definition}[Projection without empty paths elimination]
The \emph{projection without empty paths elimination} $\GG \tproj_\procC$
of a global type $\GG$ on to a role $\procC\in \Procs$ 
is an availability annotated local type
inductively defined as:
\vspace{1.5ex}

\begin{small}
$
 t \tproj_\procC \is \apair{t}{\avail(\set{\procC},\set{t},\getMuG(t))}
$
 \hfill
$
 0 \tproj_\procC \is \apair{0}{∅}
$
 \hfill {}

 \vspace{1mm}
$
 \left( μ t. G \right)
 \tproj_\procC \is
 \begin{cases}
    \apair{μ t. (G \tproj_\procC)}{\avail(\set{\procC},\set{t},G)}
        & \text{if } G \tproj_\procC \neq \apair{t}{\_} \\ \apair{0}{∅}
        & \text{otherwise}
 \end{cases}
$

$
 \left( \sum_{i ∈ I} \msgFromTo{\procA}{\procB_i}{\val_i.G_i} \right)
 \tproj_\procC \is
 \begin{cases}
    \apair{
        \IntCh_{i ∈ I} \snd{}{\procB_i}{\val_i}.(G_i \tproj_\procC)
    }{
        \bigcup_{i∈I} \avail(\set{\procB_i,\procC},∅,G_i)
    } 
    \hfill \text{ if } \procC = \procA \\[1mm]
    \merge
    \left(
    \begin{array}{l}
        \apair{
            \ExtCh_{i \in I_{[=\procC]}} 
            \rcv{\procA}{}{\val_i}.(G_i \tproj_\procC)
        }{
            \bigcup_{i∈I_{[=\procC]}} \avail(\set{\procC},∅,G_i)
        }\\
        \Merge_{i \in I_{[\neq\procC]}} 
            G_i \tproj_\procC
    \end{array}
    \right) \hfill \quad \text{otherwise}
 \end{cases}
 $
 \\
 \phantom{sometext}
 \hspace{28mm}
 $
    \begin{small}
        \text{ where }
        I_{[= \procC]} \is \set{i \in I \mid \procB_i = \procC}
        \text{ and }
        I_{[\neq \procC]} \is \set{i \in I \mid \procB_i \neq \procC}
    \end{small}
$
\end{small}

A global type $\GG$ is said to be \emph{\projectable} if $\GG \tproj_\procC$ is defined for every $\procC \in \Procs$.
\end{definition}

Projection erases events not relevant to $\procC$ by a recursive traversal of the global type;
however, at a choice not involving $\procC$, it has to ensure that
either $\procC$ is indifferent to the outcome of the choice or
it indirectly receives enough information to distinguish the outcome. 
This is managed by the merge operator $\merge$ and the use of available messages.
The merge operator takes as arguments a sequence of availability annotated local types.
Our merge operator generalises the full merge by Scalas and Yoshida~\cite{DBLP:journals/pacmpl/ScalasY19}.
When faced with choice, it only merges receptions that cannot interfere with each other.
For the sake of clarity, we define only the binary merge.
As the operator is commutative and associative, it generalises to a set of branches $I$.
When $I$ is a singleton, the merge just returns that one branch.

\begin{definition}[Merge operator $\merge$]
\label{def:merge-operator-wo-check-avail}
Let $\apair{L₁}{\mathit{Msg}₁}$ and $\apair{L₂}{\mathit{Msg}₂}$ be availability annotated local types for a role $\procC$.
$\apair{L₁}{\mathit{Msg}₁} \Merge \,\apair{L₂}{\mathit{Msg}₂}$ is defined by cases, as follows:
\begin{small}
\begin{itemize}
\item $\apair{L₁}{\mathit{Msg}₁ ∪ \mathit{Msg}₂}$ $\;$ if $L₁ = L₂$
\item $\apair{μt₁.(\mathit{AL}₁ \merge \mathit{AL}₂[t₂/t₁])}{\mathit{Msg}₁ ∪ \mathit{Msg}₂}$
        $\;$
        if $L₁ =  μt₁.\mathit{AL}₁ \text{, } L₂ = μt₂.\mathit{AL}₂$
\item $\apair{\IntCh_{i \in I} \snd{}{\procB_i}{m_i}.(\mathit{AL}_{1,i} \merge \mathit{AL}_{2,i})}{\mathit{Msg}₁ ∪ \mathit{Msg}₂}$
        $\;$
        if $\begin{cases}
                L₁ = \IntCh_{i ∈ I} \snd{}{\procB_i}{\val_i.\mathit{AL}_{1,i}} \text{, } \\
                L₂ = \IntCh_{i ∈ I} \snd{}{\procB_i}{\val_i.\mathit{AL}_{2,i}}
           \end{cases}$
\item
$\begin{array}{llr}
\langle &   \ExtCh_{i \in I \setminus J} \rcv{\procB_i}{}{m_i}.\mathit{AL}_{1,i} & \ExtCh \\
        &   \ExtCh_{i \in I∩J} \rcv{\procB_i}{}{m_i}.(\mathit{AL}_{1,i} \merge \mathit{AL}_{2,i}) & \ExtCh \\
        &   \ExtCh_{i \in J \setminus I} \rcv{\procB_i}{}{m_i}.\mathit{AL}_{2,i} & , \\
        & \mathit{Msg}₁ ∪ \mathit{Msg}₂ \quad \rangle
\end{array}$
        $\;$
       if $\begin{cases}
            L₁ = \ExtCh_{i ∈ I} \rcv{\procB_i}{}{\val_i.\mathit{AL}_{1,i}} \text{, } \\
            L₂ = \ExtCh_{i ∈ J} \rcv{\procB_i}{}{\val_i.\mathit{AL}_{2,i}} \text{, } \\
            ∀ i∈ I \setminus J.\, \rcv{\procB_i}{\procC}{\val_i} \notin \mathit{Msg}₂ \text{, } \\
            ∀ i∈ J \setminus I.\, \rcv{\procB_i}{\procC}{\val_i} \notin \mathit{Msg}₁
          \end{cases}$
\end{itemize}
\end{small}
\end{definition}
When no condition applies, the merge and, thus, the projection are undefined.\footnote{
When we use the n-ary notation $\Merge_{i∈I}$ and $\card{I} = 0$, we implicitly omit this part.
Note that this can only happen if $\procC$ is the receiver among all branches for some choice so there is either another local type to merge with, or the projection is undefined anyway.
}

The important case of the merge is the external choice.
Here, when a role can potentially receive a message that is unique to a branch, it checks that the message cannot be available in another branch so actually being able to receive this message uniquely determines which branch was taken by the role to choose.
For the other cases, a role can postpone learning the branch as long as the actions on both branches are the same.

\subparagraph*{Adding Empty Paths Elimination.}
The preliminary version of projection requires every role to have at least one event in each branch of a loop and, thus, rejects examples where a role has no event in some loop branch.
Such paths can be eliminated.
However, determining such empty paths cannot be done on the level of the merge operator but only when projecting.
To this end, we introduce an additional parameter $E$ for the generalised projection: $E$ contains those variables $t$ for which $\procC$ has not observed any message send or receive event since $μt$.

\begin{definition}[Generalised projection -- with empty paths elimination]
\label{def:gen-projection}
The \emph{projection} $\GG \tproj^E_\procC$
of a global type $\GG$ on to a role $\procC\in \Procs$ is an availability annotated local type
which is inductively defined as follows:
\vspace{1ex}

\begin{footnotesize}
$
 t \tproj^{\textcolor{red}{E}}_\procC \is \apair{t}{\avail(\set{\procC},\set{t},\getMuG(t))}
$
 \hfill
$
 0 \tproj^{\textcolor{red}{E}}_\procC \is \apair{0}{∅}
$
 \hfill {}
 
 \vspace{1mm}
$
 \left( μ t. G \right)
 \tproj^{\textcolor{red}{E}}_\procC \is
 \begin{cases}
    \apair{μ t. (G \tproj^{\textcolor{red}{E ∪ \set{t}}}_\procC)}{\avail(\set{\procC},\set{t},G)}
        & \text{if } G \tproj^{\textcolor{red}{E ∪ \set{t}}}_\procC \neq \apair{t}{\_} \\
    \apair{0}{∅}
        & \text{otherwise}
 \end{cases}
 $ 
 
 \vspace{1mm}
 $
 \left( \sum_{i ∈ I} \msgFromTo{\procA}{\procB_i}{\val_i.G_i} \right)
 \tproj^{\textcolor{red}{E}}_\procC \is
 \begin{cases}
    \apair{
        \IntCh_{i ∈ I} \snd{}{\procB_i}{\val_i}.(G_i \tproj^{\textcolor{red}{∅}}_\procC)
    }{
        \bigcup_{i∈I} \avail(\set{\procB_i,\procC},∅,G_i)
    }
\hfill \text{if } \procC = \procA \\
    \merge
    \left(
    \begin{array}{l}
        \apair{
            \ExtCh_{i \in I_{[=\procC]}} 
            \rcv{\procA}{}{\val_i}.(G_i \tproj^{\textcolor{red}{∅}}_\procC)
        }{
            \bigcup_{i \in I_{[=\procC]}} \avail(\set{\procC},∅,G_i)
        }\\
        \Merge_{\,i \in I_{[\neq\procC]} \, 
        {\textcolor{red}{∧ \, ∀ t ∈ E.\, G_i \tproj^E_\procC ≠ \apair{t}{\_}}}}
        G_i \tproj^{\textcolor{red}{E}}_\procC
    \end{array}
    \right) \hfill \quad \text{otherwise} 
 \end{cases}
 $
 \phantom{sometext}
 \hspace{29mm}
 $
    \begin{small}
        \text{ where }
        I_{[= \procC]} \is \set{i \in I \mid \procB_i = \procC}
        \text{ and }
        I_{[\neq \procC]} \is \set{i \in I \mid \procB_i \neq \procC}
    \end{small}
$
\end{footnotesize}

Since the merge operator $\Merge$ is partial, the projection may be undefined.
We use $\GG \tproj_\procC$ as shorthand for $\GG \tproj^∅_\procC$ and 
only consider the generalised projection with empty paths elimination from now on.
A global type $\GG$ is called \emph{\projectable} if $\GG \tproj_\procC$ is defined for every role $\procC \in \Procs$.
\end{definition}

We highlight the differences for the empty paths elimination.
Recall that $E$ contains all recursion variables from which the role $\procC$ has not encountered any events.
To guarantee this, for the case of recursion $\mu t.\, G$, the (unique) variable $t$ is added to the current set $E$, while the parameter turns to the empty set $\emptyset$ as soon as $\procC$ encounters an event.
The previous steps basically constitute the necessary bookkeeping.
The actual elimination is achieved with the condition 
$
∀ t ∈ E.\, G_i \tproj^E_\procC ≠ \apair{t}{\_}
$
which filters all branches without events of role~$\procC$.

Other works \cite{DBLP:conf/fase/HuY17,DBLP:journals/acta/CastellaniDG19} achieve this with \emph{connecting actions}, marking non-empty paths.
Like classical MSTs, we do not include such explicit actions.
Still, we can automatically eliminate such paths in contrast to previous work.

\subparagraph*{Computing Available Messages.}
Finally, the function $\avail$ is computed recursively: \begin{center}
\begin{minipage}{0.98\linewidth}
\begin{footnotesize}
\noindent
$
\avail(B, T, 0) \is ∅ \\\avail(B, T, μt.G) \is \avail(B, T ∪ \set{t}, G) \\
\avail(B, T, t) \is \begin{cases}
∅ & \text{if} ~ t ∈ T\\
\avail(B, T ∪ \set{t}, \getMuG(t)) & \text{if} ~ t ∉ T
\end{cases}\\
\avail(B, T, \sum_{i ∈ I} \msgFromTo{\procA}{\procB_i}{\val_i.G_i}) \is \begin{cases}
\bigcup_{i∈I,m∈\MsgVals} (\avail(B, T, G_i) \setminus \set{ \rcv{\procA}{\procB_i}{\val} }) ∪  \set{ \rcv{\procA}{\procB_i}{\val_i} } \quad \hfill \text{if} ~ \procA ∉ B\\
\bigcup_{i∈I} \avail(B ∪ \set{ \procB_i }, T, G_i) \quad \hfill \text{if} ~ \procA ∈ B 
\end{cases}
$
\end{footnotesize}
\end{minipage}
\end{center}

Since all channels are FIFO, we only keep the first possible message in each channel.
The fourth case discards messages not at the head of the channel.

Our projection is different from the one of Scalas and Yoshida \cite{DBLP:journals/pacmpl/ScalasY19}, not just because our syntax is more general.
It also represents a shift in paradigm.
In their work, the full merge works only on local types.
No additional knowledge is required.
This is  possible because their type system limits the flexibility of communication.
Since we allow more flexible communication, we need to keep some information, in form of available messages, about the possible global executions for the merge operator.

\begin{figure}[t]
\centering
\begin{subfigure}[b]{0.35\textwidth}
    \includegraphics[width=\textwidth]{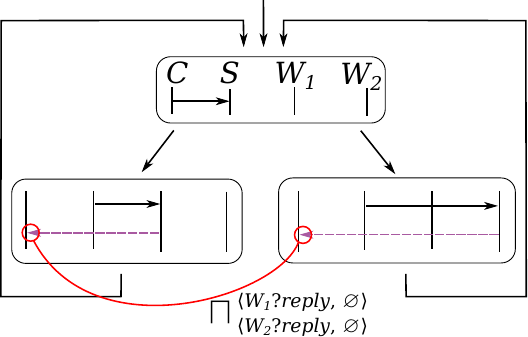}
    \caption{Merging for \cref{fig:load-balance}}
    \label{fig:merge-1}
\end{subfigure}
\hfill
\begin{subfigure}[b]{0.46\textwidth}
    \includegraphics[width=\textwidth]{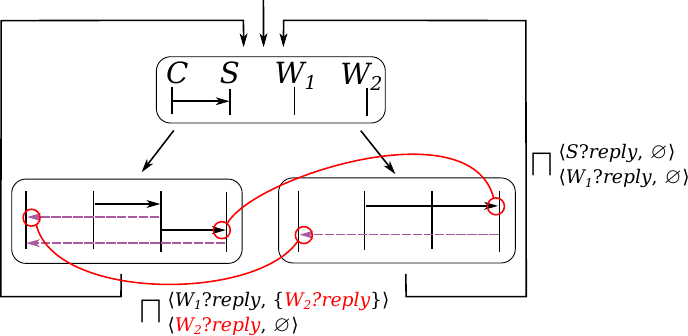}
    \caption{Merging for \cref{fig:load-balance-confusion} \phantom{some text}}
    \label{fig:merge-2}
\end{subfigure}
\caption{
    Availability annotated types for merging on the two examples.
    The red lines connect the receptions that get merged during projection.
    The annotations only show the receiver's messages.
}
\vspace{-1em}
\label{fig:simple-choice-avail}
\end{figure}

\begin{example}
Let us explain how our projection operator catches the problem in $\GG \tproj_{\mathit{Client}}$ of \cref{fig:load-balance-confusion}.
\cref{fig:simple-choice-avail} shows the function of available messages during the projection for \cref{fig:load-balance} and \cref{fig:load-balance-confusion}. 
In \cref{fig:merge-1}, the messages form chains, i.e., except for the role making the choice, a role only sends in reaction to another message.
Therefore, only a single message is available at each reception and the protocol is \projectable.
On the other hand, in \cref{fig:merge-2} both replies are available and, therefore, the protocol is not \projectable.

Here are the details of the projection for \cref{fig:merge-2}.
If not needed, we omit the availability annotations for readability.
Recall that $\mathsf{Client}$ receives $\mathit{reply}$ from $\mathsf{Worker_2}$ in the left branch, which is also present in the right branch.
Let us denote the two branches as follows:

\vspace{1ex}
\noindent
{\small
\hspace{5mm}
$
\begin{array}{ll}
 G_1  \is &
    \msgFromTo{\mathsf{Server}}{\mathsf{Worker_1}}{\mathit{req}}. \;
    \msgFromTo{\mathsf{Worker_1}}{\mathsf{Client}}{\mathit{reply}}. \; \\
    &
    \phantom{
    \msgFromTo{\mathsf{Server}}{\mathsf{Worker_1}}{\mathit{req}}. \;
    }
    \msgFromTo{\mathsf{Worker_1}}{\mathsf{Worker_2}}{\mathit{req}}. \;
    \msgFromTo{\mathsf{Worker_2}}{\mathsf{Client}}{\mathit{reply}}. \;
    t,
    \text{ and}\\
 G_2  \is &
    \msgFromTo{\mathsf{Server}}{\mathsf{Worker_2}}{\mathit{req}}. \;
    \msgFromTo{\mathsf{Worker_2}}{\mathsf{Client}}{\mathit{reply}}.\;
    t
\end{array}
$
}
\vspace{1ex}
\\
\noindent
Since the first message in $G_1$ does not involve $\mathsf{Client}$, the projection descends and we compute:

\vspace{1ex}
{ \small
\hspace{5mm}
$
\begin{array}{cl}
 \GG \tproj_\mathsf{Client} & = \langle μ t.(
            \apair{\rcv{\mathsf{Worker_1}}{}{\mathit{reply}}.\,(G'_1 \tproj_\mathsf{Client})}{\,\avail(\set{\mathsf{Client}}, \emptyset, G'_1)} \\
            & ~ \qquad \merge \hspace{0.2mm}
            \apair{\rcv{\mathsf{Worker_2}}{}{\mathit{reply}}.\,(G'_2 \tproj_\mathsf{Client})}{\,\avail(\set{\mathsf{Client}}, \emptyset, G'_2)} ), \_ \rangle \\
  \text{where } & G'_1 = 
  \msgFromTo{\mathsf{Worker_1}}{\mathsf{Worker_2}}{\mathit{req}}. \;\msgFromTo{\mathsf{Worker_2}}{\mathsf{Client}}{\mathit{reply}}.\; t 
  \text{ and } G'_2 = t.
\end{array}
$
}
\vspace{1ex}
\\
\noindent
For this, we compute both availability annotations:
$\avail(\set{\mathsf{Client}}, \emptyset, G'_2) = \emptyset$
which is empty as all message exchanges in the recursion ($t$) initiate from $\mathsf{Client}$ which is blocked;
and
$\avail(\set{\mathsf{Client}}, \emptyset, G'_1) = \set{\rcv{\mathsf{Worker_1}}{\mathsf{Worker_2}}{\mathit{req}}, \rcv{\mathsf{Worker_2}}{\mathsf{Client}}{\mathit{reply}}}$
which contains both receptions from $G'_1$ only since no receptions from the recursion are added for the same reason.
The conditions are not satisfied as the message for $\mathsf{Client}$ in the second branch is available in the first one:
$\rcv{\mathsf{Worker_2}}{\mathsf{Client}}{\mathit{reply}} \in \avail(\set{\mathsf{Client}}, \emptyset, G'_1)$.
Thus, the projection is undefined.
\qed
\end{example}

 \section{Type Soundness}
\label{sec:soundness}

We now show a soundness theorem for generalised projection; roughly, a \projectable global type can be implemented by communicating state machines
in a distributed way.
Our proof uses automata-theoretic techniques from the theory of MSCs.
We assume familiarity with the basics of formal languages.

As our running example shows, a protocol implementation often cannot enforce the event ordering specified in the type but only a weaker order.
In this section, we capture both notions through a type language and an execution language.

\subsection{Type Languages}

A \emph{state machine} $A = (Q, \Sigma, \delta, q_{0}, F)$
consists of
a finite set $Q$ of states,
an alphabet $\Sigma$,
a transition relation $\delta \subseteq Q \times (\Sigma \union \set{\emptystring}) \times Q$,
an initial state $q_{0}\in Q$, and
a set $F \subseteq Q$ of final states.
We write $q \xrightarrow{x} q'$ for $(q, x, q')\in\delta$.
We define the runs and traces in the standard way.
A run is maximal if it is infinite or if it ends at a final state.
The \emph{language} $\lang(A)$ is the set of (finite or infinite) maximal traces.
The projection $A \wproj_{\Delta}$ of a state machine is its projection to a sub-alphabet $\Delta\subseteq \Sigma$ obtained
by replacing all letters in $\Sigma\setminus \Delta$ with $\emptystring$-transitions.
It accepts the language $\lang(A)\wproj_\Delta = \set{w\wproj_\Delta \mid w \in \lang(A)}$.

\begin{definition}[Type language for global types]
\label{def:language-global-mst}
The semantics of a global type $\GG$ is given as a regular language.
We construct a state machine $\semglobal(\GG)$ using an auxiliary state machine $M(\GG)$.
First, we define $M(\GG) = (Q_{M(\GG)}, \Sigma_{\mathit{sync}}, δ_{M(\GG)}, q_{0 M(\GG)}, F_{M(\GG)})$ where
\vspace{-1ex}
\begin{itemize}
\item $Q_{M(\GG)}$ is the set of all syntactic subterms in $\GG$ together with the term $0$,
\item $Σ_{\mathit{sync}} = \set{ \msgFromTo{\procA}{\procB}{\val} \mid \procA,\procB ∈ \Procs \text{ and } \val ∈ \MsgVals}$,
\item $δ_{M(\GG)}$ is the smallest set containing
            $(\sum_{i ∈ I} \msgFromTo{\procA}{\procB_i}{\val_i.G_i}, \msgFromTo{\procA}{\procB_i}{\val_i}, G_i)$ for each $i ∈ I$,
            as well as $(μ t. G', ε, G')$ and $(t, ε, μ t. G')$ for each subterm~$\mu t.G'$ of~$\GG$,
\item $q_{0 M(\GG)} = \GG$ and 
$F_{M(\GG)} = \set{0}$.
\end{itemize}

Next, we expand each message $\msgFromTo{\procA}{\procB}{\val}$ into two events, $\snd{\procA}{\procB}{\val}$ followed by $\rcv{\procA}{\procB}{\val}$.
We define $\semglobal(\GG) = (Q_{\semglobal(\GG)}, \Sigma_{\semglobal(\GG)}, \delta_{\semglobal(\GG)}, q_{0 \semglobal(\GG)}, F_{\semglobal(\GG)})$ as follows:
\vspace{-1ex}
\begin{itemize}
\item $Q_{\semglobal(\GG)} = Q_{M(\GG)} ∪ (Q_{M(\GG)} × Σ_{\mathit{sync}} × Q_{M(\GG)})$,
\item $Σ_{\semglobal(\GG)} = \set{ \snd{\procA}{\procB}{\val} \mid \procA, \procB ∈ \Procs, \val ∈ \MsgVals} \cup \set{ \rcv{\procA}{\procB}{\val} \mid \procA, \procB ∈ \Procs, \val ∈ \MsgVals}$,
\item $δ_{\semglobal(\GG)}$ is the smallest set containing the
  transitions
  $(s,\snd{\procA}{\procB}{\val},(s,\msgFromTo{\procA}{\procB}{\val},s'))$
  and $((s,\msgFromTo{\procA}{\procB}{\val},s'),\rcv{\procA}{\procB}{\val},
  s'))$ for each transition $(s,\msgFromTo{\procA}{\procB}{\val},s') ∈
  δ_{M(\GG)}$,
\item $q_{0\semglobal(\GG)} = q_{0{M(\GG)}}$
and
$F_{\semglobal(\GG)} = F_{M(\GG)}$. 
\end{itemize}

The \emph{type language} $\lang(\GG)$ of a global type $\GG$ is given by $\lang(\semglobal(\GG))$.
\end{definition}

\begin{definition}[Type language for local types]
\label{def:language-local-mst}
Given a local type $L$ for $\procA$, we construct a state machine $\semlocal(L) = (Q, \Sigma_\procA, δ, q₀, F)$ where
\vspace{-1ex}
\begin{itemize}
\item $Q$ is the set of all syntactic subterms in $L$, 
\item $\Sigma_\procA = \set{ \snd{\procA}{\procB}{\val} \mid \procB ∈ \Procs, \val ∈ \MsgVals} \cup \set{ \rcv{\procB}{\procA}{\val} \mid \procA, \procB ∈ \Procs, \val ∈ \MsgVals}$,
\item $δ$ is the smallest set containing \\
            $(\IntCh_{i ∈ I} \snd{}{\procB_i}{\val_i.L_i}, \snd{\procA}{\procB_i}{\val_i}, L_i)$ and
            $(\ExtCh_{i ∈ I} \rcv{\procB_i}{}{\val_i.L_i}, \rcv{\procB_i}{\procA}{\val_i}, L_i)$ for each $i ∈ I$, 
            \\ 
            as well as 
            $(μ t. L', ε, L')$
            and $(t, ε, μ t. L')$
	    for each $\mu t.L'$ in $L$,
\item $q₀ = L$ and
$F = \set{0}$ if $0$ is a subterm of $L$, and empty otherwise.
\end{itemize}
We define the \emph{type language} of $L$ as language of this automaton:
$\lang(L) = \lang(\semlocal(L))$.
\end{definition}

\subsection{Implementability}
\label{sec:implementability}

An \emph{implementation} consists of a set of state machines, one per role, communicating with each other through
asynchronous messages and pairwise FIFO channels.
We use communicating state machines (CSMs) \cite{DBLP:journals/jacm/BrandZ83} as our formal model.
A CSM $\CSM{A}$ consists of a set of state machines $A_\procA$, one for each $\procA\in\Procs$ over the alphabet of message sends and receives.
Communication between machines happens asynchronously through FIFO channels.
The semantics of a CSM is a language $\lang(\CSM{A})$ of maximal traces over the alphabet of message sends and receives satisfying the FIFO condition
on channels.
A CSM is \emph{deadlock free} if every trace can be extended to a maximal trace.
We omit the (standard) formal definition of CSMs (see \cref{app:csms} for details).

\subparagraph*{Indistinguishability Relation.}
In the type language of a global type, every send event is always immediately succeeded by its receive event.
However, in a CSM, other independent events may occur between the send and the receipt and there is no way to force the order specified by the type language.
To capture this phenomenon formally, we define a family of \emph{indistinguishability relations} 
${\interswap_i} \subseteq \Sigma^* \times \Sigma^*$, for $i\geq 0$
and $\Sigma = \Sigma_{\semglobal(\GG)}$, as follows.
For all $w\in\Sigma^*$, we have $w \interswap_0 w$.
For $i=1$, we define:
\vspace{-1ex}
\begin{enumerate}[label=(\arabic*)]
\item
If $\procA ≠ \procC$, then
$
 w.\snd{\procA}{\procB}{\val}.\snd{\procC}{\procD}{\val'}.u
 \; \interswap_{1} \;
 w.\snd{\procC}{\procD}{\val'}.\snd{\procA}{\procB}{\val}.u
$.

\item
If $\procB ≠ \procD$, then
$
 w.\rcv{\procA}{\procB}{\val}.\rcv{\procC}{\procD}{\val'}.u
 \; \interswap_{1} \;
 w.\rcv{\procC}{\procD}{\val'}.\rcv{\procA}{\procB}{\val}.u
$.

\item
If $\procA ≠ \procD \land (\procA ≠ \procC ∨ \procB ≠ \procD)
$, then 
$
 w.\snd{\procA}{\procB}{\val}.\rcv{\procC}{\procD}{\val'}.u
 \; \interswap_{1} \;
 w.\rcv{\procC}{\procD}{\val'}.\snd{\procA}{\procB}{\val}.u
$.
\item
If $\card{w \wproj_{\snd{\procA}{\procB}{\_}}} >
    \card{w \wproj_{\rcv{\procA}{\procB}{\_}}}$,
then 
$
 w.\snd{\procA}{\procB}{\val}.\rcv{\procA}{\procB}{\val'}.u
 \; \interswap_{1} \;
 w.\rcv{\procA}{\procB}{\val'}.\snd{\procA}{\procB}{\val}.u
$.
\end{enumerate}
We refer to the proof of \cref{lm:csm-closed-under-interswap} in \cref{app:props-interswaplang} for further details on the conditions for swapping events.
Let $w, w', w''$ be sequences of events s.t.~$w \interswap_1 w'$ and $w' \interswap_i w''$ for some~$i$.
Then, $w \interswap_{i+1} w''$.
We define $w \interswap u$ if $w \interswap_n u$ for some $n$.
It is straightforward that $\interswap$ is an equivalence relation.
Define $u \preceq_\interswap v$ if there is $w\in\Sigma^*$ such that $u.w \interswap v$.
Observe that $u \interswap v$ iff
$u \preceq_\interswap v$ and $v \preceq_\interswap u$.
To extend $\interswap$ to infinite words, we follow the approach of Gastin~\cite{DBLP:conf/litp/Gastin90}.
For infinite words $u, v\in\Sigma^\omega$, we define $u \preceq_\interswap^\omega v$ 
if for each finite prefix $u'$ of $u$, there is a finite prefix~$v'$ of~$v$ such that
$u' \preceq_\interswap v'$.
Define $u \interswap v$ iff $u \preceq_\interswap^\omega v$ and $v\preceq_\interswap^\omega u$.

We lift the equivalence relation $\interswap$ on words to languages:

For a language $\Lambda$, we define
{
\small
$
  \interswaplang(\Lambda) = \left\{ w' \mid \bigvee
    \begin{array}{l}
    w' \in \Alphabet^* \land ∃ w ∈ \Alphabet^*. \; w \in \Lambda \text{ and } w' \interswap w \\
    w' ∈ \Alphabet^ω \land \exists w \in \Alphabet^\omega. \; w \in 
    \Lambda \text{ and } w' \preceq_\interswap^\omega w
  \end{array} \right\}.
$
}

For the infinite case, we take the downward closure w.r.t.~$\preceq_\interswap^\omega$.
Unlike \cite[Definition~2.1]{DBLP:conf/litp/Gastin90}, our closure operator is asymmetric.
Consider the protocol $(\snd{\procA}{\procB}{\val}.\;\rcv{\procA}{\procB}{\val})^ω$.
Since we do not make any fairness assumption on scheduling, we need to include in the closure the execution where only the sender is scheduled, i.e., $(\snd{\procA}{\procB}{\val})^ω \preceq_\interswap^\omega (\snd{\procA}{\procB}{\val}.\;\rcv{\procA}{\procB}{\val})^ω$.

\begin{example}[Indistinguishability relation $\interswap$ by examples]
 \label{ex:indistinguishability-relation}
 The four rules for $\interswap_1$ present conditions under which two adjacent events in an execution (prefix) can be swapped.
 These conditions are designed such that they characterise possible changes in an execution (prefix) which cannot be recognised by any CSM.
 To be precise, if $w$ is recognised by some CSM $\CSM{A}$ and $w' \interswap_1 w$ holds, then $w'$ is also recognised by $\CSM{A}$.
 In this example, we illustrate the intuition behind these rules.
 
 For the remainder of this example, the \emph{active role} of an event is the receiver of a receive event and the sender of a send event. 
 Visually, the active role is always the first role in an event.
 In addition, we assume that variables do not alias, i.e., two roles or messages with different names are different.
 
 Two send events (or two receive events) can be swapped if the active roles are distinct because there cannot be any dependency between two such events which do occur next to each other in an execution.
 For send events, the 1st rule, thus, admits 
 $\snd{\procA}{\procC}{\val}.\;\snd{\procB}{\procC}{\val}$
 $\interswap_1$
 $\snd{\procB}{\procC}{\val}.\;\snd{\procA}{\procC}{\val}$
 even though the receiver is the same.
 In contrast, the corresponding receive events cannot be swapped:
 $\rcv{\procA}{\procC}{\val}.\;\rcv{\procB}{\procC}{\val}$
 $\not \interswap_1$
 $\rcv{\procB}{\procC}{\val}.\;\rcv{\procA}{\procC}{\val}$.
 Note that the 1st rule is the only one with which two send events can be swapped while the 2nd rule is the only one for receive events so indeed no rule applies for the last case.
 
 The 3rd rule allows one send and one receive event to be swapped if either both senders or both receivers are different -- in addition to the requirement that both active roles are different.
 For instance, it admits
 $\snd{\procA}{\procC}{\val}.\;\rcv{\procC}{\procB}{\val}$
 $\interswap_1$
 $\rcv{\procC}{\procB}{\val}.\;\snd{\procA}{\procC}{\val}$.
 However, it does not admit two swap 
 $\snd{\procA}{\procB}{\val}.\;\rcv{\procA}{\procB}{\val}$
 $\not \interswap_1$
 $\rcv{\procA}{\procB}{\val}.\;\snd{\procA}{\procB}{\val}$.
 This is reasonable since the send event could be the one which emits $\val$ in the corresponding channel.
 In this execution prefix, this is in fact the case since there have been no events before, but in general one needs to incorporate the context to understand whether this is the case.
 The 4th rule does this and therefore admits swapping the same events when appended to some prefix:
 $\snd{\procA}{\procB}{\val}.
  \snd{\procA}{\procB}{\val}.\;\rcv{\procA}{\procB}{\val}$
 $\not \interswap_1$
 $\snd{\procA}{\procB}{\val}.
  \rcv{\procA}{\procB}{\val}.\;\snd{\procA}{\procB}{\val}$.
 Then, the FIFO order of channels ensures that the first message will be received first and the 2nd send event can happen after the reception of the 1st message.
\qed
\end{example}

\begin{example}[Load balancing revisited]
Let us consider the execution with confusion in \cref{fig:load-balance-confusion-execution}.
Compared to a synchronous execution, we need to delay the reception $\rcv{W_1}{C}{\mathit{reply}}$ to come after the first $\rcv{W_2}{C}{\mathit{reply}}$.
Using the 2nd and 3rd cases of $\interswap$ we can move $\rcv{W_1}{C}{\mathit{reply}}$ across the communications between the two workers. 
Finally, we use the 3rd case again to swap $\rcv{W_1}{C}{\mathit{reply}}$ and $\snd{W_2}{C}{\mathit{reply}}$ to get the desired sequence.
\qed
\end{example}

This example shows that $\interswap$ does not change the order of send and receive events of a single role.
Thus, if $w, w'\in \Sigma_{\procA}^\infty$, then $w\interswap w'$ iff $w = w'$.
Hence, any language over the message alphabet of a single role is (trivially) closed under $\interswap$.
Further, two runs of a CSM on $w$ and $w'$ with $w \interswap w'$ end in the same configuration. 

\subparagraph*{Execution Languages.}
For a global type $\GG$, the above discussion implies that any implementation $\CSM{A}$ can at most achieve that $\lang(\CSM{A}) = \interswaplang(\lang(\GG))$.
This is why we call $\interswaplang(\lang(\GG))$ the \emph{execution language} of $\GG$.
One might also call $\interswaplang(\lang(L))$ of a local type $L$ an execution language, however, since $\interswap$ does not swap any events on the same role, 
the type language and execution language are equivalent.

\begin{definition}[Implementability]
A global type $\GG$ is said to be \emph{implementable} if there exists a CSM $\CSM{A}$ s.t.
(i) [protocol fidelity] $\lang(\CSM{A}) = \interswaplang(\lang(\GG))$, and
(ii) [deadlock freedom] $\CSM{A}$ is deadlock free.
We say that $\CSM{A}$ implements $\GG$.
\end{definition}

\subsection{Type Soundness: Projectability implies Implementability}

The projection operator preserves the local order of events for every role and does not remove any possible event. 
Therefore, we can conclude that, for each role, the projected language of the global type is subsumed by the language of the projection.

\begin{proposition}
\label{prop:projection-preserves-per-process-runs}
For every \projectable $\GG$, role $\procC\in\Procs$, run with trace $w$ in $\semglobal(\GG) \wproj_{\Alphabet_\procC}$, there is a run with trace $w$ in $\semlocal(\GG \tproj_\procC)$.
Therefore, it holds that $\lang(\GG) \wproj_{\Alphabet_\procC} \subseteq \lang(\GG \tproj_\procC)$.
\end{proposition}

The previous result shows that the projection does not remove behaviours.
Now, we also need to show that it does not add unwanted behaviours.
The main result is the following.

\begin{theorem}
\label{thm:projectable-MST-sat-protocol-fidelity}
If a global type $\GG$ is \projectable, then
$\GG$ is implementable.
\end{theorem}

The complete proof can be found in the extended version \cref{sec:local-proof}. Here, we give a brief summary of the main ideas.
To show that a \projectable global type is implemented by its projections, we need to show that the projection preserves behaviours, does not add behaviours, and is deadlock free.
With \cref{prop:projection-preserves-per-process-runs},
showing that the projections combine to admit at least the behaviour specified by the global protocol is straightforward.
For the converse direction, we establish a property of the executions of the local types with respect to the global type:
all the projections agree on the run taken by the overall system in the global type.
We call this property \emph{control flow agreement}.
Executions that satisfy control flow agreement also satisfy protocol fidelity and are deadlock free.
The formalisation and proof of this property is complicated by the fact that not all roles learn about a choice at the same time.
Some roles can perform actions after the choice has been made and before they learn which branch has been taken.
In the extreme case, a role may not learn at all that a choice happened.
The key to control flow agreement is in the definition of the merge operator.
We can simplify the reasoning to the following two points.

\subparagraph*{Roles learn choices before performing distinguishing actions.}
When faced with two branches with different actions, a role that is not making the choice needs to learn the branch by receiving a message.
This follows from the definition of the merge operator.
Let us call this message the choice message.
Merging branches is only allowed as long as the actions are similar for this role.
When there is a difference between two (or more) branches, an external choice is the only case that allows a role to continue on distinct branches.

\subparagraph*{Checking available messages ensures no confusion.}
From the possible receptions ($\rcv{\procB_i}{}{\val_i}$) in an external choice, any pair of sender and message is unique among this list for the choice.
This follows from two facts.
First, the projection computes the available messages along the different branches of the choice.
Second, merging uses that information to make sure that the choice message of one branch does not occur in another branch as a message independent of that branch's choice messages.

\begin{example}
Let us use an example to illustrate why this is non-trivial. Consider:

$
    \GG \is
    \bigl( \msgFromTo{\procA}{\procB}{l}.\,
        \mu t.\, \msgFromTo{\procC}{\procA}{\val}.\, t \bigr)
  +
    \bigl( \msgFromTo{\procA}{\procB}{r}.\,
        \mu s.\, \msgFromTo{\procC}{\procA}{\val}.\, s \bigr)
$

\noindent
with its projections:

$
 \GG \tproj_\procA =
    \bigl( \snd{}{\procB}{l}.\,
        \mu t.\,\rcv{\procC}{}{\val}.\,t \bigr)
  \; \IntCh \;
    \bigl( \snd{}{\procB}{r}.\,
        \mu s.\,\rcv{\procC}{}{\val}.\,s \bigr) 
\quad \;
 \GG \tproj_\procB =
   \rcv{\procA}{}{l}.\,0
  \; \ExtCh \;
   \rcv{\procA}{}{r}.\,0 
\quad \;
 \GG \tproj_\procC  =
        \mu t.\,\snd{}{\procA}{\val}.\,t
$

\noindent
and an execution prefix $w$ of $\CSMl{\semlocal(\GG \tproj_\procA)}$:
\hfill
$
    \snd{\procC}{\procA}{\val}. \;
    \snd{\procC}{\procA}{\val}. \;
    \snd{\procA}{\procB}{l}. \;
    \rcv{\procC}{\procA}{\val}. \;
    \snd{\procC}{\procA}{\val}. 
$

\noindent
For this execution prefix, we check which runs in $\semglobal(\GG)$ each role could have pursued.
In this case, $\procC$ is not directly affected by the choice so it can proceed before the $\procA$ has even made the choice.
As the part of the protocol after the choice is a loop, we cannot bound how far some roles can proceed before the choice gets resolved.
\qed
\end{example}

\begin{remark}
Our projection balances expressiveness with tractability:
it does not unfold recursion,
i.e., the merge operator never expands a term $μt.G$ to obtain the local type (and we only unfold once to obtain the set of available messages).
Recursion variables are only handled by equality. While this restriction may seem arbitrary, unfolding can lead to comparing unbounded sequences of messages and, hence, undecidability \cite{DBLP:journals/tcs/AlurEY05} or non-effective constructions~\cite{DBLP:journals/corr/abs-1203-0780}. Our projection guarantees that a role is either agnostic to a choice or receives a choice message in an horizon bounded by the size of the type.
\end{remark}
 \section{Evaluation}

We implemented our generalised projection in a prototype tool which is publicly available \cite{prototype}.
The core functionality is implemented in about 800 lines of Python3 code.
Our tool takes as input a global type and outputs its projections (if defined).
We run our experiments on a machine with an Intel Xeon E5-2667 v2 CPU.
\Cref{tb:proj} presents the results of our evaluation.

\begin{table}
\caption{Projecting MSTs. For each example, we report
    the size as the number of nodes in the AST,
    the number of roles,
    whether it uses our extension,
    the time to compute the projections.
}
\label{tb:proj}

\begin{center}

\begin{tabular}{l l  c c  c c  c  c  c c}
 \toprule
 Source &
 Name &
 Size &
 $|𝓟|$ &
 Gen.\ Proj.\ needed &
 Time \\
 \midrule
 \multirow{6}{*}{\cite{DBLP:journals/pacmpl/ScalasY19}} &
 Instrument Contr.\ Prot.\ A &
 $16$ &
 $3$ &
 ✗ &
 $0.50$\,ms
 \\
 &
 Instrument Contr.\ Prot.\ B &
 $13$ &
 $3$ &
 ✗ &
 $0.41$\,ms
 \\
 &
 Multi Party Game &
 $16$ &
 $3$ &
 ✗ &
 $0.48$\,ms
 \\
 &
 OAuth2 &
 $7$ &
 $3$ &
 ✗ &
 $0.29$\,ms
 \\
 &
 Streaming &
 $7$ &
 $4$ &
 ✗ &
 $0.33$\,ms
 \\
 \cmidrule{1-1}
 \cite{DBLP:journals/corr/abs-1203-0780} &
 Non-Compatible Merge &
 $5$ &
 $3$ &
 ✓ &
 $0.22$\,ms
\\
 \cmidrule{1-1}
 \cite{springhibernate} &
 Spring-Hibernate &
 $44$ &
 $6$ &
 ✓ & $1.97$\,ms
\\
 \cmidrule{1-1}
 \multirow{4}{*}{New} &
 Group Present &
 $43$ &
 $4$ &
 ✓ & $1.62$\,ms
 \\
 &
 Late Learning &
 $12$ &
 $4$ &
 ✓ &
 $0.56$\,ms
 \\
 &
 Load Balancer ($n=10$) &
 $32$ &
 $12$ &
 ✓ & $8.18$\,ms \\
 &
 Logging ($n=10$) &
 $56$ &
 $13$ &
 ✓ & $20.96$\,ms
 \\
\bottomrule
\end{tabular}
 \end{center}
\end{table}

Our prototype successfully projects global types from the
literature \cite{DBLP:journals/pacmpl/ScalasY19}, in particular Multi-Party Game, OAuth2, Streaming, and two corrected versions of the Instrument Control Protocol.
These existing examples can be projected, but do not require generalised projection.

\subparagraph*{The Need for Generalised Projection.}
The remaining examples exemplify when our generalised projection is needed.
In fact, if some role can receive from different senders along two paths (immediately or after a sequence of same actions), its projection is only defined for the generalised projection operator.
To start with, our generalised projection can project a global type presented by Castagna et al.~\cite[p.~19]{DBLP:journals/corr/abs-1203-0780} which is not \projectable with their effective projection operator (see \cref{sec:rel-generalisations} for more details).
The Spring-Hibernate example was obtained by translating a UML sequence diagram \cite{springhibernate} to a global type.
There, \texttt{Hibernate Session} can receive from two different senders along two paths.
The Group Present example is a variation of the traditional book auction example \cite{DBLP:journals/jacm/HondaYC16} and describes a protocol where friends organise a birthday present for someone; in the course of the protocol, some people can be contacted by different people.
The Late Learning example models a protocol where a role submits a request and the server replies either with \texttt{reject} or \texttt{wait}, however, the last case applies to two branches where the result is served by different roles.
The Load Balancer (\cref{ex:load-balancing}) and Logging examples are simple versions of typical communication patterns in distributed computing.
The examples are parameterised by the number of workers, respectively, back-ends that call the logging service, to evaluate the efficiency of projection.
For both, we present one instance ($n=10$) in the table.
All new examples are rejected by previous approaches but shown projectable by our new projection.

\subparagraph*{Overhead.}
The generalised projection does not incur any overhead for global types that do not need it.
Our implementation computes the sets of available messages lazily, i.e., it is only computed if our message causality analysis is needed.
These sets are only needed when merging receptions from different senders.
We modelled four more parameterised protocols: Mem Cache, Map Reduce, Tree Broadcast, and P2P Broadcast.
We tested these examples, which do not need the generalised projection, up to size 1000 and found that our generalised projection does not add any overhead.
Thus, while the message causality analysis is crucial for our generalised projection operator and hence applicability of MST verification, it does not affect its efficiency.

 \section{Related Work}
\label{rel-work}

\subparagraph*{Session Types.}
MSTs stem from process algebra and they have been proposed for typing communication channels.
The seminal work on binary session types by Honda~\cite{DBLP:conf/concur/Honda93} identified channel duality as a condition for safe two party communication.
This work was inspired by linear logic \cite{DBLP:journals/tcs/Girard87} and lead to further studies on the connections between session types and linear logic \cite{DBLP:journals/jfp/Wadler14,DBLP:journals/mscs/CairesPT16}.
Moving from binary to multiparty session types, Honda et al.~\cite{DBLP:conf/popl/HondaYC08} identified consistency as the generalisation of duality for the multiparty setting.
The connection between MSTs and linear logic is still ongoing \cite{DBLP:conf/forte/CairesP16,DBLP:conf/concur/CarboneLMSW16,DBLP:journals/acta/CarboneMSY17}.
While we focus solely on communication primitives, the theory is extended with other features such as delegation \cite{DBLP:conf/esop/HondaVK98,DBLP:conf/popl/HondaYC08,DBLP:journals/tcs/CastellaniDGH20} and dependent types \cite{DBLP:conf/ppdp/ToninhoCP11,DBLP:journals/corr/abs-1208-6483,DBLP:conf/fossacs/ToninhoY18}.
These extensions have their own intricacies and we leave incorporating such features into our generalised projection for future work.

In this paper, we use local types directly as implementations for roles for simplicity.
Subtyping investigates ways to give implementation freedom while preserving the correctness properties.
For further details on subtyping, we refer to work by Lange and Yoshida~\cite{DBLP:conf/fossacs/LangeY17}, Bravetti et al.~\cite{DBLP:journals/tcs/BravettiCZ18}, and
Chen et al.~\cite{DBLP:conf/ppdp/ChenDY14,DBLP:journals/lmcs/ChenDSY17}.

\subparagraph*{Generalisations of Choice in MSTs.}
\label{sec:rel-generalisations}
Castagna et al.~\cite{DBLP:journals/corr/abs-1203-0780} consider a generalised choice similar to this work.
They present a non-effective general approach for projection, relying on global information, and an algorithmic projection which is limited to local information.
Our projection keeps some global information in the form of message availability and, therefore, handles a broader class of protocols.
For instance, our generalised projection operator can project the following example \cite[p.~19]{DBLP:journals/corr/abs-1203-0780} but their algorithmic version cannot:

{ \small
\hspace{5mm}
$
\bigl(
    \msgFromTo{\procA}{\procC}{a}. \;
    \msgFromTo{\procC}{\procA}{a}. \;
    \msgFromTo{\procA}{\procB}{a}. \;
    \msgFromTo{\procB}{\procC}{b}. \;
    0
\bigr)
+
\bigl(
    \msgFromTo{\procA}{\procB}{a}. \;
    \msgFromTo{\procB}{\procC}{b}. \;
    0
\bigr)
$
}

Hu and Yoshida \cite{DBLP:conf/fase/HuY17} syntactically allow a sender to send to different recipients in global and local types as well as a receiver to receive from different senders in local types.
However, their projection is only defined if a receiver receives messages from a single role.
From our evaluation, all the examples that needs the generalised projection are rejected by their projection.
Recently, Castellani et al.~\cite{DBLP:journals/acta/CastellaniDG19} investigated ways to allow local types to specify receptions from multiple senders for reversible computations but
only in the synchronous setting.
Similarly, for synchronous communication only, Jongmans and Yoshida \cite{DBLP:conf/esop/JongmansY20} discuss generalising choice in MSTs. Because their calculus has an explicit parallel composition, they can emulate some asynchronous communication but their channels have bag semantics
instead of FIFO queues.
The correctness of the projection also computes causality among messages as in our case and shares the idea of annotating local types.

\subparagraph*{Choreography Automata.}
Choreography automata \cite{DBLP:conf/coordination/BarbaneraLT20} and graphical choreographies \cite{DBLP:conf/popl/LangeTY15} model protocols as automata with transitions labelled by message exchanges, e.g., $\msgFromTo{\procA}{\procB}{\val}$.
Barbanera el al.~\cite{DBLP:conf/coordination/BarbaneraLT20} develop conditions for safely mergeable branches that ensure implementability on synchronous choreography automata.
However, when lifting them to the asynchronous setting, they miss the subtle introduction of partial order for messages from different senders.
Consider the choreography automaton in \cref{fig:conditions-unsound-choreography-automata}.
It can also be represented as a global type:
\vspace{0.2mm}
{ \small
\hspace{5mm}
$
+ \;
\begin{cases}
    \msgFromTo{\procA}{\procD}{{\color{orange}m_1}}. \;
    \msgFromTo{\procA}{\procE}{{\color{black}m}}. \;
    \msgFromTo{\procA}{\procD}{{\color{black}m}}. \;
    \msgFromTo{\procD}{\procE}{{\color{black}m}}. \;
    \msgFromTo{\procE}{\procA}{{\color{orange}m_1}}. \; 0
    \\
    \msgFromTo{\procA}{\procD}{{\color{purple}m_2}}. \;
    \msgFromTo{\procD}{\procE}{{\color{black}m}}. \;
    \msgFromTo{\procD}{\procA}{{\color{black}m}}. \;
    \msgFromTo{\procA}{\procD}{{\color{black}m}}. \;
    \msgFromTo{\procA}{\procE}{{\color{black}m}}. \;
    \msgFromTo{\procE}{\procA}{{\color{purple}m_2}}. \; 0
\end{cases}
$
}
\vspace{0.2mm}

It is \emph{well-formed} according to their conditions. However, $\procE$ cannot determine which branch was chosen since the messages $\val$ from $\procA$ and $\procD$ are not ordered when sent asynchronously.
As a result, it can send $\val_2$ in the top (resp.\ left) branch which is not specified as well as $\val_1$ in the bottom (resp.\ right) branch.

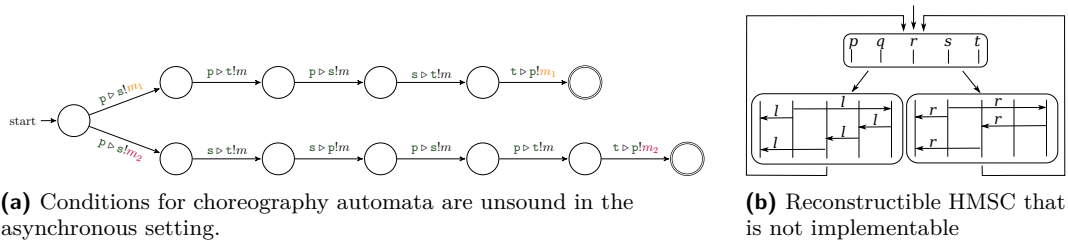
\begin{figure}[t]
\begin{subfigure}[b]{0.7\textwidth}
\resizebox{0.95\textwidth}{!}{
    \begin{tikzpicture}[->,>=stealth',shorten >=1pt,auto,node distance=2.8cm,
                    semithick]
  \tikzstyle{every state}=[circle,draw,text=white]

  \node[initial,state] (A)              {};
  \node[state]         (B1) [yshift=3em, right of=A] {};
  \node[state]         (B2) [right of=B1] {};
  \node[state]         (B3) [right of=B2] {};
  \node[state]         (B4) [right of=B3] {};
  \node[state,accepting]         (B5) [right of=B4] {};
  \node[state]         (C1) [yshift=-3em, right of=A] {};
  \node[state]         (C2) [right of=C1] {};
  \node[state]         (C3) [right of=C2] {};
  \node[state]         (C4) [right of=C3] {};
  \node[state]         (C5) [right of=C4] {};
  \node[state,accepting]         (C6) [right of=C5] {};

  \path (A) edge node [sloped]{$\snd{\procA}{\procD}{{\color{orange}m_1}}$} (B1)
        (A) edge node [sloped,below]{$\snd{\procA}{\procD}{{\color{purple}m_2}}$} (C1)
        (B1) edge node [sloped]{$\snd{\procA}{\procE}{{\color{black}m}}$} (B2)
        (B2) edge node [sloped]{$\snd{\procA}{\procD}{{\color{black}m}}$} (B3)
        (B3) edge node [sloped]{$\snd{\procD}{\procE}{{\color{black}m}}$} (B4)
        (B4) edge node [sloped]{$\snd{\procE}{\procA}{{\color{orange}m_1}}$} (B5)
        (C1) edge node [sloped]{$\snd{\procD}{\procE}{{\color{black}m}}$} (C2)
        (C2) edge node [sloped]{$\snd{\procD}{\procA}{{\color{black}m}}$} (C3)
        (C3) edge node [sloped]{$\snd{\procA}{\procD}{{\color{black}m}}$} (C4)
        (C4) edge node [sloped]{$\snd{\procA}{\procE}{{\color{black}m}}$} (C5)
        (C5) edge node [sloped]{$\snd{\procE}{\procA}{{\color{purple}m_2}}$} (C6);

\end{tikzpicture}
 }
\caption{Conditions for choreography automata are unsound in the \\ asynchronous setting.}
\label{fig:conditions-unsound-choreography-automata}
\end{subfigure}\begin{subfigure}[b]{0.3\textwidth}
\begin{center}
    \includegraphics[width=1.00\textwidth]{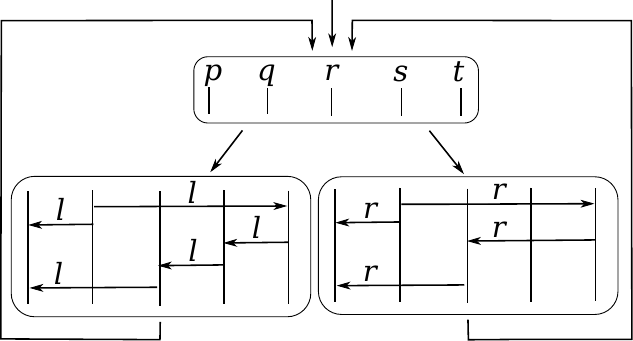}
    \caption{Reconstructible HMSC that is not implementable}
    \label{fig:wrong-local-choice}
\end{center}
\end{subfigure}\caption{These examples show that previous results for the asynchronous setting are flawed.}
\end{figure}

Lange et al.~\cite{DBLP:conf/popl/LangeTY15} have shown how to obtain graphical choreographies from CSM executions.
Unfortunately, they cannot fully handle unbounded FIFO channels as their method internally uses Petri nets.
Still, their branching property \cite[Def. 3.5]{DBLP:conf/popl/LangeTY15} consists of similar -- even though more restrictive -- conditions as our MST framework: a single role chooses at each branch but roles have to learn with the first received message or do not commit any action until the branches merge back.
We allow a role to learn later by recursive application of $\merge$.

\subparagraph*{Implementability in Message Sequence Charts.}
Projection is studied in hierarchical message sequence charts (HMSCs) as \emph{realisability}.
There, variations of the problem like changing the payload of existing messages or even adding new messages in the protocol are also considered \cite{DBLP:conf/stacs/Morin02,DBLP:journals/jcss/GenestMSZ06}.
Here, we focus on implementability without altering the protocol.
HMSCs are a more general model than MSTs and, unsurprisingly, realisability is undecidable~\cite{DBLP:conf/ac/GenestMP03,DBLP:journals/tcs/AlurEY05}.
Thus, restricted models have been studied.
Boundedness \cite{DBLP:journals/tcs/AlurEY05} is one such example:
checking safe realisability for bounded HMSCs is EXPSPACE-complete \cite{DBLP:journals/tcs/Lohrey03}.
Boundedness is a very strong restriction.
Weaker restrictions, as in MSTs, center on choice.
As we explain below, these restrictions are either flawed, overly restrictive, or not effectively checkable.

The first definition of (non-)local choice for HMSCs by Ben-Abdallah and Leue \cite{DBLP:conf/tacas/Ben-AbdallahL97} 
suffers from severely restrictive assumptions 
and only yields finite-state systems.
Given an HMSC specification, research on \emph{implied scenarios}, e.g. by Muccini et al.~\cite{DBLP:conf/fase/Muccini03}, investigates whether there are behaviours which, due to the asynchronous nature of communication, every implementation must allow in addition to the specified ones.
In our setting, an implementable protocol specification must not have any implied scenarios.
Mooij et al.~\cite{DBLP:conf/fase/MooijGR05} point out several contradictions of the observations on implied scenarios and non-local choice.
Hence, they propose more variants of non-local choices but allow implied scenarios.
In our setting, this corresponds to allowing roles to follow different branches. 

Similar to allowing implied scenarios of specifications, H{\'{e}}lou{\"{e}}t~\cite{DBLP:conf/sdl/Helouet01} pointed out that non-local choice has been frequently misunderstood as it actually does not ensure implementability but less ambiguity.
H{\'{e}}lou{\"{e}}t and Jard  proposed the notion of reconstructibility~\cite{NO-DBLP-wrong-local-choice} for a quite restrictive setting:
first, messages need to be unique in the protocol specification and, second, each node in an HMSC is implicitly terminal.
Unfortunately, we show their results are flawed. Consider the HMSC in \cref{fig:wrong-local-choice}.
(For simplicity, we use the same message identifier in each branch but one can easily index them for uniqueness.)
The same protocol can be represented by the following global type:

\vspace{0.3mm}
{ \small
\hspace{5mm}
$
\mu t. \;
+ \;
\begin{cases}
    \msgFromTo{\procB}{\procE}{{\color{orange}l}}. \;
    \msgFromTo{\procB}{\procA}{{\color{orange}l}}. \;
    \msgFromTo{\procE}{\procD}{{\color{orange}l}}. \;
    \msgFromTo{\procD}{\procC}{{\color{orange}l}}. \;
    \msgFromTo{\procC}{\procA}{{\color{orange}l}}. \; t
    \\
    \msgFromTo{\procB}{\procE}{{\color{purple}r}}. \;
    \msgFromTo{\procB}{\procA}{{\color{purple}r}}. \;
    \msgFromTo{\procE}{\procC}{{\color{purple}r}}. \;
    \msgFromTo{\procC}{\procA}{{\color{purple}r}}. \; t
\end{cases}
$
}
\vspace{0.3mm}

Because their notion of reconstructibility \cite[Def.~12]{NO-DBLP-wrong-local-choice} only considers loop-free paths, they report 
that the HMSC is reconstructible.
However, the HMSC is not implementable. 
Suppose that $\procB$ first chooses to take the top (resp.\ left) and then the bottom (resp.\ right) branch.
The message $l$ from $\procD$ to $\procC$ can be delayed  until after $\procC$ received $r$ from $\procE$.
Therefore, $\procC$ will first send $r$ to $\procA$ and then $l$ which contradicts with the order of branches taken.
This counterexample contradicts their result \cite[Thm.~15]{NO-DBLP-wrong-local-choice} and shows that reconstructibility is not sufficient for~implementability.

Dan et al.\ \cite{DBLP:conf/sefm/DanHC10}, improving Baker et al.~\cite{DBLP:conf/sigsoft/BakerBJKTMB05},  provide 
a definition that ensures implementability.
They provide a definition which is based on projected words of the HMSC in contrast to the choices. It is unknown how to check their condition for HMSCs.

\subparagraph*{CSMs and MSTs.}
The connection between MSTs, CSMs, and automata \cite{DBLP:journals/corr/abs-1203-0780,DBLP:conf/esop/DenielouY12}
came shortly after the introduction of MSTs.
Denielou and Yoshida \cite{DBLP:conf/esop/DenielouY12} use CSMs but they preserve the restrictions on choice from MSTs.
It is well-known that CSMs are Turing-powerful \cite{DBLP:journals/jacm/BrandZ83}.
Decidable instances of CSM verification can be obtained
by restricting the communication topology \cite{DBLP:journals/acta/PengP92, DBLP:conf/tacas/TorreMP08} or
by altering the semantics of communication, e.g. by making channels
lossy \cite{DBLP:conf/cav/AbdullaBJ98}, half-duplex \cite{DBLP:journals/iandc/CeceF05}, or input-bounded \cite{DBLP:conf/concur/BolligFS20}.
Lange and Yoshida \cite{DBLP:conf/cav/LangeY19} proposed additional notions that resemble ideas from MSTs.
 \section{Conclusion}

We have presented a generalised projection operator for asynchronous MSTs.
The key challenge lies in the generalisation of the external choice to allow roles to receive from more than one sender.
We provide a new projectability check and a soundness theorem that shows projectability implies implementability.
The key to our results is a message causality analysis and an automata-theoretic soundness proof.
With a prototype implementation, we have demonstrated that our MST framework can project examples from the literature as well as new examples, including typical communication patterns in distributed computing, which were not projectable by previous projection operators.
 
\phantomsection\label{paper-last-page}
\bibliographystyle{plain}

\clearpage
\appendix
\section{Communicating State Machines}
\label{app:csms}

A \emph{communicating state machine} (CSM) $\mathcal{A} = \CSM{A}$
over $\Procs$ and $\MsgVals$ consists of a state machine
${A}_\procA$
over $\Sigma_\procA$ for each $\procA\in\Procs$.
A state machine for $\procA$ will be denoted by $(Q_\procA, \Sigma_\procA, \delta_\procA, q_{0, \procA}, F_\procA)$.
If a state $q$ has multiple outgoing transitions, all labelled with send actions, then $q$ is called an \emph{internal choice} state.
If all the outgoing transitions are labelled with receive actions, $q$ is called an \emph{external choice} state.
Otherwise, $q$ is a \emph{mixed choice} state.
In this paper, we only consider state machines without mixed choice states.

Intuitively, a CSM represents a set of state machines, one for each role in $\Procs$,
interacting via message sends and receipts.
We assume that each pair of roles $\procA, \procB\in \Procs$, $\procA \neq\procB$, is connected
by a \emph{message channel}.
A transition $q_\procA \xrightarrow{\snd{\procA}{\procB}{\val}} q'_\procA$ in the state machine of $\procA$ specifies
that, when $\procA$ is in the state $q_\procA$, it sends a message $\val$ to $\procB$, and updates its local state to $q'_\procA$.
The message $\val$ is appended to the channel $\channel{\procA}{\procB}$.
Similarly, a transition $q_\procB \xrightarrow{\rcv{\procA}{\procB}{\val}} q'_\procB$ in the state machine of $\procB$
specifies that $\procB$ in state $q_\procB$ can retrieve the message $\val$ from the head of the channel $\channel{\procA}{\procB}$ and update
its local state to $q'_\procB$.

Let $\channels = \set{\channel{\procA}{\procB} \mid \procA,\procB\in \Procs, \procA\neq \procB}$ denote the set of channels.
The set of global states of the CSM is given by $\prod_{\procA \in \Procs} Q_\procA$.
For a global state $q$, we write $q_\procA$ for the state of $\procA$ in~$q$.
A \emph{configuration} of $\mathcal{A}$ is a pair $(q, \xi)$, where $q$ is a global state and
$\xi : \channels \rightarrow \MsgVals^*$ maps each channel to the queue of messages currently in the channel.
The initial configuration is $(q_0, \xi_\emptystring)$, where $q_{0,\procA}$ is the initial state of $A_\procA$ for each $\procA\in\Procs$ and
$\xi_\emptystring$ maps each channel to $\emptystring$.
A~configuration $(q, \xi)$ is \emph{final} iff $q_\procA$ is final for every $\procA$ and $\xi = \xi_\emptystring$.

In a global move of a CSM, a single role executes a local transition to change its state, while all other roles remain stationary.
For a send or a receive action, the corresponding channel is updated, while all other channels remain unchanged.
Formally, the global transition relation $\rightarrow$ on configurations is defined as follows:
\vspace{-2ex}
\begin{itemize}
\item
$(q,\xi) \xrightarrow{\snd{\procA}{\procB}{\val}} (q',\xi')$ if
$(q_\procA, \snd{\procA}{\procB}{\val}, q'_\procA)\in\delta_\procA$,
$q_\procC = q'_\procC$ for every $\procC \neq \procA$,
$\xi'(\channel{\procA}{\procB}) =  \xi(\channel{\procA}{\procB})\cdot\val$ and $\xi'(c) = \xi(c)$ for every other channel $c\in \channels$.

\item
$(q,\xi) \xrightarrow{\rcv{\procA}{\procB}{\val}} (q',\xi')$ if
$(q_\procB, \rcv{\procA}{\procB}{\val}, q'_\procB)\in\delta_\procB$,
$q_\procC = q'_\procC$ for every $\procC \neq \procB$,
$\xi(\channel{\procA}{\procB}) = \val\cdot \xi'(\channel{\procA}{\procB})$
and $\xi'(c) = \xi(c)$ for every other channel $c\in \channels$.

\item
$(q,\xi) \xrightarrow{\tau} (q',\xi)$ if
$(q_\procA, \emptystring, q'_\procA)\in\delta_\procA$ for some role
$\procA$, and
$q_\procB = q'_\procB$ for every role $\procB \neq \procA$.
\end{itemize}
A run of the CSM is a finite or infinite sequence: $
(q_0, \xi_0) \xrightarrow{x_0} (q_1, \xi_1) \xrightarrow{x_1} \ldots, 
$ such that $(q_0,\xi_0)$ is the initial configuration and for each $i \geq 0$, we have  $(q_i,\xi_i) \xrightarrow{x_i} (q_{i+1},\xi_{i+1})$.
The \emph{trace} of the run, written $\trace_{\mathcal{A}}(\rho)$, is the finite or infinite word $x_0x_1\ldots\in \Sigma^\infty$.
We also call $x_0 x_1\ldots$ an \emph{execution prefix}.
We may omit the subscript $\mathcal{A}$ when clear from context.
A run is maximal if it is infinite or if it is finite and ends in a final configuration. A trace is maximal if it is the trace of a maximal run.
The language $\lang(\mathcal{A})$ of the CSM $\mathcal{A}$ is the set of maximal traces.
A CSM is \emph{deadlock~free} if every finite run can be extended to a maximal~run.

\subparagraph*{Notations.}
From now on, we fix $\Alphabet = \Alphabet_\semglobal(\GG)$.
Let $w \in \Alphabet$.
We write $w\wproj_{\snd{\procA}{\procB}{\_}}$  and $w\wproj_{\rcv{\procB}{\procA}{\_}}$ to denote the projection of $w$ onto the send actions from $\procA$ to $\procB$ and receive actions of $\procA$ from $\procB$, respectively.
We write $\MsgVals(w)$ to project the send and receive actions in $w$ onto their message values.

Equipped with this notation, we define conditions that specify correct executions of a distributed message-passing system.

\begin{definition}\begin{enumerate}
\item \emph{\Channelcompliant}:
    A word $w\in \Sigma^\infty$ is \channelcompliant if messages are received after they are sent and, between two processes,
    the reception order is the same as the send order.
    Formally, for each prefix $w'$ of $w$, we require
    $\MsgVals(w'\wproj_{\rcv{\procA}{\procB}{\_}})$ to be a prefix of
    $\MsgVals(w'\wproj_{\snd{\procA}{\procB}{\_}})$, for every $\procA,\procB \in \Procs$.
\item \emph{Complete}:
    A \channelcompliant word $w\in\Sigma^\infty$ is \emph{complete} if it is infinite or the send and receive events match:
    if $w \in \Sigma^*$, then $\MsgVals(w\wproj_{\snd{\procA}{\procB}{\_}}) = \MsgVals(w\wproj_{\rcv{\procA}{\procB}{\_}})$ for every $\procA,\procB \in \Procs$.
\end{enumerate}
\end{definition}

The following lemma summarises that execution prefixes of CSMs satisfy these conditions.

\begin{lemma}
\label{lm:execution-prefix-and-channel-content}
Let $\CSM{A}$ be a CSM.
For any run $(q_0, \xi_0) \xrightarrow{x_0} \cdots \xrightarrow{x_n} (q,\xi)$ with trace $w = x_0\ldots x_n$,
it holds that (1) $\xi(\channel{\procA}{\procB}) = u$
where $\MsgVals(w \wproj_{\snd{\procA}{\procB}{\_}}) = \MsgVals(w \wproj_{\rcv{\procA}{\procB}{\_}}).u$
for every pair of roles $\procA, \procB \in \Procs$
and
(2) $w$ is \channelcompliant.
Maximal traces of $\CSM{A}$ are \channelcompliant and complete.
\label{lm:csm-channelcompliant-traces-only}
\end{lemma}
\begin{proof}
We prove the claims by induction on an execution prefix $w$.
The base case where $w = \emptystring$ is trivial.
For the induction step, we consider $wx$ with the following run in $\CSM{A}$:
$
 (q_0, \xi_0) \redtoover{w}
 (q, \xi) \redtoover{x}
 (q', \xi').
$
The induction hypothesis holds for $w$ and $(q, \xi)$ and we prove the claims for $(q', \xi')$ and $wx$.
We do a case analysis on $x$.
If $x = \tau$, the claim trivially follows.

Let $x = \rcv{\procA}{\procB}{\val}$.
From the induction hypothesis, we know that $\xi(\channel{\procA}{\procB}) = u $
where $\MsgVals(w \wproj_{\snd{\procA}{\procB}{\_}}) =
\MsgVals(w \wproj_{\rcv{\procA}{\procB}{\_}}).u$.
Since $x$ is a possible transition, we know that $u = \val.u'$ for some~$u'$ and $\xi'(\channel{\procA}{\procB}) = u'$.
By definition, it holds that
$\MsgVals(w \wproj_{\rcv{\procA}{\procB}{\_}}).\val.u'
= \MsgVals((wx) \wproj_{\rcv{\procA}{\procB}{\_}}).u'$.
For all other pairs of roles, the induction hypothesis applies since the above projections do not change.
Hence, $wx$ is \channelcompliant.

Let $x = \snd{\procA}{\procB}{\val}$.
From the induction hypothesis, we know that $\xi(\channel{\procA}{\procB}) = u $
where $\MsgVals(w \wproj_{\snd{\procA}{\procB}{\_}}) =
\MsgVals(w \wproj_{\rcv{\procA}{\procB}{\_}}).u$.
Since $x$ is a possible transition, we know that $\xi'(\channel{\procA}{\procB}) = u.\val$.
By definition and induction hypothesis, we have:
$
\MsgVals((wx) \wproj_{\snd{\procA}{\procB}{\_}}) =
 \MsgVals(w \wproj_{\snd{\procA}{\procB}{\_}}).\val =
\MsgVals(w \wproj_{\rcv{\procA}{\procB}{\_}}).u.\val.
$
For all other combinations of roles, the induction hypothesis applies since the above projections do not change.
Hence, $wx$ is \channelcompliant.

From (1) and (2), it follows directly that maximal traces of $\CSM{A}$ are \channelcompliant and complete.
\end{proof}
 \section{Properties of $\interswaplang$}
\label{app:props-interswaplang}

\begin{lemma}
\label{lm:proc-lang-closed-interswap}
Let $L \subseteq \Alphabet_{\procA}^\infty$.
Then, $L = \interswaplang(L)$.
\end{lemma}

\begin{proof}
For any $w\in \Alphabet_{\procA}^\infty$,  none of the rules of $\interswap_1$ applies to $w$, and we have that
$w \interswap w'$ iff $w = w'$.
Thus, $L = \interswaplang(L)$ for any language $L \subseteq \Alphabet_{\procA}^\infty$.
\end{proof}

\begin{lemma}
\label{lm:csm-closed-under-interswap}
Let $\CSM{A}$ be a CSM.
Then, for every finite $w$ with a run in $\CSM{A}$ and every $w' \interswap w$, $w'$ has a run that ends with the same configuration.
The language $\lang(\CSM{A})$ is closed under $\interswap$: $\lang(\CSM{A}) = \interswaplang(\lang(\CSM{A}))$.
\end{lemma}

\begin{proof}

Let $w$ be a finite word with a run in $\CSM{A}$ and $w' \interswap w$.
By definition, $w' \interswap_n w$ for some $n$.
We prove that $w'$ has a run that ends in the same configuration by induction on $n$.
The base case for $n = 0$ is trivial.
For the induction step, we assume that the claim holds for $n$ and prove it for $n+1$.
Suppose that $w \interswap_{n+1} w'$.
Then, there is $w''$ such that $w' \interswap_1 w''$ and $w'' \interswap_n w$.
By assumption, we know that $w' = u' u_1 u_2 u''$ and $w'' = u' u_2 u_1 u''$ for some $u', u'' \in \Alphabet^*$, $u_1, u_2 \in \Alphabet$.
By induction hypothesis, we know that $w'' \in \lang(\CSM{A})$ so there is run for $w''$ in $\CSM{A}$.
Let us investigate the run at $u_1$ and $u_2$:
$
 \cdots
 (q_1, \xi_1)
    \redtoover{u_1}
 (q_2, \xi_2)
    \redtoover{u_2}
 (q_3, \xi_3)
 \cdots.
$
It suffices to show that
$
 \cdots
 (q_1, \xi_1)
    \redtoover{u_2}
 (q'_2, \xi'_2)
    \redtoover{u_1}
 (q_3, \xi_3)
 \cdots
$
is possible in $\CSM{A}$ for some configuration $(q'_2, \xi'_2)$.
We do a case analysis on the rule that was applied for $w' \interswap_1 w''$.
\vspace{-2ex}
\begin{itemize}
\item
 $u_1 = \snd{\procA}{\procB}{\val}$,
 $u_2 = \snd{\procC}{\procD}{\val'}$, and
 $\procA ≠ \procC$:

 We define $q'_2$ such that
 $q'_{2,\procA} = q_{1,\procA}$,
 $q'_{2,\procC} = q_{3,\procC}$,
 and
 $q'_{2,\procE} = q_{3,\procE}$
 for all $\procE \in \Procs$ with $\procE \neq \procA$ and $\procE \neq \procC$.
 Both transitions are feasible in $\CSM{A}$ because both $\procA$ and $\procC$ are different and send a message to different channels.
 They can do this independently from each other.
\item
 $u_1 = \rcv{\procA}{\procB}{\val}$,
 $u_2 = \rcv{\procC}{\procD}{\val'}$, and
 $\procB ≠ \procD$:

 We define $q'_2$ such that
 $q'_{2,\procB} = q_{1,\procB}$,
 $q'_{2,\procD} = q_{3,\procD}$,
 and
 $q'_{2,\procE} = q_{3,\procE}$
 for all $\procE \in \Procs$ with $\procE \neq \procA$ and $\procE \neq \procC$.
 Both transitions are feasible in $\CSM{A}$ because both $\procB$ and $\procD$ are different and receive a message from a different channel.
 They can do this independently from each other.
\item
 $u_1 = \snd{\procA}{\procB}{\val}$,
 $u_2 = \rcv{\procC}{\procD}{\val'}$, and
 and $\procA ≠ \procD \land (\procA ≠ \procC ∨ \procB ≠ \procD)$.

 We define $q'_2$ such that
 $q'_{2,\procA} = q_{1,\procA}$,
 $q'_{2,\procD} = q_{3,\procD}$,
 and
 $q'_{2,\procE} = q_{3,\procE}$
 for all $\procE \in \Procs$ with $\procE \neq \procA$ and $\procE \neq \procC$.
 Let us do a case split according to the side conditions.
 First, let $\procA \neq \procD$ and $\procA \neq \procC$.
 The channels of $u_1$ and $u_2$ are different and $\procA$ and $\procD$ are different, so both transitions are feasible in $\CSM{A}$.

 Second, let $\procA \neq \procD$ and $\procB \neq \procD$.
 The channels of $u_1$ and $u_2$ are different and $\procB$ and $\procD$ are different, so both transitions are feasible in $\CSM{A}$.
\item
 $u_1 = \snd{\procA}{\procB}{\val}$,
 $u_2 = \rcv{\procA}{\procB}{\val'}$,
 and $\card{u' \wproj_{\snd{\procA}{\procB}{\_}}} >
   \card{u' \wproj_{\rcv{\procA}{\procB}{\_}}}$:

 We define $q'_2$ such that
 $q'_{2,\procA} = q_{1,\procA}$,
 $q'_{2,\procB} = q_{3,\procB}$,
 and
 $q'_{2,\procE} = q_{3,\procE}$
 for all $\procE \neq \procA$ and $\procE \neq \procB$.

 In this case, the channel of $u_1$ and $u_2$ is the same but the side condition ensures that $u_2$ actually has a different message read since the channel $\xi_1(\channel{\procA}{\procB})$ is not empty by
 \cref{lm:execution-prefix-and-channel-content}
 and, hence, both transitions can act independently and lead to the same configuration.
\end{itemize}

This proves that $w'$ has a run in $\CSM{A}$ that ends in the same configuration which concludes the proof of the first claim.

For the second claim, we know that $\lang(\CSM{A}) \subseteq \interswaplang(\lang(\CSM{A}))$ by definition.
Hence, it suffices to show that
$\interswaplang(\lang(\CSM{A})) \subseteq \lang(\CSM{A})$.

We show the claim for finite traces first:
\[
    \interswaplang(\lang(\CSM{A})) \inters \Alphabet^*
    \subseteq
    \lang(\CSM{A}) \inters \Alphabet^*.
\]
Let $w' \in \interswaplang(\lang(\CSM{A})) \; \inters \; \Alphabet^*$.
There is $w \in \lang(\CSM{A}) \; \inters \; \Alphabet^*$ such that $w \interswap w'$.
By definition, $w$ has a run in $\CSM{A}$ which ends in a final configuration.
From the first claim, we know that $w'$ also has a run that ends in the same configuration which is final.
Therefore, $w \in \lang(\CSM{A}) \inters \Alphabet^*$.
Hence, the claim holds for finite traces.

It remains to show the claim for infinite traces.
To this end, we show that for every execution prefix $w$ of $\CSM{A}$ such that $w \interswap u$ for $u \in \pref(\lang(\CSM{A}))$ and any continuation $x$ of $w$, i.e., $wx$ is an execution prefix of $\CSM{A}$, it holds that $wx \interswap ux$ and $ux \in \pref(\lang(\CSM{A}))$ $(\square)$.
We know that $w \interswap_n u$ for some $n$ by definition, so $wx \interswap_n ux$ since we can mimic the same $n$ swaps when extending both $w$ and $u$ by $x$.
From the first claim, we know that $\CSM{A}$ is in the same configuration $(q, \xi)$ after processing $w$ and~$u$.
Therefore, $ux$ is an execution prefix of $\CSM{A}$ because $wx$~ is which yields $(\square)$.

We show that
\[
    \interswaplang(\lang(\CSM{A}) \inters \Alphabet^\omega
    \subseteq
    \lang(\CSM{A}) \inters \Alphabet^\omega.
\]
Let $w \in \interswaplang(\lang(\CSM{A}) \, \inters \, \Alphabet^\omega$.
We show that $w$ has an infinite run in $\CSM{A}$.

Consider a tree $\mathcal{T}$ where each node corresponds to a run $\rho$ on some finite prefix $w' \leq w$ in $\CSM{A}$.
The root is labelled by the empty run.
The children of a node $\rho$ are runs that extend $\rho$ by a single transition --- these exist by $(\square)$.
Since our CSM, derived from a global type, is built from a finite number of finite state machines, $\mathcal{T}$ is finitely branching.
By König's Lemma, there is an infinite path in~$\mathcal{T}$ that corresponds to an infinite run for $w$ in $\CSM{A}$, so $w \in \lang(\CSM{A}) \inters \Alphabet^\omega$.
\end{proof}

\begin{lemma}
\label{lm:interswap-complete-wrt-channelcompliancy}
Let $w \in \Alphabet^\infty$ be \channelcompliant.
Then, $w \interswap w'$
iff $w'$ is \channelcompliant and
$w \wproj_{\Alphabet_\procA} = w' \wproj_{\Alphabet_\procA}$ for all roles $\procA \in\Procs$.
\end{lemma}
\begin{proof}
We use the characterisation of $\interswap$ using dependence graphs \cite{DBLP:conf/litp/Gastin90}.
For a word $w$ and a letter $a \in \Sigma$ that appears in $w$, let
$(a,i)$ denote the $i$th occurrence of $a$ in $w$.  Define the dependence graph $(V, E, \lambda)$, where
$V = \set{(a, i) \mid a\in\Sigma, i \geq 1}$,
$E = \set{((a,i), (b,j)) \mid \mbox{$a$ and $b$ cannot be swapped and
    $(a,i)$ occurs before  $(b,j)$ in $w$}}$,
and
$\lambda(a, i) = a$ for all $a\in\Sigma$, $i\geq 1$.
A fundamental result of trace theory states that $w \interswap w'$ iff they have isomorphic dependence graphs.
We observe that for two channel compliant words, the ordering of the letters on each $\Sigma_\procA$ for $\procA\in\Procs$
ensures isomorphic dependence graphs, since the ordering of receives is thus fixed.
\end{proof}

 \section{Implementing \Projectable Types: Proofs}
\label{sec:local-proof}

Let us start with some structural properties of $\semglobal(G)$ and $\semlocal(L)$.

\begin{proposition}[Shape of $\semglobal(G)$ for global type $G$ and $\semlocal(L)$ for local type~$L$]
\label{prop:shape-semantics-local-L}
\label{prop:shape-semantics-global-G}
Let $G$ be a global type and $L$ a local type.
Then, the following holds:
\vspace{-2ex}
\begin{enumerate}[(a)]
 \item every state $q$ in $\semglobal(G)$ or $\semlocal(L)$ either only has non-$\emptystring$-transitions or a single $\emptystring$-transition;
 \item the state that corresponds to the syntactic subterm $0$, if it exists, is the only state without any outgoing transitions.
 \item for each word $w$, there are at most three runs $\rho$, $\rho_1$, and $\rho_2$ in $\semglobal(G)$; moreover, when they exist,
the runs can be chosen such that $\rho_1 := \rho \xrightarrow{\emptystring} q$ and $\rho_2 := \rho_1 \xrightarrow{\emptystring} q'$ for some $q$,~$q'$.
\end{enumerate}
\end{proposition}
\begin{proof}
Fact (a) follows from the definition of $\semlocal(L)$ and $\semglobal(G)$ and the syntax restriction that every variable is only bound once.
Fact (b) also follows directly from the definition of $\semlocal(L)$ and $\semglobal(G)$.
Let us prove Fact (c).
For $\semlocal(G)$ branches occur for sends only.
By definition of the syntax for global types, these send options are always distinct.
Hence, the runs are unique up to $\emptystring$-transitions.
Such $\emptystring$-transitions can only occur in the presence of recursion.
By assumption, recursion is always guarded and hence the number of consecutive $\emptystring$-transitions is bounded by~$2$.
\end{proof}

\subsection{Projection Preserves Behaviours}
\label{sec:global-type-subset-CSM-projected}

\begin{lemma}
\label{lm:pref-MST-has-run}
\label{lm:projections-yield-superset}
Let $\GG$ be a \projectable global type.
Then, the following holds:
\vspace{-1ex}
\begin{enumerate}[(1)]
 \item for every prefix $w \in \pref(\lang(\GG))$, there is a run $\rho$ in $\CSMl{\semlocal(\GG \tproj_\procA)}$ and for every extension $wx$ of $w$ in $\pref(\lang(\GG))$, $\rho$ can be extended
 \item $\lang(\GG) \subseteq \lang(\CSMl{\semlocal(\GG \tproj_\procA)})$.
\end{enumerate}
\end{lemma}

\begin{proof}
For (1), let $w \in \pref(\lang(\GG))$.
We prove the claim by induction on the length of $w$.

The base case, $w = \emptystring$, is trivial.
For the induction step, we assume that the claim holds for $w$ with run $\rho$ and we extend $\rho$ for $wx \in \pref(\lang(\GG))$.
We do a case analysis on $x$.

First, suppose that $x = \snd{\procA}{\procB}{\val}$ for some $\procA$, $\procB$, and $\val$.
By \cref{prop:projection-preserves-per-process-runs}, we know that $(wx) \wproj_{\Alphabet_\procA}$ has a run in $\semlocal(\GG \tproj_\procA)$.
Therefore, we can simply extend the run $\rho$ in $\CSMl{\semlocal(\GG \tproj_\procA)}$.
Second, suppose that $x = \rcv{\procA}{\procB}{\val}$ for some $\procA$, $\procB$, and $\val$.
Again, by \cref{prop:projection-preserves-per-process-runs}, we know that $(wx) \wproj_{\Alphabet_\procB}$ has a run in $\semlocal(\GG \tproj_\procB)$.
Any prefix of $\lang(\GG)$ is \channelcompliant by construction.  From \cref{lm:execution-prefix-and-channel-content}, we know that $\channel{\procA}{\procB}$ is of the form $m \cdot u$.  Therefore $\procB$ can receive $\val$ from $\procA$ and we can extend $\rho$ for $wx$ in $\CSMl{\semlocal(\GG \tproj_\procA)}$.

In both cases, we have extended $\rho$ on $w \in \pref(\lang(\GG))$ to $wx \in \pref(\lang(\GG))$, completing the argument for (1).

For (2), let $w \in \lang(\GG)$.
We do a case analysis whether $w$ is finite or infinite.

For finite words $w \in \lang(\GG) \inters \Alphabet^*$, we know that $w$ is \channelcompliant and complete.
(1)~tells us that there is a run for $w$ in $\CSMl{\semlocal(\GG \tproj_\procA)}$ that reaches $(q, \xi)$ for some state $q$ and channel evaluation $\xi$.
We need to show that $(q, \xi)$ is a final configuration, i.e. (a) $q_\procA$ is a final state in $\semlocal(\GG \tproj_\procA)$ and (b) all channels in $\xi$ are empty.

To establish (b), we observe that all channels are empty after a run of a finite and complete word $w$, by \cref{lm:execution-prefix-and-channel-content}. The fact that $w$ is complete follows from \cref{def:language-global-mst}.

For (a), recall that $w \in \semglobal(\GG)$ and $w$ is finite, so the run with trace $w$ in $\semglobal(\GG)$ ends in a final state.
By \cref{prop:shape-semantics-global-G}, every final state corresponds to $0$ and has no outgoing transitions.
The projection of $0 \tproj_\procA$ is defined as $0$ for every $\procA$ and $0$ is never merged with local types different from $0$.
The definition of local types prescribes that all different send and reception options are distinct at all times.
This implies that $\semlocal(\GG \tproj_\procA)$ is deterministic (if it exists which is given by assumption).
Therefore, each $q_\procA$ corresponds to the local type $0$ and is hence final in $\semlocal(\GG \tproj_\procA)$.
Thus, $(q, \xi)$ is a final configuration of $\CSMl{\semlocal(\GG \tproj_\procA)}$, so $w \in \lang(\CSMl{\semlocal(\GG \tproj_\procA)})$.

Suppose that $w$ is infinite, i.e.,  $w \in \lang(\GG) \inters \Alphabet^\omega$.
We show that $w$ has an infinite run in $\CSMl{\semlocal(\GG \tproj_\procA)}$.
Consider a tree $\mathcal{T}$ where each node corresponds to a run $\rho$ on some finite prefix $w' \leq w$.
The root is labelled by the empty run.  The children of a node $\rho$ are runs that extend $\rho$ by a single transition---these exist by part (1) of the lemma.
Since our CSM, derived from a global type, is a finite number of finite state machines, $\mathcal{T}$ is finitely branching.
By König's Lemma, there is an infinite path in $\mathcal{T}$ that corresponds
to an infinite run for $w$ in $\CSMl{\semlocal(\GG \tproj_\procA)}$, so $w \in \lang(\CSMl{\semlocal(\GG \tproj_\procA)})$.
\end{proof}
 
\subsection{Projection does Not Introduce New Behaviours}
\label{sec:run-mapping}

\begin{definition}[Run mappings]
Let $\GG$ be a global type and $\CSM{A}$ be a CSM.
For an execution prefix $w$ of $\CSM{A}$ and a run $\rho$ in $\semglobal(\GG)$ such that $w$ is a prefix of some $w'$ with $w' \interswap \trace_{\semglobal(\GG)}(\rho)$,
a \emph{run mapping} $\rho_{\scriptscriptstyle \Procs} \from \Procs \to \pref(\rho)$ is a mapping such that
$w \wproj_{\Alphabet_\procA} = \trace_{\semglobal(\GG)}(\rho_{\scriptscriptstyle \Procs}(\procA)) \wproj_{\Alphabet_\procA}$.

We say that $\CSM{A}$ has a \emph{family of run mappings} for $\semglobal(\GG)$
iff
for every execution prefix $w$ of $\CSM{A}$,
there is a \emph{witness run} $\rho$ in $\semglobal(\GG)$ with a run mapping $\rho_{\scriptscriptstyle \Procs}$.
\end{definition}
When clear from context, we omit the subscript of run mappings~$\rho_{\scriptscriptstyle \Procs}(\hole)$.

Given a global type $\GG$ and an execution prefix of $\CSM{A}$, a run mapping guarantees that there is some common run in $\semglobal(\GG)$ that all roles could have pursued when observing this execution prefix.
They might not have processed all their events along this run, but they have not diverged.
A family of run mappings lifts this to all execution prefixes.
However, this does not yet ensure that every role will be able to follow this run to the end.
This will be part of \emph{Control Flow Agreement}.

\begin{lemma}[$\set{\semlocal(\GG \tproj_\procA)}_{\procA \in \Procs}$ has a family of run mappings]
\label{lm:proj-types-have-run-mappings}
Let $\GG$ be a \projectable global type.
Then, $\CSMl{\semlocal(\GG \tproj_\procA)}$ has a family of run mappings for $\semglobal(\GG)$.
\end{lemma}
The proof is not trivial because the merge can collapse different branches for a role so it follows sets of runs.
The intersection of all these sets contains (a prefix of) the witness~run.
In some case, e.g., $\procC$ in the previous example, a role can ``overtake'' a choice and perform actions that belong to the branches of a choice before the choosing process has made the~choice.

\subsubsection*{Outline and Motivation of Concepts}

As outlined in the main text, the overall proof works by induction on the execution prefix.
While the skeleton is quite basic, let us give an intuition which properties are needed for the cases where some role receives or sends a message in the extension of an execution prefix.
Based on these properties, we can explain the different concepts which need to be formalised for these properties.
In both cases, we need to argue that the previous run mapping can be kept for all roles but the one that takes a step.

For the receive case, suppose that $\procA$ receives from $\procB$.
On the one hand, we need to argue that the run for $\procA$ can actually be extended to match the one of of $\procB$.
On the other hand, we also need to ensure that there is no other message that $\procA$ could receive so that it follows a run different from the one for $\procB$.
While the first is easy to obtain by the definition of projection, we need to argue that the message availability check does indeed compute the available messages when $\procA$ has not proceeded with any other event yet.

For the send case, suppose that $\procA$ sends to $\procB$ and chooses between different branches, i.e., there actually is a choice between different continuations.
Though, there might be roles with (some) actions that do not depend on this choice and we have to ensure that $\procA$'s run can be extended to some prefix of theirs.
For this, we exploit the idea of prophecy variables and assume that all of those roles took the branch $\procA$ will choose but we need to ensure that this does not restrain them in any way.
So intuitively for each of these roles, we need to compare the executions along the different branches $\procA$ can choose from and need to prove that they are the same if $\procA$ has not yet committed this choice yet.
(Note the subtlety that only the part of the execution which does not depend on this choice needs to be the same.)

To formalise these ideas, we introduce three concepts.
\vspace{-2ex}
\begin{itemize}
 \item \emph{Blocked languages:}
        We inductively define the (global) language $L_{(G \ldots)}^{\mathcal{B}}$ that captures all executions that are possible from subterm $G$ in some global type~$\GG$, under the assumption that no role in $\mathcal{B} \subseteq \Procs$ can take any further step.
        We might call these roles \emph{blocked}.
        This definition serves for the formalisation of semantic arguments about the availability of messages and the languages along different choice branches.
 \item \emph{Blocked projection operator:}
        We introduce a variation of the projection operator that also allows to account for a set of blocked roles $\mathcal{B}$, prove its connection to the above languages, and use the relation to the standard projection operator to show the equality of the ``blocked'' languages along different choices.
 \item \emph{Extended local types:}
        In contrast to the standard projection operator, the blocked one can declare intermediate syntactic subterms as final states.
        Standard local types do not allow this so we introduce extended local types.
\end{itemize}

\subsubsection*{Blocked Languages -- The Language-theoretic Perspective}

We will use blocked languages to show that no role can proceed with an action that is specific to some choice which has not been committed yet.

\begin{definition}[Blocked language along path]
Let $\GG$ be a \projectable global type and $\mathcal{B} \subseteq \Procs$ a set of roles. Then, $L^{\mathcal{B}}_{(G\ldots)}$ is inductively defined:

\begin{footnotesize}
\noindent
$
L^{\mathcal{B}}_{(0 \ldots)} \is \set{\emptystring}
$
\hfill
$
L^{\mathcal{B}}_{(t \ldots)} \is L^\mathcal{B}_{(\mu t.G \ldots)}
$
\hfill
$
L^{\mathcal{B}}_{(\mu t.G \ldots)} \is L^\mathcal{B}_{(G \ldots)}
$\\
$
L^{\mathcal{B}}_{(\Sum_{i \in I} \msgFromTo{\procA}{\procB_i}{\val_i}.G_i \ldots)}
                \is
                \Union_{i \in I}
                    \set{
                    f(\mathcal{B}, \msgFromTo{\procA}{\procB_i}{\val_i}).
                    w
                    \mid
                    w \in L^{g(\mathcal{B}, \msgFromTo{\procA}{\procB_i}{\val_i})}_{(G_i \ldots)}
}
$
\hspace{0.5cm} where \\
$
f(\mathcal{B}, \msgFromTo{\procA}{\procB_i}{\val_i})
=
\begin{cases}
\snd{\procA}{\procB_i}{\val_i}.\rcv{\procA}{\procB_i}{\val_i}
& \text{if }
    \procA \notin \mathcal{B} \land
    \procB_i \notin \mathcal{B} \\
\snd{\procA}{\procB_i}{\val_i}
& \text{if }
    \procA \notin \mathcal{B} \land
    \procB_i \in \mathcal{B} \\
\emptystring
& \text{otherwise}
\end{cases}
$
\hfill
$
g(\mathcal{B}, \msgFromTo{\procA}{\procB_i}{\val_i})
=
\begin{cases}
\mathcal{B}
& \text{if }
    \procA \notin \mathcal{B} \\
\mathcal{B} \union
\set{\procB_i}
& \text{otherwise}
\end{cases}
$
\end{footnotesize}
Based on this definition, we compute the set of messages that can occur when some roles, i.e, the ones in $\mathcal{B}$ are blocked:
$
M^\mathcal{B}_{(G \ldots)} \is
    \Union_{x = x_1 \ldots \in L^\mathcal{B}_{(G \ldots)}} \set{x_i \mid x_i \in \Alphabet_r}.
$
\end{definition}

Note that we did not use the intermediate states but only the ones from $M(\GG)$ in this definition for clarity but mimic splitting the message exchanges faithfully.
We implicitly remove $\emptystring$, however, they are important so that languages are defined.

\subsubsection*{Extended Local Types and Blocked Projection Operator -- \\ The Type-theoretic Perspective}

\begin{definition}[Extended local types]
\label{def:extended-local-types}
We extend local types with marks $\final$ or $\nonfinal$ for syntactic subterms and obtain \emph{extended local types} with the following inductive rules where $\optionfn \in \set{\nonfinal, \final}$:
\vspace{-1ex}
    \begin{grammar}
     \mathit{\mathit{EL}} \is
            \epair{0}{\optionfn}
| \epair{\IntCh_{i ∈ I} \snd{}{\procB_i}{\val_i}.\mathit{EL}_i}{\optionfn}
| \epair{\ExtCh_{i ∈ I} \rcv{\procB_i}{}{\val_i}.\mathit{EL}_i}{\optionfn}
| \epair{μ t. \mathit{EL}}{\optionfn}
| \epair{t}{\optionfn}
\end{grammar}
\end{definition}

\vspace{-1ex}
Extended local types can be annotated with availability information about messages as local types.
When working with message set annotations, we assume to have extended availability annotated local types of shape $\epair{\apair{\hole}{\hole}}{\hole}$.
Thus, we can re-use the merge-operator~$\merge$.

The mark $\final$ indicates that this syntactic subterm shall be considered final in the semantics of the extended local types.
Any syntactic subterm $\epair{\mathit{EL}'}{\optionfn'}$ of $\epair{EL}{\optionfn}$ has a syntax tree.
Intuitively, $\optionfn'$ marks the root of the syntax tree of $\mathit{EL}'$ and tells whether one needs to continue, i.e.~$\epair{EL'}{\nonfinal}$, to obtain an accepting word or may stop, i.e.~$\epair{EL'}{\final}$, at this point.
Without such an extension, there is no way to specify such behaviour in local types (and global type) since $0$ explicitly marks the end of the type.
In the semantics, we consider two subterms equal if they agree on all (nested) markings.

\begin{definition}[Semantics of extended local types]
\label{def:language-extended-local-mst}
Let $\mathit{EL}$ be an extended local type for~$\procA$.
We construct a state machine $\semextlocal(\mathit{EL}) = (Q, \Sigma, δ, q₀, F)$ with
\vspace{-1ex}
\begin{itemize}
\item $Q$ is the set of all syntactic subterms in $\mathit{EL}$,
$Σ = \Alphabet_\procA$, 
$q₀ = \mathit{EL}$,
\item $F =
\set{\epair{\mathit{EL}'}{\final} \mid \epair{EL'}{\final}$ syntactic subterm of $EL}$, and
\item $δ$ is the smallest set containing
\vspace{-1ex}
\begin{itemize}
 \item $(\epair{\IntCh_{i ∈ I} \snd{}{\procB_i}{\val_i.\epair{\mathit{EL}_i}{\optionfn_i}}}{\optionfn'}, \snd{\procA}{\procB_i}{\val_i}, \epair{EL_i}{\optionfn_i})$ for every $i \in I$ where
 $\optionfn' \in \set{\final, \nonfinal}$ and
 $\optionfn_i \in \set{\final, \nonfinal}$ for each $i ∈ I$,
 \item $(\epair{\ExtCh_{i ∈ I} \rcv{\procB_i}{}{\val_i.\epair{\mathit{EL}_i}{\optionfn_i}}}{\optionfn'}, \rcv{\procB_i}{\procA}{\val_i}, \epair{EL_i}{\optionfn_i})$ for every $i \in I$ where
 $\optionfn' \in \set{\final, \nonfinal}$ and
 $\optionfn_i \in \set{\final, \nonfinal}$ for each $i ∈ I$, as well as
 \item $(\epair{μ t. \epair{\mathit{EL}'}{\optionfn'}}{\optionfn''}, ε, \epair{EL}{\optionfn'})$ and $(\epair{t}{\optionfn_t}, ε, \epair{μ t. \epair{EL'}{\optionfn'}}{\optionfn''})$ \\ for each subterm $\epair{\mu t.\epair{EL'}{\optionfn'}}{\optionfn''}$ and $\epair{t}{\optionfn_t}$ of $EL$ with $\optionfn_t, \optionfn', \optionfn'' \in \set{\final, \nonfinal}$.
\end{itemize}
\end{itemize}
We define the language of $\mathit{EL}$ as language of this state machine:
$\lang(\mathit{EL}) = \lang(\semextlocal(EL))$.
\end{definition}

It is straightforward to turn an extended (availability annotated) local type into a (availability annotated) local type by recursively projecting onto the first component of the extended local type and hence ignoring the markers $\final$ and $\nonfinal$.
Therefore, we apply this implicit coercion for operations defined on (availability annotated) local types only.

\begin{figure}[t]
\begin{subfigure}[b]{0.35\textwidth}
\begin{center}
\includegraphics[width=0.9\textwidth]{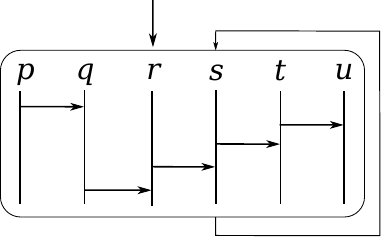}
\caption{To be Unfolded}
    \label{fig:unfolding-needed}
\end{center}
\end{subfigure}
\hfill
\begin{subfigure}[b]{0.4\textwidth}
\begin{center}
\includegraphics[width=0.8\textwidth]{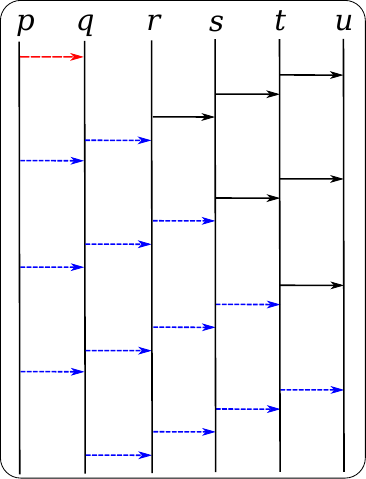}
\caption{Unfolded until all roles are blocked}
\label{fig:unfolding-needed-unfolded}
\end{center}
\end{subfigure}
\caption{An HMSC and some Unfolding}
\end{figure}

\begin{example}
Consider the HMSC in \cref{fig:unfolding-needed} and its unfolding in \cref{fig:unfolding-needed-unfolded}.
The unfolding exemplifies that one needs to unfold the recursion three times until all roles get blocked.
Still, if there were two more roles that would send messages to each other independently, one needs to account for these messages with a recursion.
\end{example}

The set of blocked roles does only increase.
If some message is blocked, two roles are blocked from the beginning and one iteration of the recursion is always present in any type.
In this case, we might need to unfold a recursion up to $\card{\Procs} - 3$ times to reach a fixed point of blocked roles.
In general, a fixed point is reached after $\card{\Procs} - 2$ unfoldings.

\begin{definition}[Blocked projection operator]
\label{def:special-proj}
Let $\GG$ be some global \projectable global type, $E$ a set of recursion variables, $\mathcal{B} \subseteq \Procs$ a set of roles and $n$ some natural number.
Let $T$ be a set of recursion variables and $\eta \from T \to \powerset(\Procs) \times \Nat$.
We define the blocked projection $\tproj^{E, \mathcal{B}}$ of $\GG$ parametrised by $n$ onto a role $\procC$ inductively:\\

\begin{minipage}{0.98\linewidth}
\begin{scriptsize}
\noindent
$
    0 \tproj^{E, \mathcal{B}}_{\procC, n} \is \epair{0 \tproj_\procC}{\final}
$
\hfill
$
  \left(\mu t. G'\right) \tproj^{E, \mathcal{B}}_{\procC, n}
  \is
 \begin{cases}
    \epair{
    \apair{μ t. (G \tproj^{E ∪ \set{t}}_\procC)}{\avail(\set{R},\set{t},G)}
    }{
    \nonfinal
    }
    \text{ and $\eta(t) \is (\mathcal{B}, 0)$} \\
        \hfill \text{if } G \tproj^{E ∪ \set{t}}_\procC \neq \epair{\apair{t}{\_}}{\_} ∧ t \text{ occurs in } G \tproj^{E ∪ \set{t}}_\procC\\
    \apair{(G \tproj^{E ∪ \set{t}}_\procC)}{\avail(\set{R},\set{t},G)}
        \qquad \text{if } t \text{ does not occur in } G \tproj^{E ∪ \set{t}}_\procC\\
    \epair{\apair{0}{∅}}{\final}
        \hfill \text{otherwise}
 \end{cases}
$

\[
    t \tproj^{E, \mathcal{B}}_{\procC, n} \is
    \begin{cases}
    \epair{
    \apair{
    \mu t'.((\getMuG(t))[t'/t]) \tproj^{E, \mathcal{B}}_{\procC,n}
    }{
    \avail(\set{\procC}, \set{t}, \getMuG(t))
    }
    }
    {\nonfinal}
    \text{ and }
    \eta(t') \is (\mathcal{B}, \eta_2(t) + 1) \\
        \hfill \text{if } \eta_1(t) \neq \mathcal{B} \land \eta_2(t) < n \\ \epair{
    \apair{
    t
    }{
    \avail(\set{\procC}, \set{t}, \getMuG(t))
    }
    }{
    \nonfinal} \\
        \hfill \text{if }\eta_1(t) = \mathcal{B} \land \eta_2(t) \leq n \\
    \text{undefined} \\
        \hfill \text{otherwise}
    \end{cases}
\]
\[
 \left( \sum_{i ∈ I} \msgFromTo{\procA}{\procB_i}{\val_i.G_i} \right)
 \tproj^{E, \mathcal{B}}_{\procC, n} \is
 \begin{cases}
    \epair{
        \apair{
            \IntCh_{i ∈ I} \snd{}{\procB_i}{\val_i}.(G_i \tproj^{E, \mathcal{B}}_{\procC, n})
        }{
            \bigcup_{i∈I} \avail(\set{R},∅,G_i)
        }
    }{\nonfinal} \hspace{1cm} \\
        \hfill \text{if } \procA \notin \mathcal{B} \land \procC = \procA \\
        \Mmerge
        \begin{cases}
            \epair{
            \apair{
                \ExtCh_{i \in I ∧ \procB_i = \procC ∧ \procC \notin \mathcal{B}} \rcv{\procA}{}{\val_i}.(G_i \tproj^{E, \mathcal{B}}_{\procC, n})
            }{
                \bigcup_{i∈I} \avail(\set{\procC},∅,G_i)
            }
            }{
            \nonfinal
            }\\
            \Mmerge_{i \in I ∧ \procB_i = \procC ∧ \procC \in \mathcal{B}} \;
            \epair{\apair{0}{\emptyset}}{\final} \\
\Mmerge_{i \in I ∧ \procB_i ≠ \procC ∧ ∀ t ∈ E.\, G_i \tproj^{E, \mathcal{B}}_{\procC, n} ≠ \epair{\apair{t}{\_}}{\_}} \;
            G_i \tproj_{\procC, n}^{E, \mathcal{B}}
        \end{cases} \\
        \hfill \text{if } \procA \notin \mathcal{B} \land \procC \neq \procA
        \\
    \epair{\apair{0}{\emptyset}}{\final}
        \hfill \text{if } \procA \in \mathcal{B} \land \procC = \procA \\
    \Mmerge_{i \in I}
    G_i \tproj^{E, \mathcal{B} \union \set{\procB_i}}_{\procC, n}
        \hfill \text{if } \procA \in \mathcal{B} \land \procC \neq \procA \\
 \end{cases}
\]
\end{scriptsize}
\end{minipage}
\\
where $\Mmerge$ is defined using $\Merge$:
\vspace{-2ex}
\begin{itemize}
 \item for any extended availability annotated local type
        $\epair{\apair{\mathit{EL}}{\mathit{Msg}}}{\optionfn}$, \\
       it holds that
$
            \epair{\apair{\mathit{EL}}{\mathit{Msg}}}{\optionfn}
                \mmerge
            \epair{\apair{0}{\emptyset}}{\final}
                \is
            \epair{\apair{\mathit{EL}}{\mathit{Msg}}}{\final}
$
for
       $\optionfn \in \set{\final, \nonfinal}$, and
 \item  for any two extended availability annotated local types
            $\epair{\apair{\mathit{EL}_1}{\mathit{Msg}_1}}{\optionfn_1}$
            and
            $\epair{\apair{\mathit{EL}_2}{\mathit{Msg}_2}}{\optionfn_2}$,
        it holds that
$
            \epair{\apair{\mathit{EL}_1}{\mathit{Msg}_1}}{\optionfn_1}
                \mmerge
            \epair{\apair{\mathit{EL}_2}{\mathit{Msg}_2}}{\optionfn_2}
                \is
            \epair{\apair{(\mathit{EL}_1 \merge EL_2)}{\mathit{Msg}}}{\optionfn'}
$ \\
\begin{scriptsize}
        where
        $\optionfn_1, \optionfn_2, \optionfn' \in \set{\final, \nonfinal}$,
        $\optionfn' = \final$ iff $\optionfn_i = \final$ for any $i \in \set{1, 2}$
        and
$
            \apair{\_}{\mathit{Msg}} =
                \apair{\mathit{EL}_1}{\mathit{Msg}_1} \merge \apair{EL_2}{\mathit{Msg}_2}.
$
\end{scriptsize}
\end{itemize}
By assumption, every variable $t \in T$ is only bound once and is bound before use in any global type.
Therefore $\eta(t)$ is always defined.
The blocked projection returns an extended availability annotated local type.
Erasing the annotations yields an extended local type.
\end{definition}

The special rule for recursion ensures that we unfold the definition (equi-recursively) if the set $\mathcal{B}$ has changed since the recursion variable has been bound.
As for the standard projection, we define $\GG \tproj^\mathcal{B}_{\procC, n} \is \GG \tproj^{E, \mathcal{B}}_{\procC, n}$.

When we want to emphasise the difference, we call projection $\tproj_{\_}$ the \emph{standard} projection, in contrast to the \emph{blocked} projection.
The blocked projection neglects all actions that depend on actions by any role in $\mathcal{B}$.
In the course of computation, $\mathcal{B}$ grows and represents the set of roles that are not able to proceed any further since their actions depend on some previous action from a blocked role.
With each unfolding, the set either grows or there is no unfolding necessary anymore.
Therefore, it holds that $\forall n \geq \card{\Procs} - \card{\mathcal{B}}. \; \GG \tproj_{\procC, n}^{\mathcal{B}} = \GG \tproj_{\procC, n+1}^{\mathcal{B}}$.
In case $n$ is greater than this threshold, we omit it for readability.

By definition, the blocked projection operator mimcs the standard one for $\mathcal{B} = \emptyset$ and $n = 0$.

\begin{proposition}
For any global type $\GG$ and $\procC$, it holds that $\GG \tproj^{\emptyset}_{\procC, 0} = \GG \tproj_\procC$.
\end{proposition}

\begin{remark}
Even though the blocked projection operator and the message causality analysis both share the concept of a growing set of blocked roles, it does not seem beneficial to unify both approaches.
First, we would establish a cyclic dependency between the projection operator and the message causality analysis used when merging.
Second, the message availability analysis fixes a role $\procC$ that gets blocked.
However, the blocked merging operator would return $0$ as soon as it encounters $\procC$ again which will not yield the desired information about the channel contents.
\end{remark}

\subsubsection*{Correspondence of Blocked Projection Operator and Blocked Languages}

\begin{lemma}[Correspondence of blocked projection operator and blocked languages]
\label{lm:correspondence-bpo-bl}
Let $\GG$ be a \projectable global type, $G$ a syntactic subterm of $\GG$, and $\procC \in \Procs$.
Then, it holds that 
$ 
    \forall \mathcal{B} \subseteq \Procs.\;
    L^{\mathcal{B}}_{(G \ldots)} \wproj_{\Alphabet_\procC}
    =
    \lang(G \tproj^{\mathcal{B}}_{\procC}).
$
\end{lemma}
\begin{proof}
Let $T'$ be all recursion variables in $G$ that are not bound in $G$.
We define $\widehat{G}$ to be the type that is obtained by substituting every $t \in T'$ by $\getMuG(t)$.
By construction, $\widehat{G}$ is a global type (and not only a syntactic subterm of a latter).
Then, it suffices to show that
$
    \forall \mathcal{B} \subseteq \Procs.\;
    L^{\mathcal{B}}_{(\widehat{G} \ldots)} \wproj_{\Alphabet_\procC}
    =
    \lang(\widehat{G} \tproj^{\mathcal{B}}_{\procC})
$
where $\widehat{G}$ can be considered to be the overall global type for $\semglobal(\widehat{G})$ as well.
We prove this claim by induction on $\widehat{G}$.
There are two base cases.
For $0$, the claim trivially holds and $t$ is no global type which $\widehat{G}$ is by assumption.
For the first induction step, let $\widehat{G} = \mu t.\widehat{G'}$ and the induction hypothesis holds for $\widehat{G'}$.
Both definition agree on the fact that $\mu t$ itself does not carry any event.
The blocked projection ensures to only bind a recursion variable if $\procC$ does use it later on.
There is no such need in the blocked language along paths.
We do the same case analysis as the blocked projection operator.
First, if $t$ is used after some actual event, $t$ is bound and the induction hypotheses concludes the claim.
Second, if $t$ is never used, $t$ is not bound and the blocked language along the path yields the same language by induction hypothesis.
Third, otherwise, i.e. if the projection of the continuation yields $\apair{t}{\_}$, the blocked language along path enters an empty loop and the claim follows.
For the second induction step, let $\widehat{G} = \Sum_{i \in I} \msgFromTo{\procA}{\procB_i}{\val_i}.\widehat{G_i}$ and the induction hypotheses hold for all $\widehat{G_i}$.
We do two case analyses and inspect the different cases in both definitions that apply: one whether $\procA \in \mathcal{B}$ and another one whether $\procC = \procA$.
In every but the last case, we only consider the added events and omit the recursion since this is handled by the induction hypotheses.

 (i) Suppose that $\procA \notin \mathcal{B}$ and $\procA = \procC$. On both sides, $\mathcal{B}$ does not change in the recursive calls.
        Thanks to the projection $\wproj_{\Alphabet_\procC}$, the same event $\snd{\procA}{\procB_i}{\val_i}$ is added before recursing.
        (Due to the projection on the left side, there is no need for a case analysis for $\procB_i \in \mathcal{B}$).

(ii) Suppose that $\procA \notin \mathcal{B}$ and $\procA \neq \procC$. \\
        For this case, we consider each branch $i \in I$ individually and do another case analysis whether $\procC = \procB_i$. \\
        For $\procC \neq \procB_i$, no events are added on the left side (and branches of empty loops, yielding empty languages along the path, ignored for further projection) and the ones added on the right are projected away. \\
        For $\procC = \procB_i$, we need another case split and check whether $\procC \in \mathcal{B}$.
        If so, the left side recurses and produces a language $\set{\emptystring}$ with this, the right side stops immediately and returns $\apair{\epair{0}{\emptyset}}{\final}$ which produces the same language.
        If not, both sides add the same reception event for $\procC$ (and the send event is projected away).
        
(iii) Suppose that $\procA \in \mathcal{B}$ and $\procA = \procC$. On both sides, the result is $\set{\emptystring}$.
        
(iv) Suppose that $\procA \in \mathcal{B}$ and $\procA \neq \procC$. On both sides, no event is added but the recursions add $\procB_i$ to the set $\mathcal{B}$.
        These recursions are also covered by the induction hypotheses because we quantified over $\mathcal{B}$.
\end{proof}

\subsubsection*{Properties of \Projectable Global Types}

In this section, we establish properties which will be used in the induction step of the proof for \cref{lm:proj-types-have-run-mappings} --- one for the send and one for the receive case.

\begin{lemma}[Property of \Projectable global types I - send]
\label{lm:property-send-projectable-MST}
Let $\GG$ be a \projectable global type and $G = \Sum_{i \in I} \msgFromTo{\procA}{\procB_i}{\val_i}.G_i$ be some syntactic subterm.
Then, for all $i, j \in I$ and $\procC \in \Procs$, it holds that
$
L^{\set{\procA, \procB_i}}_{(G_i \ldots)} \wproj_{\Alphabet_\procC}
=
L^{\set{\procA, \procB_j}}_{(G_j \ldots)} \wproj_{\Alphabet_\procC}.
$
\end{lemma}
\begin{proof}
With \cref{lm:correspondence-bpo-bl}, it suffices to show that
$
G_i \tproj^{\set{\procA, \procB_i}}_{\procC} =
G_j \tproj^{\set{\procA, \procB_j}}_{\procC} $
for all $i, j$ and $\procC$.
First, we take care of the two special cases where $\procC$ is either $\procA$ or $\procB_i$ for some $i \in I$.
If $\procC = \procA$, it is straightforward that $\apair{\epair{0}{\emptyset}}{\final}$ is the result of both projections.
For the second case, suppose that there was $i, j \in I$ such that both projections are different and w.l.o.g.\ $\procC = \procB_i$.
Then, it holds that
$G_i \tproj^{\set{\procA, \procB_i}}_{\procC} = \apair{\epair{0}{\emptyset}}{\final}$.
We know that the $\merge$ of both projections is defined and therefore, $\procC$ cannot send as first action in the projection for $G_j$.
Therefore, it needs to receive.
However, then this message must be unique for this branch and by definition, the choice which branch to take is currently blocked and such a message cannot exist in the blocked projection.
Therefore, both projections are
$\apair{\epair{0}{\emptyset}}{\final}$.

For the case where $\procC \neq \procA$ and $\forall i. \; \procC \neq \procB_i$, we introduce a slight variation $\mmmerge$ of $\mmerge$ to state a sufficient claim.
The merge operator $\mmmerge$ is defined as $\mmerge$ but does only merge $\optionfn_1$ and $\optionfn_2$ if $\optionfn_1 = \optionfn_2$ and does require that $I = J$ when merging receptions.
Then, $\mmmerge$ does only merge extended local types that are exactly the same by definition.

Hence, it suffices to show that
$
G_i \tproj^{\set{\procA, \procB_i}}_{\procC} \mmmerge
G_j \tproj^{\set{\procA, \procB_j}}_{\procC} \text{ is defined for all } i, j, \text{ and } \procC.
$

For this, we first show that the set $\mathcal{B}$ in two branches will be the same if some recursion variable is merged with $\mmmerge$. \\
\emph{Claim I:}
For all $i,j$, and $\procC$, if
$t \tproj_\procC^{\mathcal{B}_i} \mmmerge t \tproj_\procC^{\mathcal{B}_j}$
occurs in the computation of the above merge with $\mmmerge$,
then $\mathcal{B}_i = \mathcal{B}_j$.

\emph{Proof Claim I.} \\
(Note that in the presence of empty loops, the definition prevents such kinds of merge by definition.)
Recall that $G_i$ and $G_j$ are two branches of the same choice in which $\procC$ is not involved.
Intuitively, this is true because the $\merge$-operator does not unfold recursions and a role is either agnostic to some choice (merging two recursion variables) or has learnt about the choice before encountering the recursion variable.
Towards a contradiction, suppose that there are $i, j$ and $\procC$ such that $\mathcal{B}_i \neq \mathcal{B}_j$.
W.l.o.g., there must be some role $\procD$ such that $\procD \in \mathcal{B}_i$ and $\procD \notin \mathcal{B}_j$.
It follows that $\procD \notin \set{\procA, \procB_j}$.
We do a case analysis whether $\procD = \procB_i$.

If so, $\procD$ receives in branch $i$ and since $\merge$ is defined, it must also receive in every continuation specified by $G_j$ -- otherwise $t$ would be merged with a reception which is undefined.
(There could be subsequent branches in $G_j$ and therefore $\procD$ could potentially receive different messages but at least one in every continuation.)
By definition, these possible receptions must be unique and therefore distinct from $\val_i$ by $\procA$ and all possible messages for $\procD$ in the continuation of $G_j$.
Because of this uniqueness, the send event for any reception in $G_j$ must not be possible in $G_i$.
But for this, the send event must depend on the decision by $\procA$ and therefore $\procD \in \mathcal{B}_j$ which yields a contradiction.

Suppose that $\procD \neq \procB_i$.
By assumption, $\procD \in \mathcal{B}_i$ so there must be some reception for $\procD$ for which the corresponding send event depends on $\snd{\procA}{\procB_i}{\val_i}$.
From here, the reasoning is analogous to the previous case.

\emph{End Proof of Claim I.} 

In Claim I, $\mathcal{B}_i$ and $\mathcal{B}_j$ capture all roles that depend on the choice by $\procA$ (without unfolding the recursion).
The standard projection and merge operator $\merge$ do not unfold any recursions and $t \merge t$ is the only rule to merge recursion variables.
Therefore, a role either has to learn about a choice by receiving some unique message before recursing or will never learn about it (since $t$ maps back to the same continuation).
Recall that we omit the parameter $n$ for the blocked projection operator if it is bigger than the number of iterations needed to reach a  fixed point.
Combined with Claim I, it suffices to show that
$
G_i \tproj^{\set{\procA, \procB_i}}_{\procC} \mmmerge
G_j \tproj^{\set{\procA, \procB_j}}_{\procC} $
does not return undefined until the first recursion variable is encountered, for all $i, j$, and $\procC$.
Towards a contradiction, suppose it does return undefined for some $i, j$, and $\procC$.

We check all possibilities why $\mmmerge$ can be undefined.
To start with, either of both projections~$\tproj$ can only return undefined when the parameter $n$ is too small which is not the case here.
For $l \in \set{i, j}$, let
$
\epair{
    \apair{\mathit{EL}_l}
    {\mathit{Msg_l}}
}
{\optionfn_l}
=
G_l \tproj^{\set{\procA, \procB_i}}_{\procC} $
Then,
$
\epair{
    \apair{\mathit{EL}_i}
    {\mathit{Msg_i}}
}
{\optionfn_i}
\mmmerge
\epair{
    \apair{\mathit{EL}_j}
    {\mathit{Msg_j}}
}
{\optionfn_j}
$
is undefined if $\optionfn_i \neq \optionfn_j$ or
$\mathit{EL}_i \mmmerge EL_j$ (*) is undefined.

The \textbf{first case} can only occur if
$\epair{\apair{0}{\emptyset}}{\final}$
is merged on top level in the projection of one branch and not the other.
W.l.o.g., let $G_i \tproj_\procC^{\set{\procA, \procB_i}}$ be the branch where this happens.
On the one hand, this occurs when $\procC$ receives from a role in $\mathcal{B}_i$ and since $G_i \tproj_\procC \merge G_j \tproj_\procC$ is defined by assumption, $\procC$ also receives in the branch of $j$ by definition of $\merge$.
For the standard merge $\merge$, both messages are checked to be unique for both branches.
However, if $\procC$ can still receive the message in $G_j \tproj_\procC^{\mathcal{B}_j}$, this cannot depend on the choice by $\procA$ and hence must also be possible in $G_i$ which yields a contradiction.
On the other hand, $\apair{0}{\emptyset}$ can also be the result of projecting empty loops (after some choice operation).
Then, however, since $G_i \tproj_\procC$ is defined, all other branches also return $\apair{0}{\emptyset}$ for the standard projection to be defined.
In turn, since $G_i \tproj_\procC \merge G_j \tproj_\procC$, $G_j \tproj_\procC$ is $\apair{0}{\emptyset}$ and hence the above projection cannot be undefined.

For the \textbf{second case}, we first observe that both $\merge$ and $\mmmerge$ agree on the type of events that can be merged, e.g., only send events can merge with send events.
This is why we do not consider all the different combinations, e.g., of merging send events with recursion binders, but only the ones that are defined in both cases.
Note that $(*)$ is defined for cases where the rules of $\mmmerge$ mimic the ones of $\merge$ literally.
(Recall that we do only need to check until the first recursion variable thanks to Claim I.)
In this sense, both definitions $\mmmerge$ and $\merge$ agree for merging $0$, recursion variables $t$, recursion binders $\mu t. \ldots$ and send events (internal choice $\IntCh$).
For reception events (external choice $\ExtCh$), they do not agree.
The variant $\mmmerge$ requires that both sets $I$ and $J$ are the same.
Since $\optionfn_i = \optionfn_j$, there needs to be some message exchange in one branch that cannot happen in the other branch.
However, analogous to the previous case, the send event of such a reception event must not depend on the choice of $\procA$ which branch to take.
However, this is not possible with the blocked projection operator so $\merge$ would also be undefined in this case which yields a contradiction.
\end{proof}

\begin{lemma}[Property of \Projectable global types II - receive]
\label{lm:property-receive-projectable-MST}
Let $\GG$ be some global type and $G$ be some syntactic subterm of $\GG$.
It holds that
$M^{\mathcal{B}}_{(G \ldots)} \subseteq \avail(\mathcal{B},∅,G)$ for every $\mathcal{B} \subseteq \Procs$.
\end{lemma}
\begin{proof}
Let $T'$ be all recursion variables in $G$ that are not bound in $G$.
We define $\widehat{G}$ to be the type that is obtained by substituting every $t \in T'$ by $\getMuG(t)$.
By construction, $\widehat{G}$ is a global type (and not only a syntactic subterm of a latter).

Let $T$ be the set of all recursion variables in $\widehat{G}$.
Since we unfolded the unbound recursion variables up-front in $\widehat{G}$ prohibits us from substituting for any recursion variable again,
it is straightforward, from the definition of $\avail(\hole, \hole, \hole)$, that the following holds:
\[
\avail(\mathcal{B}, \emptyset, G) = \avail(\mathcal{B}, T, \widehat{G})
\]
By construction of $\widehat{G}$, the computation of $\avail(\mathcal{B}, T, \widehat{G})$ will never apply the second case which is the only one where we do not simply descend in the structure of the third parameter so induction hypotheses will apply.
By the equi-recursive view on global types, it holds that
$
M^{\mathcal{B}}_{(G \ldots)}
=
M^{\mathcal{B}}_{(\widehat{G} \ldots)}
$
and therefore it suffices to show that
$
M^{\mathcal{B}}_{(\widehat{G} \ldots)}
\subseteq
\avail(\mathcal{B}, T, \widehat{G}).
$

We do induction on the structure of $\widehat{G}$ to prove the claim.
For the base case where $\widehat{G} = 0$, the claim follows trivially.
For $\widehat{G} = t$, we know that either $t \in T'$ has been substituted by its definition using $\getMuG(\hole)$ or it was bound in~$G$.
Normally, one would need to substitute $t$ again and consider the messages.
However, in both cases, we know that we traversed all branches from the binder to the variable at least once and unfolding further does not change the set of available messages.

For the induction step, suppose that $\widehat{G} = \mu t.\widehat{G'}$.
By induction hypothesis, we know that the claim holds for $\widehat{G'}$ and the binder for recursion variable $t$ does not add any message to the set of messages:
$
    M^{\mathcal{B}}_{(\mu t.\widehat{G'} \ldots)} =
    M^{\mathcal{B}}_{(\widehat{G'} \ldots)}.
$
By definition, we have that
$
    \avail(\mathcal{B}, T, \mu t.\widehat{G'}) =
    \avail(\mathcal{B}, T ∪ \set{t}, \widehat{G'}) =
    \avail(\mathcal{B}, T, \widehat{G'}).
$
By induction hypothesis, it holds that
$M^{\mathcal{B}}_{(\widehat{G'} \ldots)}
\subseteq
\avail(\mathcal{B}, T, \widehat{G'})$
which proves the claim.

Last, let
$\widehat{G} = \Sum_{i \in I} \msgFromTo{\procA}{\procB_i}{\val_i}.\widehat{G'_i}$
for some index set $I$, roles $\procA$, $\procB_i$, message values $\val_i$ and global subterms $\widehat{G'_i}$ for every $i \in I$.
We do a case analysis on whether $\procA \notin \mathcal{B}$.

Suppose that $\procA \notin \mathcal{B}$.
Then, all receipts are added by definition of $M^{\mathcal{B}}_{\_}$:
\[
    M^{\mathcal{B}}_{(\widehat{G} \ldots)} =
    \set{\rcv{\procA}{\procB}{\val_i} \mid i \in I} \union
    \Union_{i \in I} M^{\mathcal{B}}_{(\widehat{G'_i}.\ldots)}.
\]
For $\avail(\mathcal{B}, T, \widehat{G})$, the fourth case in the definition applies:
\[
    \avail(\mathcal{B}, T, \widehat{G}) =
    \set{ \rcv{\procA}{\procB}{\val_i} \mid i ∈ I } ∪
    \bigcup_{i∈I} \avail(\mathcal{B}, T, \widehat{G'_i}).
\]
The first part of the equations are identical and combining the induction hypotheses for each $\widehat{G'_i}$ proves the claim.

Suppose that $\procA \in \mathcal{B}$.
Then, by definition, it holds that that $M^{\mathcal{B}}_{(\widehat{G} \ldots)}$ does not contain any of the receipts $\rcv{\procA}{\procB_i}{\_}$.
Additionally, $\procB_i$ is added to $\mathcal{B}$ by definition:
$
    M^{\mathcal{B}}_{(\widehat{G} \ldots)} =
    \Union_{i \in I} M^{\mathcal{B} \union \set{\procB_i}}_{(\widehat{G'_i}.\ldots)}.
$
For $\avail(\mathcal{B}, T, \widehat{G})$, the fifth case applies:
$
    \avail(\mathcal{B}, T, \widehat{G}) =
    \bigcup_{i∈I} \avail(\mathcal{B} \union \set{\procB_i}, T, \widehat{G'_i}).
$
Combining the induction hypotheses for each $\widehat{G'_i}$ proves the claim.
\end{proof}
 
\subsubsection*{Proof of \cref{lm:proj-types-have-run-mappings}: \\ Existence of Family of Run Mappings for \Projectable Global Types}
Equipped with these properties, we can prove \cref{lm:proj-types-have-run-mappings} which reads:
Let $\GG$ be a \projectable global type.
Then, $\CSMl{\semlocal(\GG \tproj_\procA)}$ has a family of run mappings for $\semglobal(\GG)$.

\begin{proof}
Let $w$ be a prefix of an execution of $\CSMl{\semlocal(\GG \tproj_\procA)}$.
We prove the claim by induction on the length of $w$.
The base case where $w = \emptystring$ is trivial: $\rho = \emptystring$ and $\rho(\procA) = \emptystring$ for all $\procA$.

For the induction step, we append $x$ to obtain $wx$ for which $x = \tau$,
$x = \rcv{\procB}{\procA}{\val}$ or
$x = \snd{\procA}{\procB}{\val}$ for some $\procA, \procB$ and $\val$.
As induction hypothesis, we assume that there is some run $\rho$ in $\semglobal(\GG)$ such that
$w \wproj_{\Alphabet}$ is the prefix of $w'$ with
$w' \interswap \trace(\rho)$
and a run mapping $\rho(\hole)$ such that
$ w \wproj_{\Alphabet_\procA} = \trace(\rho(\procA)) \wproj_{\Alphabet_\procA}$
for every $\procA$.

For all cases, we re-use the witness run $\rho$ by extending it to $\rho'$ when necessary and re-use the run mapping $\rho(\procA)$ except for at most one process.

\textbf{First}, Suppose that $x = \tau$.
Then, the claim follows with the same run $\rho$ and mapping $\rho(\hole)$.
It is straightforward that $w \wproj_{\Alphabet_\procA} = (w\tau) \wproj_{\Alphabet_\procA}$ for every $\procA$ holds.
Let $\procA$ be the process which takes an $\emptystring$-transition and hence leads to the extension by $\tau$.
By \cref{prop:shape-semantics-local-L}, we know that a state can only have a single $\emptystring$-transition and hence the execution prefix $w \tau$ can be extended in the same way as $w$ is extended to $w'$ to obtain $w''$ for $w\tau$.
Since $w'$ and $w''$ only differ by one $\tau$, $w'' \interswap w' \interswap \trace(\rho)$ and the claim follows by transitivity of $\interswap$.

\textbf{Second}, suppose that $x = \rcv{\procB}{\procA}{\val}$. \\
\emph{Claim 1.} It holds that $\rho(\procA) \preforder \rho(\procB)$.

      \emph{Proof of Claim 1.}
      We know that $\trace(\rho(\procA)) \wproj_{\Alphabet_\procA} = w \wproj_{\Alphabet_\procA}$
      and $\trace(\rho(\procB)) \wproj_{\Alphabet_\procB} = w \wproj_{\Alphabet_\procB}$.
      By definition of $\semglobal(\GG)$, we know that $\procA$'s and $\procB$'s common actions always happen in pairs of sending and receiving a message.
      Since there is no out-of-order execution, we know that $\procA$ cannot have received $\val$ from $\procB$ yet.
      Hence, $\rho(\procA) \preforder \rho(\procB)$.

      \emph{End Proof of Claim 1.}

For every $\procC \neq \procA$, we define $\rho'(\procC) = \rho(\procC)$.
We know that $(wx) \wproj_{\Alphabet_\procC} = w \wproj_{\Alphabet_\procC}$ for every $\procC \neq \procA$ so the conditions on the run mapping is satisfied for all $\procC \neq \procA$.

For $\procA$, we extend $\rho(\procA)$ to be maximal, i.e., extending the run further would render one of the conditions for run mappings unsatisfied, and obtain a set of possible extended runs.
For every run, either the last or second-last state of these runs corresponds to some syntactic subterm of $\GG$.
Note that one of them will be a prefix of the sender's run $\rho(\procB)$.
In case of empty loop branches, the set of runs might be infinite, however, only one is a prefix of the sender's run.
We denote its syntactic subterm by $G'$ and we will show that only the option taken by $\procB$ can be pursued.

Since $\procA$ can receive $\val$ from $\procB$, we know that the current local type $L$ of $\procA$ is the result of (merging) the projection of some type(s) of shape $\Sum_{\_} \msgFromTo{\_}{\_}{\_}.\_$ and all these syntactic subterms are of this form.

Since $\procA$ can receive $\val$ from $\procB$, the local type $L$ must be the result of some merge applying the receiving rule last to compute
\[
 \apair{L}{\_} = \apair{L₁}{\mathit{Msg}_1} \merge \ldots \merge \apair{L_n}{\mathit{Msg}_n}
\]
 for some $n$ and
$L_l = \ExtCh_{i ∈ I_l} \rcv{\procB_l}{}{\val_i.AL_{l,i}}$ for every $l \in \set{1, \ldots, n}$ and for some $l \in \set{1, \ldots, n}$, it holds that $\apair{L_l}{\mathit{Msg}_l} = G' \tproj_\procA$.
The other local types correspond to the branches $\procA$ explores concurrently while processing $w \wproj_{\Alphabet_\procA}$.
Note that we unrolled the binary definition of merge but since $\merge$ is associative and commutative, we know that $\apair{L_l}{\mathit{Msg}_l} \merge \apair{L_{l'}}{\mathit{Msg}_{l'}}$ is defined for all $l, l' \in \set{1, \ldots, n}$.
Therefore it holds, for any $l, l' \in \set{1, \ldots, n}$, that
   \begin{align*}
            ∀ i∈ I_l \setminus I_{l'}.\, \rcv{\procB_l}{\procA}{\val_i} \notin \mathit{Msg}_{l'}, & \text{ and }\\
            ∀ i∈ I_{l'} \setminus I_l.\, \rcv{\procB_{l'}}{\procA}{\val_i} \notin \mathit{Msg}_l. &
   \end{align*}

W.l.o.g., let $k \in \Union_{1 \leq l \leq n} I_l$ be the index for which $\procB = \procB_k$ and $\val = \val_k$.
Therefore, for any $l$ with $k \notin I_l$, it holds that $\rcv{\procB_k}{\procA}{\val_k} \notin \mathit{Msg}_l$.

Let $l$ be some index.
To obtain every individual availability annotated local type $\apair{L_l}{\mathit{Msg}_l}$, the following case of the projection has been applied:
\[
    \apair{L_l}{\mathit{Msg}_l} =
    \apair{
        \ExtCh_{j \in J_l} \rcv{\procB_l}{}{\val_j}.(G_j \tproj_\procA)
    }{
        \bigcup_{j∈J} \avail(\set{\procA},∅,G_j)
    }
\]
for some index sets $J_l$ and global syntactic subterms $G_j$.
\\
From \cref{lm:property-receive-projectable-MST}, we have that, for every $l \in \set{1, \ldots, n}$, $j \in I_l$ and $\mathcal{B} \subseteq \Procs$, it holds that
$M^{\mathcal{B}}_{(G_j \ldots)} \subseteq \avail(\mathcal{B},∅,G_j)$ (*). Recall that $M^{\mathcal{B}}_{(G_j \ldots)}$ is defined as all messages that can be sent in any run starting with branch $G_j$ when no process in $\mathcal{B}$ can take any further step.

By definition of $\merge$, all the first actions $\rcv{\procB_k}{}{\val_k}$ will be merged together to one branch in the local type $L$ so $\procA$ will continue with this branch of $L$.
Let us denote this branch by $B_k$ for now.
By instantiating (*) with $\mathcal{B} = \set{\procA}$ whose next action is to receive from $\procB$, we know that $\val_k$ by $\procB_k$ cannot occur in the channel $\channel{\procB}{\procA}$ in any branch of $L$ different from $B_k$ so $\procA$ cannot diverge from $\procB$'s run $\rho(\procB)$ by receiving $\val_k$ out of order.
Hence, $\procA$ follows the run of $\procB$.
By assumption $\procA$ can receive from $\procB$, so $\procB$ must have taken one of the runs that corresponds to this branch of $L$ so we can choose $\rho'(\procA)$ such that $\rho'(\procA) \preforder \rho(\procB)$ which proves the claim.

\textbf{Third}, suppose that $x = \snd{\procA}{\procB}{\val}$.

We collect the processes with longer runs than $\procA$ in a set: $\mathcal{S} \is
\set{ \procC \mid \rho(\procA) \preforder \rho(\procC) \land \rho(\procA) \neq \rho(\procC)}$.
So $\procA \notin \mathcal{S}$.
We assume $\mathcal{S}$ is not empty as the claim follows trivially if it is.
(We could simply extend $\rho$, keep the mapping for all other processes and extend it for $\procA$ accordingly.)

As before, for $\procA$, we extend $\rho(\procA)$ to be maximal, i.e., extending the run further would render one of the conditions unsatisfied, and obtain a set of possible runs.
In the presence of empty loop branches, this set could be infinite. For this single step of the induction, it suffices to obtain some valid run mapping. By construction, all possible paths through the loop are present. For further extensions of the run mapping, we exploit the idea of prophecy variables to get the run mapping where $\procA$ chose the path in the loop as chosen by the deciding processes and show that this does not impose any restrictions.
For every run in this set, either the last or second-last state of these runs corresponds to some syntactic subterm of $\GG$.
Since $\procA$ can send $\val$ to $\procB$, we know that the current local type $L$ of $\procA$ is the result of (merging) the projection of some type(s) of shape $\Sum_{\_} \msgFromTo{\_}{\_}{\_}.\_$ and all these syntactic subterms are of this form.

The current local type $L$ of $\procA$ is the result of merging of this shape since $\procA$ sends at this step:
\[
 L =
 \Merge_{i \in I}
 (\IntCh_{j \in J} \snd{}{\procB_j}{\val_j}.(G_{(i,j)} \tproj_\procA))
 =
 \IntCh_{j \in J} \snd{}{\procB_j}{\val_j}.
 \Merge_{i \in I} (G_{(i,j)} \tproj_\procA)
\]
for some index sets $I$ and $J$, $\procB_j$ as well as $\val_j$ and $G_{(i,j)}$ for every $j \in J$ and $i \in I$.

So the merge $\merge$ ensures that $\procA$ has the same options to send after processing $w \wproj_{\Alphabet_\procA}$, no matter which run in $\semglobal(\GG)$ was pursued.
This ensures that we are able to adapt the mapping from the induction hypothesis such that $\trace(\rho'(\procA)) = (wx) \wproj_{\Alphabet_\procA}$.

Inspired by prophecy variables \cite{DBLP:conf/lics/AbadiL88}, we assume that each process $\procC$ in $\mathcal{S}$ followed the run that $\procA$ will choose to take, i.e., $\rho'(\procA) \preforder \rho(\procC)$.
Assuming this, we can re-use the same mapping for them: $\rho'(\procC) = \rho(\procC)$ and the overall claim follows.

It remains to show that the mappings $\rho(\procC)$ for $\procC \in \mathcal{S}$ can be chosen in such a way without prohibiting any of the processes to proceed with some events.
We will do so by applying \cref{lm:property-send-projectable-MST}.

Process $\procA$ might be pursuing more than one run in $\semglobal(\GG)$.
Since all first send options for $\procA$ could be merged, we know that they are the same over all possible continuations.
Therefore, we show that no matter which option $\procA$ chooses, no process $\procC \in \mathcal{S}$ can have processed an action that is incompatible with this branch yet, provided that neither the sender $\procA$ nor each receiver $\procB_i$ has processed some event in the continuations of the branch of option $i$.

For this, we show for all $\procC \in \mathcal{S}, i \in I$ and $j_1, j_2 \in J$ that
\[
 L^{\set{\procA, \procB_i}}_{(G_{(i,j_1)} \ldots)} \wproj_{\Alphabet_\procC}
 =
 L^{\set{\procA, \procB_j}}_{(G_{(i,j_2)} \ldots)} \wproj_{\Alphabet_\procC}.
\]
Since we do only compare runs that start at the same state, i.e., the syntactic subterm $\IntCh_{j \in J} \snd{}{\procB_j}{\val_j}.G_{(i,j)}$,
this claim follows from \cref{lm:property-send-projectable-MST}.
It suffices to compare those since $\procA$ can not and does not dictate which of the branches of $I$ were taken, but the ones in $J$.
It might happen that some choice wrt.\ the options in $I$ has already been committed but $\procA$ does not know about this at this stage yet.
Still, all processes in $\mathcal{S}$ have the same options along all the options in $J$ and therefore the use of prophecy variables does not prohibit them from proceeding with any action that is possible in the execution.

\end{proof}
  
\subsubsection*{Runs Following a Common Path are Extendable}

\begin{lemma}[Runs following a common path are extendable]
\label{lm:exec-prefix-single-run-extendable}
Let $\GG$ be a \projectable global type,
$w$ an execution prefix of $\CSMl{\semlocal(\GG \tproj_\procA)}$, and
$\rho$ a finite run in $\semglobal(\GG)$ such that
$w \wproj_{\Alphabet_\procA} \preforder \trace(\rho) \wproj_{\Alphabet_\procA}$ for every $\procA$.
Then, there is an extension $w'$ of $w$ in $\CSMl{\semlocal(\GG \tproj_\procA)}$ such that $w' \interswap \trace(\rho)$.
\end{lemma}
\begin{proof}
From \cref{lm:pref-MST-has-run}, we know that every prefix of $\lang(\GG)$ has some run in $\CSMl{\semlocal(\GG \tproj_\procA)}$.
It holds that $\trace(\rho) \in \pref(\lang(\GG))$ so there is a run for $\trace(\rho)$ in $\CSMl{\semlocal(\GG \tproj_\procA)}$.
With \cref{lm:csm-closed-under-interswap}, it follows that every trace in $\interswaplang(\trace(\rho))$ has a run in $\CSMl{\semlocal(\GG \tproj_\procA)}$.
Hence, it suffices to show that $w \in \pref(\interswaplang(\trace(\rho))$.
To prove this, we need to find $w'$ such that $w \preforder w'$ and $w' \interswap \trace(\rho)$.

By \cref{lm:interswap-complete-wrt-channelcompliancy}, it suffices to find a \channelcompliant $w'$ such that
$w' \wproj_{\Alphabet_\procA} = \trace(\rho) \wproj_{\Alphabet_\procA}$.

For every $\procA$, there is $y_\procA$ such that $(w \wproj_{\Alphabet_\procA}).y_\procA = \trace(\rho) \wproj_{\Alphabet_\procA}$.
We show that there is a \channelcompliant $w' = wy$ such that $y \wproj_{\Alphabet_\procA} = y_\procA$.
Since $w$ is an execution prefix of a CSM, $w$ is \channelcompliant (\cref{lm:execution-prefix-and-channel-content}).
We need to construct $y$ and start with $y = \emptystring$.
We extend $y$ by consuming prefixes of $y_\procA$ and do so as long as some $y_\procA$ is not empty.
At all times, we preserve the following invariant for every $\procA$:
$
    (wy) \wproj_{\Alphabet_\procA}.y_\procA =
    \trace(\rho) \wproj_{\Alphabet_\procA}.
$

If there is any $y_\procA = (\snd{\procA}{\procB}{\val}).y'_\procA$, we extend $y = y.(\snd{\procA}{\procB}{\val})$ and set $y_\procA = y'_\procA$.
This preserves the invariant and appending actions of shape $\snd{\_}{\_}{\_}$ to $y$ does preserve \channelcompliancy of $wy$.
If there is no such $y_\procA$ (and not all $y_\procA$ are empty), there is a set of roles $\mathcal{S}$ such that $y_\procA = (\rcv{\_}{\procA}{\_}).y'_\procA$ for every $\procA \in \mathcal{S}$.
We define the shortest prefix $\rho'$ of $\rho$ such that there is some role whose trace is included in the projection of $\trace(\rho')$:
\[
\rho' \is
\min_{\preforder}\set{\rho' \mid \exists \procA \in \mathcal{S}.\, (wy) \wproj_{\Alphabet_\procA} = \trace(\rho') \wproj_{\Alphabet_\procA} \text{ for } \rho' \preforder \rho}.
\]
Let $\procA \in \mathcal{S}$ such that $(wy) \wproj_{\Alphabet_\procA} = \trace(\rho') \wproj_{\Alphabet_\procA}$.
Let $\rho' = \rho''.q.q'$.
By design, we know that the label $x \in \Alphabet_\procA$ for transition $(q, x, q')$ in $\semglobal(\GG)$.
Therefore, the choice of $\procA$ is unique and for all other roles $\procB \neq \procA$, it holds that $(wy) \wproj_{\Alphabet_\procB} \neq \trace(\rho') \wproj_{\Alphabet_\procB}$.
By assumption, it holds that $(wy) \wproj_{\Alphabet_\procB} \preforder \trace(\rho)\wproj_{\Alphabet_\procB}$ and hence $\trace(\rho') \wproj_{\Alphabet_\procB} \preforder (wy) \wproj_{\Alphabet_\procB}$ because $\rho' \preforder \rho$ for $\procB \neq \procA$ (*).

Let $y_\procA = (\rcv{\procB}{\procA}{\val}).y'_\procA$ for some $\procB$ and $\val$.
We claim that appending the first action $\rcv{\procB}{\procA}{\val}$ from $y_\procA$ to $y$, i.e., $y \is y.(\rcv{\procB}{\procA}{\val})$ and $y_\procA = y'_\procA$, preserves \channelcompliancy of $wy$.
It is straightforward that the invariant is preserved.
For \channelcompliancy of $wy$, we need to show that
$
    \MsgVals((wy) \wproj_{\snd{\procB}{\procA}{\_}}) \preforder
    \MsgVals((wy) \wproj_{\rcv{\procB}{\procA}{\_}}).
$

By the invariant, we know that
$
    (wy) \wproj_{\Alphabet_\procA}.y_\procA =
    \trace(\rho) \wproj_{\Alphabet_\procA}
    \text{ and }
    (wy) \wproj_{\Alphabet_\procB}.y_\procB =
    \trace(\rho) \wproj_{\Alphabet_\procB}
$,
so it holds that
$
    (w.y.y_\procA) \wproj_{\rcv{\procB}{\procA}{\_}} =
    \trace(\rho) \wproj_{\rcv{\procB}{\procA}{\_}}
    \text{ and }
    (w.y.y_\procB) \wproj_{\snd{\procB}{\procA}{\_}} =
    \trace(\rho) \wproj_{\snd{\procB}{\procA}{\_}}.
$
Because global types describe protocols,
$\trace(\rho)$ is \channelcompliant which yields that
    $\MsgVals(\trace(\rho) \wproj_{\rcv{\procB}{\procA}{\_}}) =
     \MsgVals(\trace(\rho) \wproj_{\snd{\procB}{\procA}{\_}})$.
This ensures that the next message in $\channel{\procB}{\procA}$ is $\val$ if it exists.
By design of $\semglobal(\GG)$, the send and reception part of a $\msgFromTo{\procC}{\procD}{\val}$ always occur in pairs on the run $\rho$, i.e., it holds that for any transitions $q' \redtoover{\rcv{\procC}{\procD}{\val}} q''$ in $\rho$, there is $q$ such that $q \redtoover{\snd{\procC}{\procD}{\val}} q'$ in $\rho$.
From (*), we know that $\snd{\procB}{\procA}{\val}$ must be the next action of shape $\snd{\procB}{\procA}{\_}$ with unmatched reception in $wy$.
Since our channels preserve FIFO order, the next message in channel $\channel{\procB}{\procA}$ is~$\val$.

Overall, this procedure ensures that we can always extend $wy$ by appending prefixes of $y_\procA$ in a way that preserves \channelcompliancy of $wy$ and
$
    (wy) \wproj_{\Alphabet_\procA}.y_\procA =
    \trace(\rho) \wproj_{\Alphabet_\procA}.
$ for every $\procA$.
Note that $\rho$ is finite and hence each $y_\procA$ is finite so the procedure terminates.
For $w'$, we can choose the result of the procedure $wy$ and it holds, for every~$\procA$, that
$
  (wy) \wproj_{\Alphabet_\procA} =
  \trace(\rho) \wproj_{\Alphabet_\procA}
$
which proves the claim.
\end{proof}
 
\subsubsection*{From Run Mapping via Control Flow Agreement to Protocol Fidelity}

To show protocol fidelity, we use a strengthening of the run mappings by a progress condition.
Any execution prefix can be extended to match all actions in the run given by the run mapping.

\begin{definition}[Control Flow Agreement]
\label{def:cfa}
Let $\GG$ be a global type.
A CSM $\CSM{A}$
satisfies \emph{Control Flow Agreement (CFA)} for $\semglobal(\GG)$ iff
for every execution prefix $w$ of $\CSM{A}$,
there is a run $\rho$ in $\semglobal(\GG)$ such that the following holds:
\begin{itemize}
\item $w \wproj_{\Alphabet_\procA}$ is a prefix of
$\trace_{\semglobal(\GG)}(\rho) \wproj_{\Alphabet_\procA}$ for every role $\procA$, and
\item $w$ can be extended (in $\CSM{A}$) to $w'$ such that
$w' \interswap \trace_{\semglobal(\GG)}(\rho)$.
\end{itemize}
\end{definition}

\begin{lemma}[Run mappings entail CFA]
\label{lm:run-mappings-entail-cfa}
Let $\GG$ be a \projectable global type,
If $\CSMl{\semlocal(\GG \tproj_\procA)}$ has a family of run mappings for $\semglobal(\GG)$,
then $\CSMl{\semlocal(\GG \tproj_\procA)}$ satisfies CFA for $\semglobal(\GG)$.
\end{lemma}
\begin{proof}
Let $w$ be an execution prefix of $\CSMl{\semlocal(\GG \tproj_\procA)}$ with witness run $\rho$ and its run mapping $\rho(\hole)$.
We can assume that $\rho$ is finite since $w$ is a finite execution prefix.
We know that
$\trace(\rho(\procA)) \wproj_{\Alphabet_\procA} = w \wproj_{\Alphabet_\procA}$ and $\rho(\procA) \preforder \rho$ for every $\procA$.
Hence, it holds that
$w \wproj_{\Alphabet_\procA} \preforder \trace(\rho) \wproj_{\Alphabet_\procA}$
for every $\procA$,
which is exactly the first condition of CFA.

The second condition requires that we can extend $w$ to $w'$ in $\CSMl{\semlocal(\GG \tproj_\procA)}$
such that $w' \interswap \trace(\rho)$.
Since $\rho$ is finite, we can apply \cref{lm:exec-prefix-single-run-extendable} and obtain an extension $w'$ of $w$ in $\CSMl{\semlocal(\GG \tproj_\procC)}$ such that $w' \interswap \trace(\rho)$.
\end{proof}
 
\begin{lemma}[CFA entails protocol fidelity]
\label{lm:cfa-entails-prot-fidelity}
Let $\GG$ be a \projectable global type.
If the CSM $\CSMl{\semlocal(\GG \tproj_\procA)}$ satisfies CFA for $\semglobal(\GG)$,
then
$
    \interswaplang(\lang(\GG))
    =
    \lang(\CSMl{\semlocal(\GG \tproj_\procA)}).
$
\end{lemma}
\begin{proof}
First, \cref{lm:projections-yield-superset} yields that
$
    \lang(\GG)
    \subseteq
    \lang(\CSMl{\semlocal(\GG \tproj_\procA)}).
$
The language of a CSM is closed under $\interswap$ (\cref{lm:csm-closed-under-interswap}), hence the first inclusion holds.

It remains to prove the second inclusion:
$
    \lang(\CSMl{\semlocal(\GG \tproj_\procA)})
    \subseteq
    \interswaplang(\lang(\GG)).
$

Let $w \in \lang(\CSMl{\semlocal(\GG \tproj_\procA)})$.
For both the finite and infinite case, we will use \cref{def:cfa} to obtain a maximal run $\rho$ in $\semglobal(\GG)$ for which $w \wproj_{\Alphabet_\procA} \in \trace(\rho) \wproj_{\Alphabet_\procA}$ for every~$\procA$.

We do a case split whether $w$ is finite or infinite.

In case $w$ is finite, we use \cref{def:cfa} to find the corresponding run~$\rho$ in~$\semglobal(\GG)$ and want to show that $w \interswap \trace(\rho)$.
So we do not extend $w$ to obtain~$w'$ but keep it as it is.
First, we apply \cref{def:cfa} to obtain the shortest run~$\rho$ that satisfies both conditions.

We show that $\rho$ ends in a final state $s$ and is hence maximal in $\semglobal(\GG)$.
Let $(q, \xi)$ be the final configuration of $\CSMl{\semlocal(\GG \tproj_\procA)}$ that has been reached with trace $w$ and therefore $q_\procA$ be the states for the individual state machines $\semlocal(\GG \tproj_\procA)$ for every $\procA$.
Since $(q, \xi)$ is final, every $q_\procA$ is final in $\semlocal(\GG \tproj_\procA)$.
Then, we know that $q_\procA$ corresponds to $0$ by definition of $\semlocal(\GG \tproj_\procA)$.
By \cref{prop:shape-semantics-local-L}, $q_\procA$ does not have any outgoing transitions.
\cref{prop:projection-preserves-per-process-runs} states that every run in $\semglobal(\GG) \wproj_{\Alphabet_\procA}$ is also possible in $\semlocal(G \tproj_\procA)$.
Therefore, by contradiction, it follows that a run in $\semglobal(\GG) \wproj_{\Alphabet_\procA}$ with trace $w \wproj_{\Alphabet_\procA}$ cannot be extended to a run with a longer trace for all $\procA$.
Hence, the run of $\trace(\rho)$ in $\semglobal(\GG)$ cannot be extended to a run with a longer trace either, i.e., there can only be $\emptystring$-transitions.
By definition of $\semglobal(\GG)$, $\emptystring$-transitions occur from states corresponding to $t$ and $\mu t.G'$.
We claim that $s$ does not correspond to any $t$ or syntactic subterm of shape $\mu t.G'$.
Towards a contradiction, assume that $s$ corresponds to some $t$ or subterm of shape $\mu t.G'$.
Global types have guarded recursion, i.e., $G'$ cannot be $0$ for any $\mu t.G'$.
Then there is some non-$\emptystring$-transition from $s$ which is a contradiction.

Hence, there is also no $\emptystring$-transition from $s$ and $s$ does not have any outgoing transitions.
By \cref{prop:shape-semantics-global-G}, $s$ corresponds to $0$, is final in $\semglobal(\GG)$, and $\rho$ is maximal in $\semglobal(\GG)$.

It remains to show that we do not need to extend $w$ to $w'$ for any role.
We define $C \is \set{ \procA \mid w \wproj_{\Alphabet_\procA} \neq \trace(\rho) \wproj_{\Alphabet_\procA} }$.
We claim that $C = \emptyset$.
Towards a contradiction, let us assume that $w \wproj_{\Alphabet_\procA} \neq \trace(\rho) \wproj_{\Alphabet_\procA}$ for $\procA \in C$. Since we reached the final state $s$ with $\rho$ in $\semglobal(\GG)$ which corresponds to $0$, there is no successor in $\CSMl{\semlocal(\GG \tproj_\procA)}$ for $q_\procA$ either and there is no possibility to extend $w$ for any $\procA$ but this contradicts the second condition of CFA and we do not need to extend $w'$ for any role.

So it holds that $w \interswap \trace(\rho)$.
Because $\rho$ is maximal, we have that $\trace(\rho) \subseteq \lang(\GG)$ and the claim (for finite $w$) follows because $\interswaplang(\lang(\GG))$ is closed under~$\interswap$ by definition.

Suppose that $w$ is infinite.
We show that there is an infinite run in $\semglobal(\GG)$ that matches with $w$ when closing with $\interswap$.

From \cref{def:cfa}, we can obtain a finite run~$\rho$ for every prefix $u$ of $w$ such that the conditions for CFA hold, i.e., there is an extension $u'$ of $u$ such that $u' \interswap \trace(\rho)$.
To simplify the argument, we use an idea similar to prophecy variables \cite{DBLP:conf/lics/AbadiL88} as in the proof of \cref{lm:proj-types-have-run-mappings}.
Here, we can use a similar oracle that tells which witness run and mapping to use for every prefix of $w$.
This does not restrict the roles in any way.
All conditions of CFA hold for the prefix and this witness run since they hold for longer prefixes of $w$.
By this, we ensure that the run $\rho$ can always be extended for longer prefixes of $w$ $(\square)$.

Consider a tree $\mathcal{T}$ where each node corresponds to a run $\rho$ on some finite prefix $w' \leq w$.
The root is labelled by the empty run.  The children of a node $\rho$ are runs that extend $\rho$ by a single transition---these exist by $(\square)$.
Since our CSM, derived from a global type, is a finite number of finite state machines, $\mathcal{T}$ is finitely branching.
By König's Lemma, there is an infinite path in $\mathcal{T}$ that corresponds to an infinite run $\rho$ in $\semglobal(\GG)$ for which $w \preceq^\omega_\interswap \trace(\rho)$.
By definition of $\interswaplang(\lang(\GG))$ for infinite traces, this yields $w \in \interswaplang(\lang(\GG))$.
\end{proof}
 
\begin{lemma}[CFA for $\CSMl{\semlocal(\GG \tproj_\procA)}$ entails deadlock freedom]
\label{lm:CFA-entails-deadlockfree-CSM-of-projected-types}
Let $\GG$ be a \projectable global type.
If $\CSMl{\semlocal(\GG \tproj_\procA)}$ satisfies CFA for $\semglobal(\GG)$,
then $\CSMl{\semlocal(\GG \tproj_\procA)}$ is deadlock~free.
\end{lemma}
\begin{proof}
Towards a contradiction, assume that $\CSMl{\semlocal(\GG \tproj_\procA)}$ has a deadlock.
Then, there is an execution prefix $w$ for which $\CSMl{\semlocal(\GG \tproj_\procA)}$ reaches a non-final configuration $(q, \xi)$ and no transition is possible.
Let $\rho$ be the maximal run from \cref{def:cfa} for which the conditions of CFA  hold.
Since $\CSMl{\semlocal(\GG \tproj_\procA)}$ cannot make another step, the second condition can only apply for $w' = w$, so all events along $\rho$ have been processed by all roles, i.e., $w \wproj_{\Alphabet_\procA} = \trace(\rho) \wproj_{\Alphabet_\procA}$ for every $\procA$ by \cref{lm:interswap-complete-wrt-channelcompliancy}.

\emph{Claim I:}
We claim the last state $s$ of $\rho$ corresponds to $0$.

\emph{Proof of Claim I.}
Towards a contradiction, assume that $\rho$ can be extended, i.e., $s$ has a successor.
Since $\rho$ is maximal, there cannot be any $\emptystring$-transitions from $s$ but only some transition labelled different than $\emptystring$.
Since all roles have processed all events along $\rho$, such a transition can only be labelled by $\snd{\procA}{\procB}{\val}$ for some $\procA$, $\procB$, and $\val$.
Hence, there is a run for $(w \wproj_{\Alphabet_\procA}).\snd{\procA}{\procB}{\val}$ in $\semglobal(\GG) \wproj_{\Alphabet_\procA}$.
By \cref{prop:projection-preserves-per-process-runs}, there is a run for $(w \wproj_{\Alphabet_\procA}).\snd{\procA}{\procB}{\val}$ in $\semlocal(\GG \tproj_\procA)$.
Role $\procA$ can hence take another transition which contradicts the assumption that $(q, \xi)$ is a deadlock. \\
\emph{End Proof of Claim I.}

So we know that the last state $s$ of $\rho$ corresponds to $0$.
Again, by \cref{prop:projection-preserves-per-process-runs}, there is no transition for any $\procA$ from $q_\procA$ in $\semlocal(\GG \tproj_\procA)$ because $\semglobal(\GG) \wproj_{\Alphabet_\procA}$ has no transition for any $\procA$.
All channels in $\xi$ are empty and $w \wproj_{\Alphabet_\procA} = \trace(\rho) \wproj_{\Alphabet_\procA}$ for every $\procA$.
By \cref{prop:shape-semantics-local-L}, we know that there is only a single state in $\semlocal(\GG \tproj_\procA)$ that does not have outgoing transitions, i.e., the one that corresponds to $0$ and hence is final.
But then, all individual final states $q_\procA$ are final (and all channels in $\xi$ are empty) so $(q, \xi)$ is a final configuration which yields the desired contradiction.
\end{proof}
 
\subsection{Proof of
\cref{thm:projectable-MST-sat-protocol-fidelity}: \\
Any \Projectable Global Type is Implementable}
Now,
we can prove
\cref{thm:projectable-MST-sat-protocol-fidelity}:
Any \projectable global type is implementable.

\medskip
\begin{proof}[Proof of \cref{thm:projectable-MST-sat-protocol-fidelity}]
Let $\CSMl{\semlocal(\GG \tproj_\procA)}$ be the CSM of local types.
Then, we know that $\CSMl{\semlocal(\GG \tproj_\procA)}$ has a family of run mappings for $\semglobal(\GG)$ by \cref{lm:proj-types-have-run-mappings}.
This, in turn, entails that $\CSMl{\semlocal(\GG \tproj_\procA)}$ satisfies control-flow agreement for $\semglobal(\GG)$ by \cref{lm:run-mappings-entail-cfa}.
This satisfies the condition for \cref{lm:CFA-entails-deadlockfree-CSM-of-projected-types} (and \ref{lm:cfa-entails-prot-fidelity}) and hence $\CSMl{\semlocal(\GG \tproj_\procA}$ is deadlock free.
With \cref{lm:projections-yield-superset,lm:cfa-entails-prot-fidelity}, protocol fidelity holds and the claim follows.
\end{proof}


\begin{thebibliography}{10}

\bibitem{prototype}
Prototype {I}mplementation of {G}eneralised {P}rojection for {M}ultiparty
  {S}ession {T}ypes.
\newblock \url{https://gitlab.mpi-sws.org/fstutz/async-mpst-gen-choice/}.

\bibitem{DBLP:conf/lics/AbadiL88}
Mart{\'{\i}}n Abadi and Leslie Lamport.
\newblock The existence of refinement mappings.
\newblock In {\em Proceedings of the Third Annual Symposium on Logic in
  Computer Science {(LICS} '88), Edinburgh, Scotland, UK, July 5-8, 1988},
  pages 165--175. {IEEE} Computer Society, 1988.

\bibitem{DBLP:conf/cav/AbdullaBJ98}
Parosh~Aziz Abdulla, Ahmed Bouajjani, and Bengt Jonsson.
\newblock On-the-fly analysis of systems with unbounded, lossy {FIFO} channels.
\newblock In Alan~J. Hu and Moshe~Y. Vardi, editors, {\em Computer Aided
  Verification, 10th International Conference, {CAV}'98, Vancouver, BC, Canada,
  June 28 - July 2, 1998, Proceedings}, volume 1427 of {\em Lecture Notes in
  Computer Science}, pages 305--318. Springer, 1998.

\bibitem{DBLP:journals/tcs/AlurEY05}
Rajeev Alur, Kousha Etessami, and Mihalis Yannakakis.
\newblock Realizability and verification of {MSC} graphs.
\newblock {\em Theor. Comput. Sci.}, 331(1):97--114, 2005.

\bibitem{DBLP:journals/ftpl/AnconaBB0CDGGGH16}
Davide Ancona, Viviana Bono, Mario Bravetti, Joana Campos, Giuseppe Castagna,
  Pierre{-}Malo Deni{\'{e}}lou, Simon~J. Gay, Nils Gesbert, Elena Giachino,
  Raymond Hu, Einar~Broch Johnsen, Francisco Martins, Viviana Mascardi,
  Fabrizio Montesi, Rumyana Neykova, Nicholas Ng, Luca Padovani, Vasco~T.
  Vasconcelos, and Nobuko Yoshida.
\newblock Behavioral types in programming languages.
\newblock {\em Found. Trends Program. Lang.}, 3(2-3):95--230, 2016.

\bibitem{DBLP:conf/sigsoft/BakerBJKTMB05}
Paul Baker, Paul Bristow, Clive Jervis, David~J. King, Robert Thomson, Bill
  Mitchell, and Simon Burton.
\newblock Detecting and resolving semantic pathologies in {UML} sequence
  diagrams.
\newblock In Michel Wermelinger and Harald~C. Gall, editors, {\em Proceedings
  of the 10th European Software Engineering Conference held jointly with 13th
  {ACM} {SIGSOFT} International Symposium on Foundations of Software
  Engineering, 2005, Lisbon, Portugal, September 5-9, 2005}, pages 50--59.
  {ACM}, 2005.

\bibitem{DBLP:conf/coordination/BarbaneraLT20}
Franco Barbanera, Ivan Lanese, and Emilio Tuosto.
\newblock Choreography automata.
\newblock In Simon Bliudze and Laura Bocchi, editors, {\em Coordination Models
  and Languages - 22nd {IFIP} {WG} 6.1 International Conference, {COORDINATION}
  2020, Held as Part of the 15th International Federated Conference on
  Distributed Computing Techniques, DisCoTec 2020, Valletta, Malta, June 15-19,
  2020, Proceedings}, volume 12134 of {\em Lecture Notes in Computer Science},
  pages 86--106. Springer, 2020.

\bibitem{DBLP:conf/tacas/Ben-AbdallahL97}
Han{\^{e}}ne Ben{-}Abdallah and Stefan Leue.
\newblock Syntactic detection of process divergence and non-local choice
  inmessage sequence charts.
\newblock In Ed~Brinksma, editor, {\em Tools and Algorithms for Construction
  and Analysis of Systems, Third International Workshop, {TACAS} '97, Enschede,
  The Netherlands, April 2-4, 1997, Proceedings}, volume 1217 of {\em Lecture
  Notes in Computer Science}, pages 259--274. Springer, 1997.

\bibitem{DBLP:conf/esop/BocchiMVY19}
Laura Bocchi, Maurizio Murgia, Vasco~Thudichum Vasconcelos, and Nobuko Yoshida.
\newblock Asynchronous timed session types - from duality to time-sensitive
  processes.
\newblock In Lu{\'{\i}}s Caires, editor, {\em Programming Languages and Systems
  - 28th European Symposium on Programming, {ESOP} 2019, Held as Part of the
  European Joint Conferences on Theory and Practice of Software, {ETAPS} 2019,
  Prague, Czech Republic, April 6-11, 2019, Proceedings}, volume 11423 of {\em
  Lecture Notes in Computer Science}, pages 583--610. Springer, 2019.

\bibitem{DBLP:conf/concur/BolligFS20}
Benedikt Bollig, Alain Finkel, and Amrita Suresh.
\newblock Bounded reachability problems are decidable in {FIFO} machines.
\newblock In Igor Konnov and Laura Kov{\'{a}}cs, editors, {\em 31st
  International Conference on Concurrency Theory, {CONCUR} 2020, September 1-4,
  2020, Vienna, Austria (Virtual Conference)}, volume 171 of {\em LIPIcs},
  pages 49:1--49:17. Schloss Dagstuhl - Leibniz-Zentrum f{\"{u}}r Informatik,
  2020.

\bibitem{DBLP:journals/jacm/BrandZ83}
Daniel Brand and Pitro Zafiropulo.
\newblock On communicating finite-state machines.
\newblock {\em J. {ACM}}, 30(2):323--342, 1983.

\bibitem{DBLP:journals/tcs/BravettiCZ18}
Mario Bravetti, Marco Carbone, and Gianluigi Zavattaro.
\newblock On the boundary between decidability and undecidability of
  asynchronous session subtyping.
\newblock {\em Theor. Comput. Sci.}, 722:19--51, 2018.

\bibitem{DBLP:conf/forte/CairesP16}
Lu{\'{\i}}s Caires and Jorge~A. P{\'{e}}rez.
\newblock Multiparty session types within a canonical binary theory, and
  beyond.
\newblock In Elvira Albert and Ivan Lanese, editors, {\em Formal Techniques for
  Distributed Objects, Components, and Systems - 36th {IFIP} {WG} 6.1
  International Conference, {FORTE} 2016, Held as Part of the 11th
  International Federated Conference on Distributed Computing Techniques,
  DisCoTec 2016, Heraklion, Crete, Greece, June 6-9, 2016, Proceedings}, volume
  9688 of {\em Lecture Notes in Computer Science}, pages 74--95. Springer,
  2016.

\bibitem{DBLP:journals/mscs/CairesPT16}
Lu{\'{\i}}s Caires, Frank Pfenning, and Bernardo Toninho.
\newblock Linear logic propositions as session types.
\newblock {\em Math. Struct. Comput. Sci.}, 26(3):367--423, 2016.

\bibitem{DBLP:conf/concur/CarboneLMSW16}
Marco Carbone, Sam Lindley, Fabrizio Montesi, Carsten Sch{\"{u}}rmann, and
  Philip Wadler.
\newblock Coherence generalises duality: {A} logical explanation of multiparty
  session types.
\newblock In Jos{\'{e}}e Desharnais and Radha Jagadeesan, editors, {\em 27th
  International Conference on Concurrency Theory, {CONCUR} 2016, August 23-26,
  2016, Qu{\'{e}}bec City, Canada}, volume~59 of {\em LIPIcs}, pages
  33:1--33:15. Schloss Dagstuhl - Leibniz-Zentrum f{\"{u}}r Informatik, 2016.

\bibitem{DBLP:journals/acta/CarboneMSY17}
Marco Carbone, Fabrizio Montesi, Carsten Sch{\"{u}}rmann, and Nobuko Yoshida.
\newblock Multiparty session types as coherence proofs.
\newblock {\em Acta Informatica}, 54(3):243--269, 2017.

\bibitem{DBLP:journals/corr/abs-1203-0780}
Giuseppe Castagna, Mariangiola Dezani{-}Ciancaglini, and Luca Padovani.
\newblock On global types and multi-party session.
\newblock {\em Log. Methods Comput. Sci.}, 8(1), 2012.

\bibitem{DBLP:journals/acta/CastellaniDG19}
Ilaria Castellani, Mariangiola Dezani{-}Ciancaglini, and Paola Giannini.
\newblock Reversible sessions with flexible choices.
\newblock {\em Acta Informatica}, 56(7-8):553--583, 2019.

\bibitem{DBLP:journals/tcs/CastellaniDGH20}
Ilaria Castellani, Mariangiola Dezani{-}Ciancaglini, Paola Giannini, and Ross
  Horne.
\newblock Global types with internal delegation.
\newblock {\em Theor. Comput. Sci.}, 807:128--153, 2020.

\bibitem{DBLP:journals/iandc/CeceF05}
G{\'{e}}rard C{\'{e}}c{\'{e}} and Alain Finkel.
\newblock Verification of programs with half-duplex communication.
\newblock {\em Inf. Comput.}, 202(2):166--190, 2005.

\bibitem{DBLP:journals/lmcs/ChenDSY17}
Tzu{-}Chun Chen, Mariangiola Dezani{-}Ciancaglini, Alceste Scalas, and Nobuko
  Yoshida.
\newblock On the preciseness of subtyping in session types.
\newblock {\em Log. Methods Comput. Sci.}, 13(2), 2017.

\bibitem{DBLP:conf/ppdp/ChenDY14}
Tzu{-}Chun Chen, Mariangiola Dezani{-}Ciancaglini, and Nobuko Yoshida.
\newblock On the preciseness of subtyping in session types.
\newblock In Olaf Chitil, Andy King, and Olivier Danvy, editors, {\em
  Proceedings of the 16th International Symposium on Principles and Practice of
  Declarative Programming, Kent, Canterbury, United Kingdom, September 8-10,
  2014}, pages 135--146. {ACM}, 2014.

\bibitem{DBLP:conf/sefm/DanHC10}
Haitao Dan, Robert~M. Hierons, and Steve Counsell.
\newblock Non-local choice and implied scenarios.
\newblock In Jos{\'{e}}~Luiz Fiadeiro, Stefania Gnesi, and Andrea
  Maggiolo{-}Schettini, editors, {\em 8th {IEEE} International Conference on
  Software Engineering and Formal Methods, {SEFM} 2010, Pisa, Italy, 13-18
  September 2010}, pages 53--62. {IEEE} Computer Society, 2010.

\bibitem{DBLP:journals/corr/abs-1902-06056}
Ankush Das, Stephanie Balzer, Jan Hoffmann, and Frank Pfenning.
\newblock Resource-aware session types for digital contracts.
\newblock {\em CoRR}, abs/1902.06056, 2019.

\bibitem{DBLP:conf/esop/DenielouY12}
Pierre{-}Malo Deni{\'{e}}lou and Nobuko Yoshida.
\newblock Multiparty session types meet communicating automata.
\newblock In Helmut Seidl, editor, {\em Programming Languages and Systems -
  21st European Symposium on Programming, {ESOP} 2012, Held as Part of the
  European Joint Conferences on Theory and Practice of Software, {ETAPS} 2012,
  Tallinn, Estonia, March 24 - April 1, 2012. Proceedings}, volume 7211 of {\em
  Lecture Notes in Computer Science}, pages 194--213. Springer, 2012.

\bibitem{DBLP:journals/corr/abs-1208-6483}
Pierre{-}Malo Deni{\'{e}}lou, Nobuko Yoshida, Andi Bejleri, and Raymond Hu.
\newblock Parameterised multiparty session types.
\newblock {\em Log. Methods Comput. Sci.}, 8(4), 2012.

\bibitem{DBLP:conf/eurosys/FahndrichAHHHLL06}
Manuel F{\"{a}}hndrich, Mark Aiken, Chris Hawblitzel, Orion Hodson, Galen~C.
  Hunt, James~R. Larus, and Steven Levi.
\newblock Language support for fast and reliable message-based communication in
  singularity {OS}.
\newblock In Yolande Berbers and Willy Zwaenepoel, editors, {\em Proceedings of
  the 2006 EuroSys Conference, Leuven, Belgium, April 18-21, 2006}, pages
  177--190. {ACM}, 2006.

\bibitem{DBLP:conf/litp/Gastin90}
Paul Gastin.
\newblock Infinite traces.
\newblock In Ir{\`{e}}ne Guessarian, editor, {\em Semantics of Systems of
  Concurrent Processes, {LITP} Spring School on Theoretical Computer Science,
  La Roche Posay, France, April 23-27, 1990, Proceedings}, volume 469 of {\em
  Lecture Notes in Computer Science}, pages 277--308. Springer, 1990.

\bibitem{DBLP:conf/ac/GenestMP03}
Blaise Genest, Anca Muscholl, and Doron~A. Peled.
\newblock Message sequence charts.
\newblock In J{\"{o}}rg Desel, Wolfgang Reisig, and Grzegorz Rozenberg,
  editors, {\em Lectures on Concurrency and Petri Nets, Advances in Petri Nets
  [This tutorial volume originates from the 4th Advanced Course on Petri Nets,
  {ACPN} 2003, held in Eichst{\"{a}}tt, Germany in September 2003. In addition
  to lectures given at {ACPN} 2003, additional chapters have been
  commissioned]}, volume 3098 of {\em Lecture Notes in Computer Science}, pages
  537--558. Springer, 2003.

\bibitem{DBLP:journals/jcss/GenestMSZ06}
Blaise Genest, Anca Muscholl, Helmut Seidl, and Marc Zeitoun.
\newblock Infinite-state high-level mscs: Model-checking and realizability.
\newblock {\em J. Comput. Syst. Sci.}, 72(4):617--647, 2006.

\bibitem{DBLP:journals/tcs/Girard87}
Jean{-}Yves Girard.
\newblock Linear logic.
\newblock {\em Theor. Comput. Sci.}, 50:1--102, 1987.

\bibitem{DBLP:conf/sdl/Helouet01}
Lo{\"{\i}}c H{\'{e}}lou{\"{e}}t.
\newblock Some pathological message sequence charts, and how to detect them.
\newblock In Rick Reed and Jeanne Reed, editors, {\em {SDL} 2001: Meeting UML,
  10th International {SDL} Forum Copenhagen, Denmark, June 27-29, 2001,
  Proceedings}, volume 2078 of {\em Lecture Notes in Computer Science}, pages
  348--364. Springer, 2001.

\bibitem{NO-DBLP-wrong-local-choice}
Lo{\"{\i}}c H{\'{e}}lou{\"{e}}t and Claude Jard.
\newblock Conditions for synthesis of communicating automata from {HMSCs}.
\newblock In {\em In 5th International Workshop on Formal Methods for
  Industrial Critical Systems (FMICS)}, 2000.

\bibitem{DBLP:conf/concur/Honda93}
Kohei Honda.
\newblock Types for dyadic interaction.
\newblock In Eike Best, editor, {\em {CONCUR} '93, 4th International Conference
  on Concurrency Theory, Hildesheim, Germany, August 23-26, 1993, Proceedings},
  volume 715 of {\em Lecture Notes in Computer Science}, pages 509--523.
  Springer, 1993.

\bibitem{DBLP:conf/pvm/HondaMMNVY12}
Kohei Honda, Eduardo R.~B. Marques, Francisco Martins, Nicholas Ng,
  Vasco~Thudichum Vasconcelos, and Nobuko Yoshida.
\newblock Verification of {MPI} programs using session types.
\newblock In Jesper~Larsson Tr{\"{a}}ff, Siegfried Benkner, and Jack~J.
  Dongarra, editors, {\em Recent Advances in the Message Passing Interface -
  19th European {MPI} Users' Group Meeting, EuroMPI 2012, Vienna, Austria,
  September 23-26, 2012. Proceedings}, volume 7490 of {\em Lecture Notes in
  Computer Science}, pages 291--293. Springer, 2012.

\bibitem{DBLP:conf/esop/HondaVK98}
Kohei Honda, Vasco~Thudichum Vasconcelos, and Makoto Kubo.
\newblock Language primitives and type discipline for structured
  communication-based programming.
\newblock In Chris Hankin, editor, {\em Programming Languages and Systems -
  ESOP'98, 7th European Symposium on Programming, Held as Part of the European
  Joint Conferences on the Theory and Practice of Software, ETAPS'98, Lisbon,
  Portugal, March 28 - April 4, 1998, Proceedings}, volume 1381 of {\em Lecture
  Notes in Computer Science}, pages 122--138. Springer, 1998.

\bibitem{DBLP:conf/popl/HondaYC08}
Kohei Honda, Nobuko Yoshida, and Marco Carbone.
\newblock Multiparty asynchronous session types.
\newblock In George~C. Necula and Philip Wadler, editors, {\em Proceedings of
  the 35th {ACM} {SIGPLAN-SIGACT} Symposium on Principles of Programming
  Languages, {POPL} 2008, San Francisco, California, USA, January 7-12, 2008},
  pages 273--284. {ACM}, 2008.

\bibitem{DBLP:journals/jacm/HondaYC16}
Kohei Honda, Nobuko Yoshida, and Marco Carbone.
\newblock Multiparty asynchronous session types.
\newblock {\em J. {ACM}}, 63(1):9:1--9:67, 2016.

\bibitem{DBLP:conf/fase/HuY17}
Raymond Hu and Nobuko Yoshida.
\newblock Explicit connection actions in multiparty session types.
\newblock In Marieke Huisman and Julia Rubin, editors, {\em Fundamental
  Approaches to Software Engineering - 20th International Conference, {FASE}
  2017, Held as Part of the European Joint Conferences on Theory and Practice
  of Software, {ETAPS} 2017, Uppsala, Sweden, April 22-29, 2017, Proceedings},
  volume 10202 of {\em Lecture Notes in Computer Science}, pages 116--133.
  Springer, 2017.

\bibitem{DBLP:journals/csur/HuttelLVCCDMPRT16}
Hans H{\"{u}}ttel, Ivan Lanese, Vasco~T. Vasconcelos, Lu{\'{\i}}s Caires, Marco
  Carbone, Pierre{-}Malo Deni{\'{e}}lou, Dimitris Mostrous, Luca Padovani,
  Ant{\'{o}}nio Ravara, Emilio Tuosto, Hugo~Torres Vieira, and Gianluigi
  Zavattaro.
\newblock Foundations of session types and behavioural contracts.
\newblock {\em {ACM} Comput. Surv.}, 49(1):3:1--3:36, 2016.

\bibitem{DBLP:conf/esop/JongmansY20}
Sung{-}Shik Jongmans and Nobuko Yoshida.
\newblock Exploring type-level bisimilarity towards more expressive multiparty
  session types.
\newblock In Peter M{\"{u}}ller, editor, {\em Programming Languages and Systems
  - 29th European Symposium on Programming, {ESOP} 2020, Held as Part of the
  European Joint Conferences on Theory and Practice of Software, {ETAPS} 2020,
  Dublin, Ireland, April 25-30, 2020, Proceedings}, volume 12075 of {\em
  Lecture Notes in Computer Science}, pages 251--279. Springer, 2020.

\bibitem{DBLP:journals/corr/abs-1902-01353}
Dimitrios Kouzapas, Ramunas Gutkovas, A.~Laura Voinea, and Simon~J. Gay.
\newblock A session type system for asynchronous unreliable broadcast
  communication.
\newblock {\em CoRR}, abs/1902.01353, 2019.

\bibitem{DBLP:conf/popl/LangeTY15}
Julien Lange, Emilio Tuosto, and Nobuko Yoshida.
\newblock From communicating machines to graphical choreographies.
\newblock In Sriram~K. Rajamani and David Walker, editors, {\em Proceedings of
  the 42nd Annual {ACM} {SIGPLAN-SIGACT} Symposium on Principles of Programming
  Languages, {POPL} 2015, Mumbai, India, January 15-17, 2015}, pages 221--232.
  {ACM}, 2015.

\bibitem{DBLP:conf/fossacs/LangeY17}
Julien Lange and Nobuko Yoshida.
\newblock On the undecidability of asynchronous session subtyping.
\newblock In Javier Esparza and Andrzej~S. Murawski, editors, {\em Foundations
  of Software Science and Computation Structures - 20th International
  Conference, {FOSSACS} 2017, Held as Part of the European Joint Conferences on
  Theory and Practice of Software, {ETAPS} 2017, Uppsala, Sweden, April 22-29,
  2017, Proceedings}, volume 10203 of {\em Lecture Notes in Computer Science},
  pages 441--457, 2017.

\bibitem{DBLP:conf/cav/LangeY19}
Julien Lange and Nobuko Yoshida.
\newblock Verifying asynchronous interactions via communicating session
  automata.
\newblock In Isil Dillig and Serdar Tasiran, editors, {\em Computer Aided
  Verification - 31st International Conference, {CAV} 2019, New York City, NY,
  USA, July 15-18, 2019, Proceedings, Part {I}}, volume 11561 of {\em Lecture
  Notes in Computer Science}, pages 97--117. Springer, 2019.

\bibitem{DBLP:journals/tcs/Lohrey03}
Markus Lohrey.
\newblock Realizability of high-level message sequence charts: closing the
  gaps.
\newblock {\em Theor. Comput. Sci.}, 309(1-3):529--554, 2003.

\bibitem{majumdar2021generalising}
Rupak Majumdar, Madhavan Mukund, Felix Stutz, and Damien Zufferey.
\newblock Generalising projection in asynchronous multiparty session types.
\newblock {\em CoRR}, abs/2107.03984, 2021.
\newblock \url{https://arxiv.org/abs/2107.03984}.

\bibitem{DBLP:conf/ecoop/MajumdarPYZ19}
Rupak Majumdar, Marcus Pirron, Nobuko Yoshida, and Damien Zufferey.
\newblock Motion session types for robotic interactions (brave new idea paper).
\newblock In Alastair~F. Donaldson, editor, {\em 33rd European Conference on
  Object-Oriented Programming, {ECOOP} 2019, July 15-19, 2019, London, United
  Kingdom}, volume 134 of {\em LIPIcs}, pages 28:1--28:27. Schloss Dagstuhl -
  Leibniz-Zentrum f{\"{u}}r Informatik, 2019.

\bibitem{DBLP:conf/fase/MooijGR05}
Arjan~J. Mooij, Nicolae Goga, and Judi Romijn.
\newblock Non-local choice and beyond: Intricacies of {MSC} choice nodes.
\newblock In Maura Cerioli, editor, {\em Fundamental Approaches to Software
  Engineering, 8th International Conference, {FASE} 2005, Held as Part of the
  Joint European Conferences on Theory and Practice of Software, {ETAPS} 2005,
  Edinburgh, UK, April 4-8, 2005, Proceedings}, volume 3442 of {\em Lecture
  Notes in Computer Science}, pages 273--288. Springer, 2005.

\bibitem{DBLP:conf/stacs/Morin02}
R{\'{e}}mi Morin.
\newblock Recognizable sets of message sequence charts.
\newblock In Helmut Alt and Afonso Ferreira, editors, {\em {STACS} 2002, 19th
  Annual Symposium on Theoretical Aspects of Computer Science, Antibes - Juan
  les Pins, France, March 14-16, 2002, Proceedings}, volume 2285 of {\em
  Lecture Notes in Computer Science}, pages 523--534. Springer, 2002.

\bibitem{DBLP:conf/fase/Muccini03}
Henry Muccini.
\newblock Detecting implied scenarios analyzing non-local branching choices.
\newblock In Mauro Pezz{\`{e}}, editor, {\em Fundamental Approaches to Software
  Engineering, 6th International Conference, {FASE} 2003, Held as Part of the
  Joint European Conferences on Theory and Practice of Software, {ETAPS} 2003,
  Warsaw, Poland, April 7-11, 2003, Proceedings}, volume 2621 of {\em Lecture
  Notes in Computer Science}, pages 372--386. Springer, 2003.

\bibitem{DBLP:journals/acta/PengP92}
Wuxu Peng and S.~Purushothaman.
\newblock Analysis of a class of communicating finite state machines.
\newblock {\em Acta Informatica}, 29(6/7):499--522, 1992.

\bibitem{DBLP:journals/pacmpl/ScalasY19}
Alceste Scalas and Nobuko Yoshida.
\newblock Less is more: multiparty session types revisited.
\newblock {\em Proc. {ACM} Program. Lang.}, 3({POPL}):30:1--30:29, 2019.

\bibitem{DBLP:conf/ppdp/ToninhoCP11}
Bernardo Toninho, Lu{\'{\i}}s Caires, and Frank Pfenning.
\newblock Dependent session types via intuitionistic linear type theory.
\newblock In Peter Schneider{-}Kamp and Michael Hanus, editors, {\em
  Proceedings of the 13th International {ACM} {SIGPLAN} Conference on
  Principles and Practice of Declarative Programming, July 20-22, 2011, Odense,
  Denmark}, pages 161--172. {ACM}, 2011.

\bibitem{DBLP:conf/fossacs/ToninhoY18}
Bernardo Toninho and Nobuko Yoshida.
\newblock Depending on session-typed processes.
\newblock In Christel Baier and Ugo~Dal Lago, editors, {\em Foundations of
  Software Science and Computation Structures - 21st International Conference,
  {FOSSACS} 2018, Held as Part of the European Joint Conferences on Theory and
  Practice of Software, {ETAPS} 2018, Thessaloniki, Greece, April 14-20, 2018,
  Proceedings}, volume 10803 of {\em Lecture Notes in Computer Science}, pages
  128--145. Springer, 2018.

\bibitem{DBLP:conf/tacas/TorreMP08}
Salvatore~La Torre, P.~Madhusudan, and Gennaro Parlato.
\newblock Context-bounded analysis of concurrent queue systems.
\newblock In C.~R. Ramakrishnan and Jakob Rehof, editors, {\em Tools and
  Algorithms for the Construction and Analysis of Systems, 14th International
  Conference, {TACAS} 2008, Held as Part of the Joint European Conferences on
  Theory and Practice of Software, {ETAPS} 2008, Budapest, Hungary, March
  29-April 6, 2008. Proceedings}, volume 4963 of {\em Lecture Notes in Computer
  Science}, pages 299--314. Springer, 2008.

\bibitem{DBLP:journals/jfp/Wadler14}
Philip Wadler.
\newblock Propositions as sessions.
\newblock {\em J. Funct. Program.}, 24(2-3):384--418, 2014.

\bibitem{DBLP:conf/tgc/YoshidaHNN13}
Nobuko Yoshida, Raymond Hu, Rumyana Neykova, and Nicholas Ng.
\newblock The scribble protocol language.
\newblock In Mart{\'{\i}}n Abadi and Alberto Lluch{-}Lafuente, editors, {\em
  Trustworthy Global Computing - 8th International Symposium, {TGC} 2013,
  Buenos Aires, Argentina, August 30-31, 2013, Revised Selected Papers}, volume
  8358 of {\em Lecture Notes in Computer Science}, pages 22--41. Springer,
  2013.

\bibitem{springhibernate}
Spring and {H}ibernate {T}ransaction in {J}ava.
\newblock
  \url{https://www.uml-diagrams.org/examples/spring-hibernate-transaction-sequence-diagram-example.html}.

\end{thebibliography}
\end{document}